\documentclass[12pt]{article}

\usepackage{amsmath,graphicx,amssymb,amsthm}
\usepackage{caption,subfig}
\usepackage{float}
\usepackage{color}
\usepackage{xfrac}
\usepackage{enumerate}

\usepackage[margin=1.25in]{geometry}

\usepackage{setspace}
\usepackage[authoryear]{natbib}

\providecommand{\keywords}[1]
{
  \small	
  \textbf{{Keywords---}} #1
}

\providecommand{\JEL}[1]
{
  \small	
  \textbf{{JEL---}} #1
}

\usepackage[hidelinks]{hyperref}
\usepackage[capitalize]{cleveref}
\usepackage{tikz}

\DeclareMathOperator*{\argmax}{arg\,max}

\providecommand{\U}[1]{\protect\rule{.1in}{.1in}}
\newtheorem{theorem}{Theorem}

\newtheorem{definition}{Definition}

\newtheorem{lemma}{Lemma}

\newtheorem{proposition}{Proposition}

\usepackage{ifthen}
\usepackage{color-edits}
\addauthor{yl}{blue}    % yl for Yingkai
\addauthor{jl}{red}

\begin{document}
\newcommand{\setsize}[1]{{\left|#1\right|}}

\newcommand{\floor}[1]{
{\lfloor {#1} \rfloor}
}
\newcommand{\bigfloor}[1]{
{\left\lfloor {#1} \right\rfloor}
}

\newcommand{\super}[1]{^{(#1)}}

\newcommand{\tinteval}[1]{\left[#1\right]}
\newcommand{\distcond}[2]{\left.#1\vphantom{\big|}\right|_{#2}}
\newcommand{\evalat}[2]{\left.#1\vphantom{\big|}\right|_{#2}}

%
% probability stuff.
%
\newcommand{\given}{|}
\newcommand{\wgiven}{\,\mid\,}

% resizing brackets 
\newcommand{\prob}[2][]{\text{\bf Pr}\ifthenelse{\not\equal{}{#1}}{_{#1}}{}\!\left[{\def\givenn{\middle|}#2}\right]}
\newcommand{\expect}[2][]{\text{\bf E}\ifthenelse{\not\equal{}{#1}}{_{#1}}{}\!\left[{\def\givenn{\middle|}#2}\right]}

% brackets configured with \tparen
\newcommand{\tparen}{\big}
\newcommand{\tprob}[2][]{\text{\bf Pr}\ifthenelse{\not\equal{}{#1}}{_{#1}}{}\tparen[{\def\given{\tparen|}#2}\tparen]}
\newcommand{\texpect}[2][]{\text{\bf E}\ifthenelse{\not\equal{}{#1}}{_{#1}}{}\tparen[{\def\given{\tparen|}#2}\tparen]}

% brackets do not resize.
\newcommand{\sprob}[2][]{\text{\bf Pr}\ifthenelse{\not\equal{}{#1}}{_{#1}}{}[#2]}
\newcommand{\sexpect}[2][]{\text{\bf E}\ifthenelse{\not\equal{}{#1}}{_{#1}}{}[#2]}

% brackets
\newcommand{\lbr}[1]{\left\{#1\right\}}
\newcommand{\rbr}[1]{\left(#1\right)}
\newcommand{\cbr}[1]{\left[#1\right]}

\newcommand{\suchthat}{\,:\,}

\newcommand{\partialx}[2][]{{\tfrac{\partial #1}{\partial #2}}}
\newcommand{\nicepartialx}[2][]{{\nicefrac{\partial #1}{\partial #2}}}
\newcommand{\dd}{{\,\mathrm d}}
\newcommand{\ddx}[2][]{{\tfrac{\dd #1}{\dd #2}}}
\newcommand{\niceddx}[2][]{{\nicefrac{\dd #1}{\dd #2}}}
\newcommand{\grad}{\nabla}

\newcommand{\symdiff}{\triangle}
\newcommand{\abs}[1]{\left\lvert{#1}\right\rvert} % Absolute value notation
\newcommand{\set}[1]{\left\{ {#1} \right\} } % Set notation
\newcommand{\indicate}[1]{{\bf 1}\left[#1\right]}
\newcommand{\reals}{\mathbb{R}}
\newcommand{\posreals}{\reals_+}
\newcommand{\supp}{\text{supp}}

\newcommand{\primed}{^{\dagger}}
\newcommand{\good}{G}
\newcommand{\bad}{B}
\newcommand{\no}{N}
\newcommand{\rsignal}{R}
\newcommand{\lsignal}{L}
\newcommand{\rstate}{1}
\newcommand{\lstate}{0}

\newcommand{\inc}{{\rm Inc}}

\newcommand{\state}{\theta}
\newcommand{\states}{\Theta}
\newcommand{\prior}{D}
\newcommand{\cost}{c}
\newcommand{\effort}{e}
\newcommand{\val}{v}

\newcommand{\histories}{\mathcal{H}}
\newcommand{\history}{h}
\newcommand{\mess}{m}
\newcommand{\messages}{M}
\newcommand{\contract}{\mathcal{R}}
\newcommand{\reward}{r}
\newcommand{\rewardMenu}{\mathcal{M}}
\newcommand{\randomDevice}{\varsigma}
\newcommand{\util}{u}
\newcommand{\stoptime}{\tau}
\newcommand{\score}{P}
\newcommand{\expectScore}{U_{\score}}
\newcommand{\posterior}{\mu}
\newcommand{\subgradient}{\xi}

\newcommand{\signal}{s}
\newcommand{\signals}{S}

\newcommand{\cgame}{\mathcal{G}}

\title{Incentivizing Forecasters to Learn:\\ 
\normalsize Summarized vs. Unrestricted Advice\thanks{A preliminary version of this paper has been accepted in the Twenty-Fifth ACM Conference on Economics and Computation (EC'24) as a one-page abstract with the title \emph{Optimal Scoring for Dynamic Information Acquisition}. Part of this work was completed while Jonathan Libgober was visiting Yale University, and Yingkai Li was a postdoc at Yale University. The authors thank Dirk Bergemann, Alex Bloedel, Tilman B\"orgers, Juan Carrillo, Tan Gan, Yingni Guo, Marina Halac, Jason Hartline, David Kempe, Nicolas Lambert, Elliot Lipnowski, Andrew McClellan, Mallesh Pai, Harry Pei, Larry Samuelson, Philipp Strack, Guofu Tan, Kai Hao Yang and numerous seminar audiences for helpful comments and suggestions. Yingkai Li thanks the Sloan Research Fellowship FG-2019-12378 and the NUS Start-up Grant for financial support.}}

\author{Yingkai Li\thanks{Department of Economics, National University of Singapore.
Email: \texttt{yk.li@nus.edu.sg}.}
\and 
Jonathan Libgober\thanks{
Department of Economics, University of Southern California. 
Email: \texttt{libgober@usc.edu}.}
}
% \date{First version: November 12, 2023\\ This version: December 15, 2025 \\
% \href{https://yingkai-li.github.io/files/drafts/Dynamic_Scoring.pdf}{Newest version}
% }
\date{}
\maketitle
\thispagestyle{empty}
\begin{abstract}

\noindent How should forecasters be incentivized to acquire the most information when learning takes place over time? We address this question in the context of a novel dynamic mechanism design problem in which a designer incentivizes an expert to learn by conditioning rewards on an event's outcome and the expert's reports. Eliciting summarized advice at a terminal date maximizes information acquisition if an informative signal either fully reveals the outcome or has predictable content. Otherwise, richer reporting capabilities may be required. Our findings shed light on incentive design for consultation and forecasting by illustrating how learning dynamics shape the qualitative properties of effort-maximizing contracts.
\end{abstract}

\noindent \keywords{Scoring rules, dynamic contracts, dynamic moral hazard, Poisson learning.} 

\noindent \JEL{D82, D83}

% \shortauthor{Li and Libgober}
% \shorttitle{Dynamic Scoring Rules}

% \contact{\textsc{yingkai.li@yale.edu, libgober@usc.edu}}
% \thanknotes{Acknowledgements to be Added.}

\newpage

\setcounter{page}{1}

\section{Introduction}

\subsection{Overview}
This paper asks whether adding complexity to a contract by making rewards sensitive to the timing of reports can incentivize more information acquisition than simpler, time-independent contracts. We ask this question in a minimal model where only a single prediction is sought and there is no exogenous time pressure. We show that the answer depends on whether incentivizing effort at one belief necessarily weakens incentives at another.

Specifically, we study a novel mechanism design problem that extends past work on \emph{scoring rule design}\footnote{This literature, starting with \cite{brier1950verification}, considered how to evaluate (i.e., provide scores for) weather forecasts. We review this work in more depth in Section \ref{sect:literature}.} to a setting featuring dynamic acquisition of costly information. A scoring rule is a contract that provides a reward---over which the agent has fixed value---with some probability depending on (i) the agent's (reported) belief about the likelihood of an event and (ii) that event's outcome. While much past work studying scoring rules is concerned with incentivizing forecasters to \emph{truthfully report} information, we share a focus with a smaller line of work on how to achieve this end \emph{as well as} incentivize its acquisition. 

A leading application of our work---and scoring rule design more generally---is to the practical question of how to best incentivize forecasters. Forecasting takes a plethora of forms, ranging from familiar (e.g., weather or economic outlooks) to predictions more broadly (e.g., national security threat assessments or medical diagnoses). Our interest is in situations where the event over which a forecast is sought is idiosyncratic. Thus, the only way to obtain advice is via asking a forecaster to actively gather information, while the forecaster cannot rely on passive knowledge to make an informed prediction.

The motivation for considering dynamic information acquisition together with forecaster incentive design stems from the observation that active learning takes time. Consider, for example, a decision maker hiring an analyst to conduct due diligence; e.g., before acquiring a company or funding a startup. In this context, the analyst's task is to forecast whether the investment target will meet its promised outcomes, making this determination by analyzing its complex financial records. In applications where this information is extensive and difficult to interpret, it would be up to the analyst to decide \emph{when} to stop looking for a red flag. This timing decision makes the forecaster's problem inherently dynamic.

We focus on the question of whether eliciting a single prediction, made after all information is acquired, induces truthful reporting \emph{and} maximizes incentives for information acquisition. Our results clarify when more complex contracting capabilities are more conducive to informed forecasts than simpler ones. Some relationships---such as when a consultant is hired at arm's length to report at a prespecified date---may limit the scope for reports to be provided at any time. To the extent that closer relationships enable continuous reporting, our main results address whether this capability provides an avenue for such relationships to yield benefits.

\subsection{Model Summary}
 
In our model, an agent (forecaster) chooses privately---over time---whether to exert costly effort to learn about an uncertain future binary event. We refer to the outcome of this event as the \emph{state}, and assume the initial prior over it is commonly known. 
We assume that the private information of the agent, when exerting effort, arrives through a discretized Poisson bandit process. This Poisson learning technology allows us to tractably capture key qualitative features of dynamic information acquisition, enabling sharp comparisons across contract designs despite the complexities stemming from the combination of dynamic learning with a rich contracting space.\footnote{In particular, the agent's continuation problem typically fails to be stationary for general contracts. Much existing work on costly information acquisition relies on stationarity to obtain tractability, requiring us to develop new techniques to analyze contract design in our setting.}

When exerting effort, the agent may (privately) observe two kinds of signals, each (in our main model) with fixed probability: (a) a ``null signal'' or (b) a  ``Poisson signal.'' We interpret a Poisson signal as some particular sought-after information---e.g., a red flag in the due diligence application or an indication of an adversary's capabilities in a national-security context.  In contrast to some related work, a Poisson signal in our model need \emph{not} reveal the outcome or suggest only one outcome.  A red flag might narrow the set of failure scenarios without confirming whether any will materialize; likewise, an intelligence analyst might learn that an adversary has the capability to act on a given date without learning their intentions. This structure is relevant to many expert-forecasting problems.

Past work on forecaster incentives has emphasized that a primary motivator is reputation \citep{marinovic2013forecasters}. If the forecaster is employed by a decision maker, the latter's recommendation has a fixed value insofar as it shapes external perceptions. The aforementioned literature on scoring rule design used this observation to note that the truthful revelation of information could be incentivized by conditioning endorsements on predictions together with realized outcomes. But such schemes can not only incentivize forecasters to truthfully reveal their private information, but also to exert unobservable effort.\footnote{See \cite{zellner2021survey} for a recent survey on forecasting, including a discussion of the relevance of scoring rules for forecaster incentives. \cite{Gneiting2007} provides theoretical background on scoring rules, while \cite{Frongillo2023} provides a more recent survey.} In our model, the designer's contract specifies how to provide such an endorsement, modeled as a reward of fixed value (i.e., a score).
The objective of the designer is to incentivize the agent to acquire the most precise information by exerting maximum effort through the choice of how such endorsements are provided.

To provide sharp guidance on whether dynamic reporting provides the strongest incentives, we keep the contracting environment otherwise identical between single-elicitation and dynamic cases, except for the ability to receive reports from the agent at any time. Given this, the designer's problem in either contracting environment is to decide how to provide the endorsement or recommendation as a function of (a) the forecaster's static or dynamic prediction and (b) the outcome that arises.

Note that a crucial feature of our model is that both effort and learning are private. This reflects our interest in cases where the sought-after signal itself requires expertise to recognize and interpret. 
In particular, the assumption that learning is private separates us from the literature on contracting for experimentation, 
which typically allows contracts to increase rewards in later periods to incentivize continued effort.
This is not the case in our problem: with private learning, the agent can conceal or misreport results, undermining any direct performance-based incentives.

The literature on contract design for forecasters has a much shorter tradition when applied to \emph{dynamic} settings. We share this focus with \cite{DPS2018} who study the design of dynamic contracts to \emph{screen forecaster ability} in the absence of moral hazard. While our models make similar assumptions about the available contracts, the key difference is our focus on dynamic moral hazard rather than initial adverse selection.

% In many contracting-for-experimentation problems, rewards can depend directly on output. In such settings, contracts that increase rewards in later periods can be useful for incentivizing continued effort. This is not the case in our problem: with private learning, the agent can conceal or misreport results, undermining any direct performance-based incentives.

% Our assumption that both effort and learning are private reflects our interest in cases where the sought-after signal itself requires expertise to recognize and interpret. To make the comparison as clear as possible, we keep the contracting environment otherwise identical between single-elicitation and dynamic cases, except for the ability to receive reports from the agent at any time. Given this, the designer's problem in either contracting environment is to decide how to provide the endorsement or recommendation as a function of (a) the forecaster's prediction and (b) the outcome that arises. 

% The literature on contract design for forecasters has a much shorter tradition when applied to \emph{dynamic} settings. We share this focus with \cite{DPS2018} who study the design of dynamic contracts to \emph{screen forecaster ability} in the absence of moral hazard. While our models make similar assumptions about the available contracts, the key difference is our focus on dynamic moral hazard rather than initial adverse selection.

\subsection{Results and Intuition}

\noindent Our interest is in whether maximum effort can be implemented by a contract in which the principal promises to give the forecaster the desired endorsement with a probability depending on (i) a single belief report and (ii) the realized outcome. Such a contract is static if the reward options offered to the agent do not depend on time in any way. The challenge is that the incentives necessary to induce effort depend on the forecaster's current belief, and these belief-contingent incentive requirements generally conflict as beliefs evolve over time. Our first result, \cref{thm:optimal_dynamic_menu}, characterizes effort-maximizing dynamic contracts themselves. Such contracts admit a decreasing no-information reward path: increasing rewards over time instead creates profitable ``shirk-and-lie'' deviations. This characterization sharply reduces the relevant contract space and provides the starting point for our main question, namely whether time variation in rewards is necessary for effort-maximization.

We briefly describe why belief evolution creates a fundamental tension for incentive provision---and why moving beyond static contracts could potentially help resolve this tension. Consider first a one-shot information acquisition problem, and notice that the scoring rule that provides the greatest gain from effort typically depends on the prior. Intuitively, if rewards place too much weight on the state that is currently more likely, an agent may prefer to stop and report a favorable signal rather than exert effort; the agent may similarly lack incentives to learn if rewards place insufficient weight on the initially less-likely state. However, the same logic implies that the reward structure required for continued effort \emph{flips} once the agent becomes highly confident in the \emph{other} state. Thus, rewards maximizing incentives at one belief can sharply reduce incentives at another, suggesting that different rewards should be available at different times (and under different beliefs). While this argument uses extreme beliefs to illustrate, the same issue arises throughout the belief space. We present a more precise illustration of how effort-maximizing rewards should be calibrated to the agent's belief---and that the incentive constraints associated with different beliefs place conflicting requirements on rewards---in \cref{sect:discussions}. 

Now, if beliefs in the dynamic setting were constant before a Poisson signal arrival---i.e., no information is conveyed by the failure to find a red (or green) flag---such adjustments would clearly be unnecessary: the optimization problem is identical at every point in time. We refer to this case as a \emph{Stationary environment} and confirm that indeed static contracts suffice then (\cref{thm:stationary}). The aforementioned tension emerges when beliefs evolve over time. As the agent becomes more confident in one state, should rewards in that state decrease to maintain incentives? The answer turns out to be no in the following cases:  

\begin{enumerate}
    \item If a red (or green) flag reveals the future outcome perfectly: That is, a \emph{Perfect-learning} environment (\cref{thm:perfect}) where signals fully reveal the state; 
    \item If only a red flag can arrive: That is, a \emph{Single-signal} environment (\cref{thm:single_source}) where a Poisson signal arrival always moves beliefs in one direction.
\end{enumerate}

\noindent Outside of these cases, dynamic reports may expand the set of strategies that an agent can be induced to follow. We identify a set of parameters jointly violating these conditions for which richer contracts are necessary (\cref{thm:dynamic optimal}).\footnote{In particular, as we describe in Section \ref{sec:dynamic_optimal}, these parameters essentially amount to a ``sufficient violation'' of the conditions of \cref{thm:perfect,thm:single_source}.}

Our key insight for perfect-learning and single-signal environments is that, despite belief drift, adjusting rewards over time is counterproductive. Lowering rewards in increasingly-likely states might seem to sharpen incentives at later beliefs. But doing so reduces the agent's continuation value from working, weakening earlier incentives. In these environments, any such time-varying contract can be replaced by a static scoring rule that provides weakly stronger incentives at all times.
However, when beliefs drift and signals can push posteriors in either direction without fully revealing the state, conflicts arise that cannot be resolved by static contracts alone. 

Our unifying message is that the benefits from dynamic contracts require informative, but imperfect, Poisson signals that can move beliefs in the same direction as the belief drift absent signal arrival. For \emph{these} signals, adjusting rewards over time can strengthen incentives. In addition, the environment must feature \emph{sufficiently important dynamics}. \cref{apx:two_period} shows that static and dynamic implementations coincide with only two periods---an agent cannot be incentivized to work for two periods by a dynamic contract if this is impossible under a static contract. Our characterization in \cref{thm:stationary,thm:perfect,thm:single_source} is notable because the conditions on the learning environment \emph{do not} depend on the time horizon or effort cost. The necessity of dynamic contracts emerges once the horizon is sufficiently long.

We mention that effort-maximizing contracts can in general be computed numerically via linear programming, as we describe in \cref{sec:structure,sec:numerical}. This numerical procedure yields explicit solutions for effort-maximizing contracts, illustrating concretely how dynamics can influence contract design. The analytical results that we obtain serve a complementary goal: isolating conditions on primitives with a transparent economic interpretation.

Two further points clarify how dynamics influence the design of optimal contracts even in cases where a static implementation exists (\cref{subsec:scoring_under_dynamic}). First, it need not be the case that the effort-maximizing scoring rule provides the strongest incentives at the prior, since incentives must be balanced as the agent's belief changes, as alluded to above. Second, the optimal scoring rule may require offering intermediate options that provide strictly positive rewards even for incorrect predictions---contrasting with results for static cases \citet{li2022optimization,szalay2005}. This distinction reflects that beliefs evolve: an option that is ``intermediate'' early may be ``extreme'' later. While our analysis also speaks to other aspects of scoring rule design, we focus on the problem of static implementation as this is a sharp, qualitative contract characteristic whose practical relevance is transparent.   

\subsection{Related Literature} \label{sect:literature}

Our paper joins a long line of work in economic theory asking how to incentivize information acquisition or experimentation. A key theoretical novelty that arises in such settings is the introduction of \emph{endogenous adverse selection} since different effort choices (typically themselves subject to moral hazard) will provide the agent with different beliefs over the relevant state. This basic interaction, where an agent exerts effort under moral hazard to acquire information, has been analyzed under varying assumptions regarding the underlying information acquisition problem and contracting abilities.\footnote{For instance, static information acquisition technologies where information is acquired after contracting \citep{krahmer2011optimal} or where the outcome of experimentation may be contractable (\citet{yoder_2022}, as well as \cite{chade_kovrijnykh_2016} in a repeated setting).} 

As noted above, much of the literature on scoring rule design focuses exclusively on the \emph{elicitation} of information (see, for instance \cite{Mcc-56,sav-71,lam-11}, as well as \citet{chambers2021dynamic} for the dynamic setting). To the best of our knowledge, \cite{osband1989} is the earliest work sharing our focus on the question of incentivizing the \emph{acquisition} of information. Other work relevant to the application of forecasting is \cite{elliott2016forecasting}, which reviews the statistical properties of forecasting models. Aside from \cite{DPS2018}, other papers focused on the problem of screening forecasters include \citet{DPS2023,Dasgupta2023}. More recent work in economics and computer science related to scoring rules include \citet*{young_turks_yes_men_2022,Z-11,carroll2019robust,li2022optimization,neyman2021binary,li2023multidimensional,whitmeyer2023,chen2021scoringrule,bloedel-segal24}. Our main point of departure from this line of work stems from our focus on dynamics.\footnote{While \citet{neyman2021binary,li2023multidimensional} and \citet{chen2021scoringrule} allow dynamic information acquisition, all explicitly assume contracts must be static. \citet{bloedel-segal24} study dynamic contracts, but their framework models dynamics across different agents, whereas our model focuses on dynamics in which a single agent acquires information over time.}

The Poisson information acquisition technology has been a workhorse for the analysis of how to structure \emph{dynamic} contracts for experimentation.\footnote{\citet{McClellan2022,henry2019} consider related models where information acquisition instead uses a \emph{Brownian motion technology}, and an agent decides when to stop experimenting.} \citet{BH1998,BH2005} were early contributions studying a contracting problem under the assumption that a ``success'' reveals the state. The subsequent literature has considered variations on this basic environment.\footnote{For instance, \citet{Horner2013} relaxes commitment; \citet{HKL2016} allow for ex-ante adverse selection and transfers; \citet{Guo2016} considers delegation without transfers; \citet{Fu2024PrivateExploration} studies optimal approval rules in a private sequential experimentation problem with selectively disclosed hard evidence, comparing environments with conclusive good news and conclusive bad news.} The closest to our work is \cite{Gerardi2012}, who assume a Poisson arrival technology and, as in the scoring rule literature, allow for state-dependent contracts. Our information acquisition technology generalizes this technology to allow for Poisson signals that may support either state and be inconclusive. Our main message is that dynamic contracts can outperform static scoring rules only under \emph{both} modifications.

On this note, Poisson bandits have been extensively utilized in economic settings since the influential work of \citet{KRC2005}. An advantage of this framework is that it facilitates qualitative, economically-substantive properties of information acquisition and predicted behavior; a highly incomplete list of examples includes \citet*{strulovici2010,CM2019,damiano_li_suen2020,keller2015breakdowns,bardhi_guo_strulovici2024,lizzeri_shmaya_yariv}. Our exercise essentially amounts to \emph{designing a (possibly dynamic) single-agent decision problem}. Note that the agent's problem need not admit a simple stationary representation for arbitrary mechanisms in our framework.\footnote{\cite{ball_knoepfle_2024} study monitoring using a Poisson framework; while they allow bidirectional signals, their design problem maintains recursivity, unlike ours.}  This contrasts with most of the settings where payoffs are exogenous, in which case such stationarity may be crucial for tractability. Partially for this reason, our approach does not require determining the agent's exact best response following an arbitrary dynamic contract.

\section{Model}
Our model considers an agent (who, depending on the application, may be an individual expert or a team working as a single entity) as a forecaster. A mechanism designer shares a common prior with the agent over an uncertain event $\theta \in \Theta=\{0,1\}$ (e.g., whether an investment will achieve a certain target outcome, whether an adversary will attack on a certain date, etc.); we let $\prior$ denote the initial probability that $\theta=1$.

The agent is able to acquire information about $\theta$ at discrete times $\{1, 2, \dots, T\}$.\footnote{All of our results extend if we consider a variable time interval $\Delta$ and take the high-frequency limit as $\Delta \rightarrow 0$.} We refer to $\theta$ as the unknown \emph{state}. This state is contractible and is only realized after time $T$ (e.g., a successful prototype or an attempted attack from the adversary). 
For conceptual simplicity, we take $T < \infty$, and interpret $T$ as the deadline for the designer to make a payoff relevant decision regarding the unknown state (e.g., whether to invest in the company or attack the adversary). 
The detailed decision environment plays no role in our analysis; we assume only that the designer weakly prefers more information (a la Blackwell).

\subsection{Information Acquisition}  \label{sec:infoacq}

At any time $t\leq T$, the agent can acquire information by paying a cost $c > 0$. When this cost is paid, a piece of (falsifiable) evidence informative of $\theta$ may be discovered by the agent with some probability. We refer to this evidence as a \emph{Poisson signal}. We take the arrival probability of the Poisson signal to be $\lambda_{\theta}$. Without loss of generality, we assume that $\lambda_{1} \geq \lambda_{0}$, so the agent's belief drifts toward state $0$ as no Poisson signal arrives. 

If no Poisson signal arrives, the agent observes a \emph{null signal}, which we denote by $\no$. If the agent does not exert effort, then a null signal is observed with a probability of 1. 
When the Poisson signal arrives, various pieces of information may potentially be observed by the agent, leading to different beliefs. 
Specifically, we denote the set of \emph{non-null} signals as $S$, and the agent observes a signal $s\in S$ with probability $\lambda_{\state}^s$ by exerting effort, 
where $\sum_{s\in S}\lambda_{\state}^s = \lambda_{\state}$ for all states $\state\in\{0,1\}$. 
By Bayes' rule, at any time $t\in\{1,\dots,T\}$, assuming that the agent has exerted effort for all periods until $t$,
we let $\posterior^{\no}_t$ denote the posterior belief that $\theta=1$ if no Poisson signal arrived before time $t$ (including~$t$),
and let $\posterior^{s}_t$ denote the agent's posterior when receiving Poisson signals $s$ exactly at time $t$: 
\begin{equation*}
\posterior^{\no}_t = \frac{\posterior^{\no}_{t-1} (1-\lambda_{1})}{\posterior^{\no}_{t-1}(1-\lambda_{1})+(1-\posterior^{\no}_{t-1})(1-\lambda_{0})},\end{equation*}
\begin{equation*}
\posterior^{\signal}_{t} = \frac{\posterior^{\no}_{t-1} \lambda_{1}^{\signal}}{\posterior^{\no}_{t-1}\lambda_{1}^{\signal}+(1-\posterior^{\no}_{t-1})\lambda_{0}^{\signal}},
\end{equation*}
where $\posterior^{\no}_0 = \prior$ is the prior belief. 

Once the Poisson signal arrives, no further information can be acquired (e.g., if a red flag reveals the main point of concern so that additional scrutiny will not meaningfully alter the assessment of the investment's viability; if an adversary's capabilities are determined, no further information is relevant for assessing attack probability, etc.).\footnote{The assumption of a single-signal arrival represents an extreme case of information attrition in \cite{Strulovici2022}, which discusses various compelling practical instances where the number of available signals is naturally limited. In our case, this feature enhances tractability by allowing us to associate the amount of information produced with the length of time the agent works absent the arrival of a Poisson signal. It also avoids known complications in determining the agent's payoffs under Poisson learning when signals are not perfectly revealing.} 
Note that our information acquisition technology generalizes Poisson bandit learning (as in \cite{HKL2016}, for instance), because signals (a) need not reveal the state and (b) can increase the posterior probability of either state. 

In addition, our main results on static contracts implementing maximal effort focus on the following canonical cases of this model: 

\begin{enumerate}
\item \textbf{Stationary environments,} where
$\lambda_{1}= \lambda_{0}$. 

\item \textbf{Perfect-learning environments,} where either $\lambda^s_0 = 0$ or $\lambda^s_1=0$ for every $s \in S$. 

\item \textbf{Single-signal environments,} where $\abs{S}=1$. 
\end{enumerate}
Intuitively, in stationary environments, the mere passage of time conveys no information: silence (null signals) leaves beliefs unchanged, so only the content of an arriving non-null signal can move the posterior. In perfect-learning environments, every non-null signal is state-exclusive. Once any such signal arrives, the state is fully revealed. In single-signal environments, there is only one kind of evidence; thus, when the signal arrives, it shifts beliefs in a single direction, toward the state that more readily generates that signal. If beliefs drift toward $0$ when the Poisson signal does not arrive, they \emph{must} jump toward $1$ whenever it does in the single-signal environment. 

To reiterate, effort choices and signal arrivals are private information of the agent. The assumption that this signal is private reflects the notion that it requires expertise to recognize and interpret. Moreover, the signals are not verifiable and, hence, can be arbitrarily fabricated or misrepresented by the agent at no cost. The assumption that effort is private reflects the notion that the principal cannot easily monitor the agent's activity.

\subsection{Contracting} \label{sec:model:contracting}

We will focus on a menu representation of dynamic contracts for incentivizing the agent to exert costly effort.\footnote{A more general format based on arbitrary communication and its equivalence to our menu representation is provided in \cref{subapx:contract_definition}. } 
Specifically, the designer offers the agent a menu of rewards
\begin{align*}
\rewardMenu:=\{\rewardMenu_{t} \subseteq [0,1]^2\}_{t=1}^T.
\end{align*} 
At any time $t\leq T$, if the agent picks a reward profile $\reward_t=(\reward_{t,0},\reward_{t,1})\in \rewardMenu_t$, the agent receives a reward of $\reward_{t,\state}$ after time $T$ when the state $\state$ is publicly revealed.
As mentioned in the introduction, we can also interpret the reward as the probability of receiving an endorsement that positively influences the forecaster's reputation. Such reward bounds are also widely assumed in the literature on evaluating forecasters \citep[e.g.,][]{DPS2018}, as well as in the literature on information elicitation more generally.\footnote{While we interpret the reward as non-monetary, our model also applies to applications where the designer faces explicit budget constraints set by a third party. \citet{anthony2007management} discusses firms imposing budget constraints on different divisions; National meteorological agencies, such as the U.S.~National Weather Service (NWS) receive public funding for operations; the Congressional Budget Office (CBO) operates under a fixed annual budget and provides forecasts for Congress. In such cases, designers are restricted to using this budget to incentivize information acquisition. } 
Note that a crucial feature of this menu representation is that the agent only makes one irrevocable choice for the reward vectors. 
This is without loss of generality since we have assumed that the Poisson signals can arrive at most once over the entire time horizon. 

We assume that the agent does not discount the future.\footnote{We discuss the extensions with discounted utilities in \cref{subsec:non-invariant}.} That is, by exerting effort for $Z$ periods and receiving a reward of $\reward\in[0,1]$ in the end, the utility of the agent is 
\begin{align*}
\reward - cZ. 
\end{align*}

Given a menu $\rewardMenu$, let $Z_{\rewardMenu}$ denote the number of periods the agent would exert effort in the absence of a Poisson signal.\footnote{We do not obtain a closed form for $Z_{\rewardMenu}$, as the best response strategy of the agent can potentially be complex given an arbitrary menu $\rewardMenu$. We will provide a more transparent description in \cref{sub:stopping} when we simplify the best response strategies of the agent.} 
In our Poisson learning environment, the induced experiment (the distribution over posteriors) is uniquely determined by $Z_{\rewardMenu}$. Moreover, a larger $Z_{\rewardMenu}$ yields Blackwell more informative signals. 
Therefore, the designer's problem is
\[
  \rewardMenu^* \in \arg\max_{\rewardMenu} Z_{\rewardMenu}.
\]

\noindent This objective reflects the designer's goal of providing the strongest possible incentives for learning.

Contracts must address both hidden effort and hidden information: the principal does not observe whether effort was exerted, when a Poisson signal arrived, or which signal was observed. These frictions play distinct but interacting roles. Hidden effort creates the need to provide incentives for information acquisition. Private and nonverifiable learning constrains how those incentives can be provided, because rewards designed to encourage effort must also preserve the agent's incentives to select the intended report. Private signal arrival further allows the agent to conceal when information was obtained, making the timing of available rewards strategically relevant. The comparison between static and dynamic elicitation therefore asks whether time-contingent menus can strengthen effort incentives while maintaining reporting incentives as beliefs evolve.  Section \ref{sub:bothincentives} outlines the resulting incentive compatibility constraints that arise given the class of contracts we consider. 

\subsection{Static Implementation}
Since the agent's signal is private, the designer cannot compel early selection: at any date $t$, the agent can guarantee himself access to any reward that will appear in the future simply by waiting. This observation immediately implies the following Lemma:

\begin{lemma}[Shrinking Menus are Without Loss] \label{lem:shrinking_menus}
For any menu profile $\rewardMenu=(\rewardMenu_t)_{t=1}^T$, define
\[
\widetilde \rewardMenu_t=\bigcup_{t'\geq t}\rewardMenu_{t'}.
\]
Then $\widetilde{\rewardMenu}$ is weakly shrinking and induces the
same feasible reward choices, and hence the same agent behavior, as
$\rewardMenu$.
\end{lemma}

The weakly shrinking property follows immediately from the definition of $\widetilde \rewardMenu_{t}$, i.e.:

 \begin{align}\label{eq:nested_menu}
\widetilde{\rewardMenu}_t \supseteq \widetilde{\rewardMenu}_{t'} \qquad \text{for all } 1\le t\le t'\le T.
\end{align} 

\noindent In light of this nested (weakly shrinking) structure of dynamic contracts, we now ask whether time variation is necessary at all, or if a single, time-invariant menu maximizes effort. The next definition formalizes this possibility, which we call a \emph{static implementation}.

\begin{definition}[Static Implementation]
\label{def:static_implementation}\,\\
An optimal menu $\rewardMenu$ has a static implementation if there exists $\widehat{\rewardMenu}$ such that $Z_{\widehat{\rewardMenu}} = Z_{\rewardMenu}$
and 
\begin{align*}
\widehat{\rewardMenu}_{t} = \widehat{\rewardMenu}_{t'}, \qquad
\text{ for any } 1\leq t\leq t'\leq T.
\end{align*}
\end{definition}

\noindent Under a static implementation, the designer does not need to monitor the exact time at which the agent observes the Poisson signal to incentivize maximum effort. Given a menu with a static implementation, the agent can defer the choice to time $T$ and select the reward based on the aggregated information acquired over the entire process. Under a static implementation, advice is summarized at the end of the interaction.

Using terminology from the information elicitation literature, a static implementation essentially offers a \emph{scoring rule} at time $T$: Specifically, a scoring rule 
\begin{align*}
\score:\Delta(\states)\times\states\to[0,1]
\end{align*} 
maps the agent's report---which is required to take the form of a posterior belief---and the realized state to a reward. Scoring rules are far simpler than arbitrary dynamic menus: they require only a summary (in the form of a belief report) rather than real-time monitoring and time-varying rewards. In contrast, richer dynamic contracts may require closer integration within the decision-maker's organization. Our goal is to characterize whether unrestricted dynamic contracting is necessary to incentivize maximum information acquisition:

\medskip

\noindent \textbf{Main Question:} \emph{Does the effort-maximizing menu have a static implementation?}

\section{Structures of Dynamic Incentives}
\label{sec:structure}

Before answering our main question, we provide some useful results that help us formulate the approach we take. \cref{sect:discussions} provides some additional discussion and intuition. 

\subsection{Simplifying Effort Strategies to Stopping Strategies}
\label{sub:stopping}
The agent's information acquisition strategy can be arbitrary and complex in dynamic environments. 
But in our model,  
it is without loss of generality for the agent to front-load all effort and adopt a \emph{stopping strategy}---i.e., a strategy in which the agent simply decides \emph{when} to stop working in the absence of a Poisson signal arrival (exerting effort until then).  

More precisely, let $\sigma$ denote the (random) time of the first Poisson signal (with $\sigma=\infty$ if no such signal arrives).
A stopping strategy with stopping time $\stoptime\leq T$ prescribes effort at 
any time period $t\leq\min\{\stoptime,\sigma\}$ and no effort thereafter: effort ceases after the first occurrence of a Poisson signal or $\{t=\stoptime\}$.
Slightly abusing notation, we also use $\stoptime$ to represent the stopping strategy with a stopping time $\stoptime$.
Let $\stoptime_{\rewardMenu}$ be the stopping strategy implemented under the menu $\rewardMenu$.\footnote{We break ties by maximizing the stopping time if there are multiple stopping strategies.}

\begin{lemma}[Stopping Strategies are Without Loss]
\label{lem:agent_effort_response}\,\\
Given any menu of rewards $\rewardMenu$ and any best response of the agent with maximum effort length $Z_{\rewardMenu}$ conditional on not receiving any Poisson signal,\footnote{If the agent randomizes, we let $Z_{\rewardMenu}$ denote the maximum effort length among all realizations. } 
there exists a stopping strategy~$\stoptime_{\rewardMenu}$ with $\stoptime_{\rewardMenu}\geq Z_{\rewardMenu}$
that is also optimal for the agent. 
\end{lemma}

The proof of \cref{lem:agent_effort_response} is simple. Since the agent only learns about the state when exerting effort, and since effort can always be concealed until a later date, any prescribed effort strategy can be modified to one where effort is ``front-loaded,'' with all effort choices moved earlier without increasing its cost or reducing the feasible reporting options. A best response can thus be chosen to exert effort consecutively until either a Poisson signal arrives or the relevant stopping time is reached.

Given \cref{lem:agent_effort_response}, the designer's problem reduces to
\begin{align*}
\rewardMenu^* \in \arg\max_{\rewardMenu} \stoptime_{\rewardMenu}.
\end{align*}

\subsection{The Incentive Structure of Dynamic Contracts}
\label{sub:bothincentives}

\subsubsection{Information Revelation Incentives}
\label{subsec:optimal_menu}
Since it is without loss of generality to consider menus with a shrinking choice set (see \Cref{eq:nested_menu}), 
it is in the agent's best interest to select a reward from the available menu immediately once he hits his stopping time in a stopping strategy, 
i.e., when he receives a Poisson signal or decides not to exert effort anymore. 

Given any menu, the designer can infer the posterior belief of the agent based on his choice of rewards. 
Moreover, only rewards that maximize the on-path beliefs can be chosen. 
Specifically, at any time $t$, for any Poisson signal $s\in S$, 
let $r^s_t\in\rewardMenu_t$ be the reward vector chosen by the agent at time $t$ when his belief is $\posterior^s_t$. 
Moreover, let $\reward^N_{\stoptime_{\rewardMenu}}$ 
be the reward vector chosen by the agent at stopping time $\stoptime_{\rewardMenu}$ when his belief is $\posterior^N_{\stoptime_{\rewardMenu}}$.
Any rewards beyond the collection of 
$\{\reward^{\signal}_t\}_{t\leq \stoptime_{\rewardMenu},\signal\in\signals}\cup\{\reward^{\no}_{\stoptime_{\rewardMenu}}\}$
will never be chosen by the agent. 
Let $\util(\posterior, \reward) \triangleq \expect[\state\sim\posterior]{\reward(\state)}$ denote the agent's expected utility with belief $\mu$ under the reward function $\reward$. 

\begin{lemma}[Simplified Menu Representation]\label{lem:menu}
\,\\
Given any menu $\rewardMenu$ implementing stopping time $\stoptime_{\rewardMenu}$, 
let $\{\reward^{\signal}_t\}_{t\leq \stoptime_{\rewardMenu},\signal\in\signals}\cup\{\reward^{\no}_{\stoptime_{\rewardMenu}}\}$ 
be the collection of rewards chosen by the agent at some on-path history. 
At any $t \leq \stoptime_{\rewardMenu}$, it is equivalent to offer a menu 
$\rewardMenu_t\triangleq\{\reward^{\signal}_{t'}\}_{t\leq t'\leq \stoptime_{\rewardMenu},\signal\in\signals}\cup\{\reward^{\no}_{\stoptime_{\rewardMenu}}\}$ to the agent such that  
\begin{align}
\reward^{\signal}_{t} \in \argmax_{\reward \in \rewardMenu_t} \util(\posterior_t^{\signal}, \reward), 
\qquad \forall \signal\in \signals. \tag{IC-Reporting}\label{constraint:IC}
\end{align}
\end{lemma}
\noindent \noindent In essence, \cref{lem:menu} is the relevant version of the taxation principle for our environment.

For any time $t\leq \stoptime_{\rewardMenu}$ and any Poisson signal $\signal\in\signals$, let
$\util_t^\signal(\rewardMenu)\triangleq \util(\posterior_t^\signal,\reward_t^\signal)$
be the utility of the agent when his belief is $\posterior_t^\signal$.
Similarly, for any $t\leq \stoptime_{\rewardMenu}$, let
\[
\reward_t^\no \in \argmax_{\reward\in \rewardMenu_t}\util(\posterior_t^\no,\reward),
\qquad
\util_t^\no(\rewardMenu)\triangleq \util(\posterior_t^\no,\reward_t^\no)
\]
with the convention that $\posterior_0^\no$ is the prior $\prior$ and $\rewardMenu_0 = \rewardMenu_1$.
We omit the dependence on~$\rewardMenu$ whenever there is no risk of confusion.

\subsubsection{Effort Incentives}

We introduce a notion that allows us to describe the incentives for exerting effort at a given time. Specifically, we decompose the dynamic problem into a sequence of static problems. Each static problem, in turn, considers the relevant continuation incentives of the agent under the dynamic contract.  

Each static game in our decomposition is indexed by a pair of times, $t\leq t' \leq T$, as well as a menu of rewards $\rewardMenu$. We define the \emph{continuation game between $t$ and $t'$}, denoted $\cgame^{\rewardMenu}_{t,t'}$, as the \emph{static} decision problem with prior belief~$\posterior^{\no}_{t-1}$ where the agent makes a one-time binary effort choice at time $t$: (a) Exert no effort or (b) Exert effort up to and including time~$t'$ or until a Poisson signal arrives. We let $C(t,t',c)$ denote the cost of option (b) in the continuation game between $t$ and $t'$ when the cost of effort in each period is $c$. Letting $f^s_t(t')$ be the probability of receiving Poisson signal $s$ at time $t'$ conditional on not receiving Poisson signals before time $t$, and letting $F^s_t(t')$ be the corresponding cumulative probability: 
\begin{align*}
C(t,t',c)=\rbr{1-\sum_{\signal\in\signals}F^{\signal}_t(t')}
\cdot c \cdot(t' - t+1) 
+ \sum_{\hat{t}=t}^{t'} \sum_{\signal\in\signals} f^{\signal}_t(\hat{t}) \cdot c \cdot(\hat{t} - t+1).
\end{align*}
\noindent Indeed, this formula is simply the expected total cost of exerting effort from time $t$ until the earlier of $t'$ or the first Poisson signal arrival. 

An important special case is when the time $t'$ is itself the stopping time for the menu $\rewardMenu$, with $\stoptime_{\rewardMenu} \leq T$. In this case,
we say $\cgame^{\rewardMenu}_{t,\stoptime_{\rewardMenu}}$ is the continuation game at time~$t$ for menu $\rewardMenu$. 
In what follows, we omit the superscript of $\rewardMenu$ when it is clear from context. Note that given any menu of rewards $\rewardMenu$ with stopping time $\stoptime_{\rewardMenu}$, at any time $t\leq T$, it is straightforward to show that the following are equivalent:
\begin{enumerate}
\item $t \leq \stoptime_{\rewardMenu}$, i.e., the agent has an incentive to exert effort at time $t$ in this dynamic environment;
\item there exists $t' \geq t$ such that the agent has an incentive to exert effort in the continuation game $\cgame_{t,t'}$ (corresponding to some feasible continuation under $\rewardMenu$).
\end{enumerate}
When $t \leq \stoptime_{\rewardMenu}$, a feasible choice of $t' \geq t$ that ensures sufficient effort incentives in the continuation game is $t' = \stoptime_{\rewardMenu}$, as this is consistent with the agent playing a best response.

The bottom line is that to design a menu of rewards that implements any $\stoptime\leq T$, 
it suffices to ensure that the agent has an incentive to exert effort in continuation games $\cgame_{t,\stoptime}$ for all $t \leq \stoptime$. Now, fixing $c$ and the desired stopping time $\stoptime$, there typically will exist multiple effort-maximizing contracts, as incentives could be slack given the cost $c$. 
An alternative is to consider an auxiliary problem of \emph{maximizing effort incentives}. 
We describe formally what maximizing effort incentives means: given any continuation game $\cgame_{t,t'}$, 
let $\Delta(\cgame_{t,t'})$ be the difference in expected reward between exerting effort and not exerting effort. 
That is,  
\begin{align*}
\Delta(\cgame_{t,t'})=\rbr{1-\sum_{\signal\in\signals}F^{\signal}_t(t')}
\cdot u(\posterior^{\no}_{t'}, \reward^{\no}_{t'})
+ \sum_{\hat{t}=t}^{t'} \sum_{\signal\in\signals} f^{\signal}_t(\hat{t}) \cdot u(\posterior^{\signal}_{\hat{t}}, \reward^{\signal}_{\hat{t}})
- u(\posterior^{\no}_{t-1}, \reward^{\no}_{t-1}).
\end{align*}
The last term $u(\posterior^{\no}_{t-1}, \reward^{\no}_{t-1})$ is the agent's utility for not exerting effort in the continuation game $\cgame_{t,t'}$.
The index in this term is $t-1$ since the prior belief in this continuation game is $\posterior^{\no}_{t-1}$, 
and if the agent decides not to exert effort from $t$ onward, 
the agent could actually make the report at time $t-1$ to get his favorite reward option $\reward^{\no}_{t-1}$. 
Thus, the agent has an incentive to exert effort in the continuation game $\cgame$ if and only if $\Delta(\cgame_{t,t'}) \geq C(t,t',c)$.

By the stopping-strategy reduction, the effort component of implementing $\stoptime_{\rewardMenu}$ is summarized by the following constraints:\begin{align}
\Delta(\cgame_{t, \stoptime_{\rewardMenu}}) \geq C(t,\stoptime_{\rewardMenu},c) \qquad \forall t \leq \stoptime_{\rewardMenu}. \tag{IC-Effort}\label{constraint:effort_IC}
\end{align}

\noindent This reduces the relevant effort constraints to \eqref{constraint:effort_IC}. Moreover, effort incentives are stronger when the difference in expected rewards between exerting and not exerting effort is larger.

\subsection{The Structure of Effort-Maximizing Dynamic Contracts}

We can now present a characterization of effort-maximizing contracts that arise in light of the simplifications above. To summarize, constraints \eqref{constraint:IC} and \eqref{constraint:effort_IC} capture the two dimensions of the agent's incentives: The former requires the agent to select the intended reward given the information acquired, while the latter requires him to prefer the prescribed continuation of experimentation to stopping. Notice that these constraints \emph{also} accommodate joint deviations of effort and reporting. If the agent stops exerting effort after the no-information history at time $t-1$, he may select whichever reward in $\rewardMenu_{t-1}$ is most attractive at belief $\posterior_{t-1}^{\no}$. By definition, his payoff from the best such report is $\util(\posterior_{t-1}^{\no}, \reward_{t-1}^{\no})$, which is the outside option subtracted in $\Delta(\cgame_{t, \stoptime_{\rewardMenu}})$. Thus, IC-Effort compares the prescribed continuation with stopping followed by an optimal report, rather than only with truthful stopping. Together with the stopping-strategy reduction, imposing these constraints at every no-information history accounts for deviations involving fewer periods of effort, while IC-Reporting governs reporting deviations following a Poisson signal.

For our characterization, it is useful to classify signals according to whether they move beliefs to the left or right of the no-information posterior.
We say that a signal $\signal$ is \emph{left-biased} if $\posterior_t^\signal\leq \posterior_t^\no$ for all $t$, and \emph{right-biased} otherwise.
Given any time $t\leq \stoptime_{\rewardMenu}$, define
\begin{equation}\label{eq:minorant}
G_t(\posterior)
\triangleq
\sup_{\reward:\states\to[0,1]}
\left\{
\util(\posterior,\reward):
\util(\posterior_{t'}^\no,\reward)\leq \util_{t'}^\no,\ \forall t'\in\{0,\dots,t\}
\right\}.
\end{equation}
Since each reward induces an affine function of $\posterior$, $G_t$ is the greatest \emph{convex minorant} of the no-information points $\{(\posterior_{t'}^\no,\util_{t'}^\no)\}_{t'=0}^t$ generated by feasible reward vectors bounded in $[0,1]^2$.

\begin{theorem}[Effort-Maximizing Dynamic Contracts]\label{thm:optimal_dynamic_menu}\,\\
For any prior $\prior$ and any signal arrival rates $\lambda$,
there exists an effort-maximizing menu $\rewardMenu$ with stopping time $\stoptime_{\rewardMenu}$ and an associated reward sequence
$\{\reward_t^\no,\reward_t^\signal\}_{t\leq \stoptime_{\rewardMenu},\signal\in\signals}$
such that:\footnote{Even though our menu representation (\cref{lem:menu}) does not require the specification of $\reward_t^\no$ for $t< \stoptime_{\rewardMenu}$, we include them here to highlight the structure of the optimal contract.}
\begin{enumerate}
\item \textbf{Decreasing no-information rewards:} for all $1\leq t\leq t'\leq \stoptime_{\rewardMenu}$ and $\state\in\states$,
\[
\reward_{t',\state}^\no \leq \reward_{t,\state}^\no.
\]
Moreover, $\reward_{t,0}^\no=1$ whenever $\reward_{t,1}^\no>0$. Equivalently, the sequence $\reward_t^\no$ first decreases in the reward for state $1$, and once that reward reaches $0$, it decreases in the reward for state $0$.

\item \textbf{Left-biased signals:} if $\signal$ is left-biased, then one can choose
\[
\reward_t^\signal=\reward_t^\no
\qquad \text{for all } t\leq \stoptime_{\rewardMenu}.
\]

\item \textbf{Right-biased signals:} if $\signal$ is right-biased, then letting $G_t$ be defined by \eqref{eq:minorant},
\[
\util(\posterior_t^\signal,\reward_t^\signal)=G_t(\posterior_t^\signal),
\]
and the affine function $\posterior\mapsto \util(\posterior,\reward_t^\signal)$ is a supporting line to $G_t$ at $\posterior_t^\signal$.
\end{enumerate}
\end{theorem}

\noindent
Compared to the menu representation in \cref{lem:menu}, \cref{thm:optimal_dynamic_menu} reduces the relevant structure to the path of no-information rewards $\{\reward_t^\no\}_{t\leq \stoptime_{\rewardMenu}}$.
Part 1 gives a one-switch structure: the reward first decreases in state $1$, and once that reward reaches $0$, it decreases in state $0$.
Parts 2 and 3 then pin down all signal-contingent rewards from that path.
Because a reward vector has only two coordinates, any right-biased reward is pinned down by at most two binding no-information constraints, together, possibly, with one boundary constraint from the box $[0,1]^2$.

A useful way to understand Parts 2 and 3 is to fix the no-information path
$\{\reward_t^\no\}_{t\leq \stoptime_{\rewardMenu}}$ from Part 1.
Because menus are shrinking, any reward selected at date $t$ must also be feasible for every earlier no-information type.
Thus, conditional on this path, the reward assigned after signal $\signal$ at date $t$ solves the one-shot problem
\[
\max_{\reward:\states\to [0,1]} \util(\posterior_t^\signal,\reward)
\qquad
\text{s.t.}\qquad
\util(\posterior_{t'}^\no,\reward)\leq \util_{t'}^\no
\quad \forall t'\in\{0,\dots,t\}.
\]
This auxiliary problem chooses the reward that is best for the signal type while not raising the stopping utility of any earlier no-information type.
When $\signal$ is left-biased, the canonical no-information reward $\reward_t^\no$ solves this problem.
When $\signal$ is right-biased, the value of the problem is exactly $G_t(\posterior_t^\signal)$, so the optimal reward is a supporting line to the convex minorant.

When we later specialize to the binary-signal case $\signals=\{\lsignal,\rsignal\}$, with $\lsignal$ left-biased and $\rsignal$ right-biased, \cref{thm:optimal_dynamic_menu} implies that
\[
\reward_t^\lsignal=\reward_t^\no
\qquad \text{for all } t\leq \stoptime_{\rewardMenu},
\]
while $\reward_t^\rsignal$ is given by the supporting line to $G_t$ at $\posterior_t^\rsignal$.

The intuition for the decreasing reward structure in Part 1 is the same as in the introduction:
if the agent expects weakly higher rewards later, he may shirk immediately and imitate a later report or defer reporting an early signal in order to enjoy the later reward.
Decreasing the no-information reward can instead strengthen incentives.
After a null signal, the posterior drifts toward state $0$, making rewards tilted toward state $0$ more attractive and reducing the value of additional information.
Lowering those rewards over time restores dispersion in continuation payoffs and thereby encourages continued effort.

\subsection{Discussions and Intuitions} \label{sect:discussions}

\paragraph{An auxiliary perspective on effort incentives} The design of the optimal menu can be reduced to the problem of maximizing the effort incentive for all continuation games before the stopping time. 
More concretely, for any stopping time $\stoptime \leq T$, we can solve the following linear program:
\begin{align}
c_{\stoptime} := \max_{\rewardMenu,\tilde{c}} \quad & \quad \tilde{c} \tag{Effort Maximization}\label{program:effort_max}\\
\text{s.t.}\quad & \Delta(\cgame^{\rewardMenu}_{t, \stoptime}) \geq C(t,\stoptime,\tilde{c}), 
\quad\forall t\leq \stoptime. \nonumber\\
&\rewardMenu \text{ satisfies \eqref{constraint:IC}.}\nonumber
\end{align}
Because $C(t, \stoptime, c)$ is increasing in the per-period cost $c$, the set of costs at which a fixed stopping time $\stoptime$ is implementable is downward closed (i.e., the same stopping time is implementable even when costs are lower). Consequently, $c_{\stoptime}$ is the \emph{cutoff cost} for implementing $\stoptime$. For any primitive effort cost $c$, the designer's original problem is therefore equivalently:

\[ 
\stoptime^{*}(c) = \max \{\stoptime \leq T : c_{\stoptime} \geq c\}. 
\]

To find the optimal menu, we can enumerate all $\stoptime\leq T$ and identify the largest $\stoptime^*$ such that $c_{\stoptime^*} \geq c$. 
The reward menu implementing $\stoptime^*$ is an optimal menu.

\paragraph{Optimal menu for continuation games}
We first consider the design of the optimal menu that maximizes the reward difference in the continuation game given a pair of time $t\leq \stoptime$. 
In such static environments, this is essentially the optimization of the scoring rules, which is fully characterized in \citet{li2022optimization}. 

\begin{figure}[t]
\centering
\begin{tikzpicture}[xscale=1,yscale=0.8]
\hspace{-10pt}
\draw [<->] (0,5.3) -- (0,0) -- (5.2,0);
\draw [dotted] (0, 3) -- (5, 3) -- (5, 0);

\draw [thick] (0, 5) -- (3.125, 1.875) -- (5, 3);
\draw [dashed] (0, 0) -- (3.125, 1.875) -- (5, 0);

\draw (-0.3, 5) node {$r_0$};

\draw (0, -0.4) node {$0$};
\draw (5, -0.4) node {$1$};

\draw (5.5, 0) node {$\mu$};
\draw (0, 5.6) node {$\expectScore(\posterior)$};

% \draw [red] (0,5) -- (5,0);
% \draw [red] (0,0) -- (5,3);

\draw (-0.3, 3) node {$r_1$};

% \draw [dotted] (1,0) -- (1,4);
% \draw (1, -0.4) node {$\underline{\mu}(r_1)$};
% \draw [dotted] (4.5,0) -- (4.5,2.7);
% \draw (4.2, -0.4) node {$\bar{\mu}(r_1)$};

% \draw (5.55, 0) node {$\posterior^{\no}_t$};

\end{tikzpicture}
\caption{\footnotesize Expected score $\expectScore(\posterior)$ of a V-shaped scoring rule $\score$ with parameters $r_0,r_1$. }
\label{fig:v-shaped}
\end{figure}

\begin{definition}[V-shaped Scoring Rules]
\,\\
A scoring rule $\score$ is V-shaped with parameters $r_0,r_1$ if 
\begin{align*}
\score(\posterior,\state) = 
\begin{cases}
r_1 & \posterior \geq \frac{r_0}{r_1+r_0} \text{ and } \state = 1\\
r_0 & \posterior < \frac{r_0}{r_1+r_0} \text{ and } \state = 0\\
0 & \text{otherwise}.
\end{cases}
\end{align*}
We say $\score$ is a V-shaped scoring rule with a kink at $\hat{\mu}\in[0,1]$ 
if the parameters $r_0,r_1$ satisfy
$r_0 = 1, r_1 = \frac{1-\hat{\mu}}{\hat{\mu}}$ if $\hat{\mu}\geq \frac{1}{2}$
and $r_0 = \frac{\hat{\mu}}{1-\hat{\mu}}, r_1 = 1$ if $\hat{\mu} < \frac{1}{2}$.
\end{definition}
The terminology of the scoring rule as ``V-shaped'' comes from the property that the expected score $\expectScore(\posterior) \triangleq \expect[\state\sim\posterior]{\score(\posterior,\state)}$ is a V-shaped function,
which is illustrated in \Cref{fig:v-shaped}.
Furthermore, given any V-shaped scoring rule $\score$ with a kink at $\prior$, 
the agent with prior belief $\prior$ is indifferent between reward options $(r_0,0)$ and $(0,r_1)$. 
% guessing the state is either~0 or~1. 

\begin{proposition}[\citealp{li2022optimization}]\label{prop:static_moral_hazard}
For any $t\leq \stoptime$, the reward menu offering the V-shaped scoring rule with kink at $\mu_{t-1}^{\no}$ in all periods maximizes the reward difference in continuation game $\cgame_{t,\stoptime}$.
\end{proposition}

\noindent Intuitively, V-shaped scoring rules maximize the expected reward at all posteriors, subject to (a) the constraint that the indirect utility is convex (as a consequence of incentive compatibility) and (b) the expected reward at the prior being constant (so that incentives for the agent to exert effort are maximized). For any fixed information structure, adding curvature to the indirect utility in \Cref{fig:v-shaped} would only decrease the agent's expected utility under that information structure, diminishing effort incentives. And moving the kink increases the payoff from not exerting effort by more than the payoff from exerting effort. 

\paragraph{Dynamic effort incentives}
\cref{prop:static_moral_hazard} illustrates the tensions involved in designing the optimal menu in dynamic environments. As the agent's posterior evolves over time, the priors of the continuation games at different time periods vary, leading to inconsistencies in the scoring rules that maximize incentives for effort across time.
In particular, to implement maximum effort in dynamic environments, the moral hazard constraints often bind at \emph{both} time~$1$ and the stopping time $\tau$ when the signals are perfectly revealing. 
A V-shaped scoring rule with a kink at $\posterior^N_{\tau}$ would yield insufficient incentives for the agent to exert effort at time~$0$, preventing the agent from working in the first place.  
Conversely, a V-shaped scoring rule with a kink at $\posterior^N_{0}$ results in insufficient incentives for the agent to exert effort at time~$\tau$, causing the agent to stop prematurely. 
To balance the incentives for exerting effort across different time periods, the V-shaped scoring rule may need a kink located at some interior belief. Moreover, as we will discuss, the optimal scoring rule need not always take a V-shaped form to provide balanced incentives in dynamic environments.

\paragraph{Comparison to contracts for experimentation}
In standard experimentation models with publicly observed outputs \citep[e.g.,][]{HKL2016}, the principal can condition rewards on outputs, potentially increasing rewards over time to offset rising effort costs.  In our setting, this approach is infeasible: an agent with favorable early ``outputs'' (i.e., Poisson signals) could strategically withhold information to claim higher rewards later. Instead, to strengthen incentives for continued effort, the optimal dynamic contract decreases rewards over time. This reduction in rewards is beneficial because it minimizes the agent's utility from stopping early, thereby encouraging continued effort. On top of this, the arrival of ``output'' yields private information for the agent. This additional endogenous private information creates an extra complication in our model.

\section{Numerical Illustrations} \label{sec:numerical}

This section presents numerical computations of optimal contracts for an important special case of our model which we will revisit later, namely when the set of Poisson signals is binary: 
\begin{align*}
S=\{\lsignal,\rsignal\}.
\end{align*}
Furthermore, we assume that $\posterior_t^\lsignal<\posterior_t^\no<\posterior_t^\rsignal$---thus, the left-biased signal $\lsignal$ moves beliefs toward $\lstate$, and the right-biased signal $\rsignal$ moves beliefs toward $\rstate$.  Specializing to this case, \cref{thm:optimal_dynamic_menu} implies that
\[
\reward_t^\lsignal=\reward_t^\no
\qquad \text{for all } t\leq \stoptime_{\rewardMenu},
\]
while $\reward_t^\rsignal$ is given by the supporting line to $G_t$ at $\posterior_t^\rsignal$.

Note that since the optimization program \eqref{program:effort_max} is a linear program for any $\stoptime\leq T$, it suggests a general method for how to numerically find an  optimal menu by applying standard linear programming algorithms. The solution to the linear program not only tells us how high $c$ can be, but also identifies a particular contract that can implement efforts for $\tau$ periods given this $c$.

\begin{figure}[t]
    \centering
    \hspace{-85pt}
    \subfloat{\includegraphics[width=0.45\textwidth]{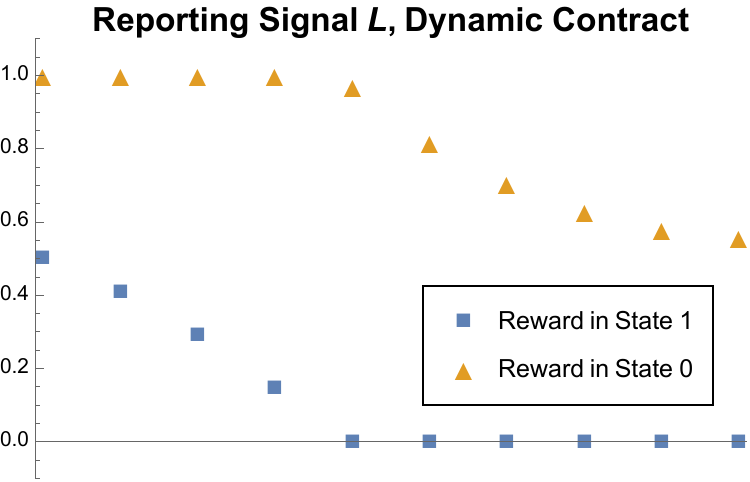}
   \hspace{-40mm}  \small{Time}
}
    \hspace{100pt}% \hfill
    \subfloat{\includegraphics[width=0.45\textwidth]{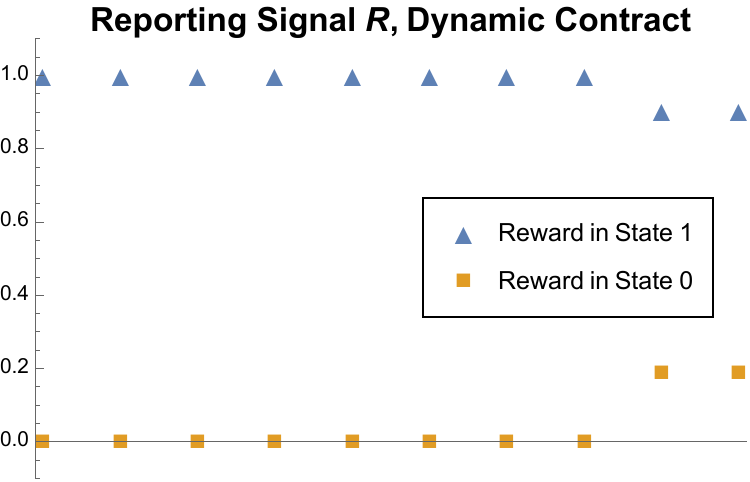}
  \hspace{-40mm}  \small{Time}
}

\caption{\footnotesize Solution to the linear program for the effort-maximizing contract with 10 periods, allowing for reporting at any time---i.e., unrestricted advice. The parameters are chosen as $\lambda_{1}^{\rsignal}=1/3, \lambda_{1}^{\lsignal}=\lambda_{0}^{\rsignal}=1/10, \lambda_{0}^{\lsignal}=1/5, \prior=2/3$.}
    \label{fig:DynamicSolutionNumerical}
\end{figure}

Figure \ref{fig:DynamicSolutionNumerical} presents a solution for certain representative parameter values. We compute the optimal menu for implementing a stopping time $\stoptime = 10$. 
The left side of \Cref{fig:DynamicSolutionNumerical} illustrates the decreasing reward structure in the optimal menu characterized in \cref{thm:optimal_dynamic_menu}: rewards following $\lsignal$ signals are first decreasing for state $1$, then decreasing for state $0$ after their rewards for state $1$ hit $0$. 
Moreover, the realized rewards following $\rsignal$ signals are not necessarily decreasing over time in all states;
in particular, rewards in state $0$ may increase over time, although due to incentive compatibility, such an increase in rewards in state $0$ implies a decrease in rewards in state $1$.

We now describe how to modify the linear program to restrict it to \emph{static} contracts. Notice that the optimal dynamic menu allows the principal to discriminate over time: An agent who observes signal $s$ at time $t$ cannot choose the reward function that would have been selected had that signal arrived at time $t' < t$. Such deviations \emph{are} possible under single-elicitation mechanisms.
Therefore, the incentive compatibility constraints in the static mechanisms would instead require that for any $t\leq \stoptime_{\rewardMenu}$,
\begin{align}
\reward^{\signal}_{t} \in \argmax_{\reward \in \rewardMenu_0} \util(\posterior_t^{\signal}, \reward), 
\qquad \forall \signal\in \signals. \tag{Static-IC}\label{constraint:static_IC}
\end{align}
By replacing the \eqref{constraint:IC} with \eqref{constraint:static_IC} in \eqref{program:effort_max}, we obtain the linear program for optimization over scoring rules in this dynamic environment. 

\begin{figure}[t]
    \centering
    \hspace{-85pt}
    \subfloat{\includegraphics[width=0.45\textwidth]{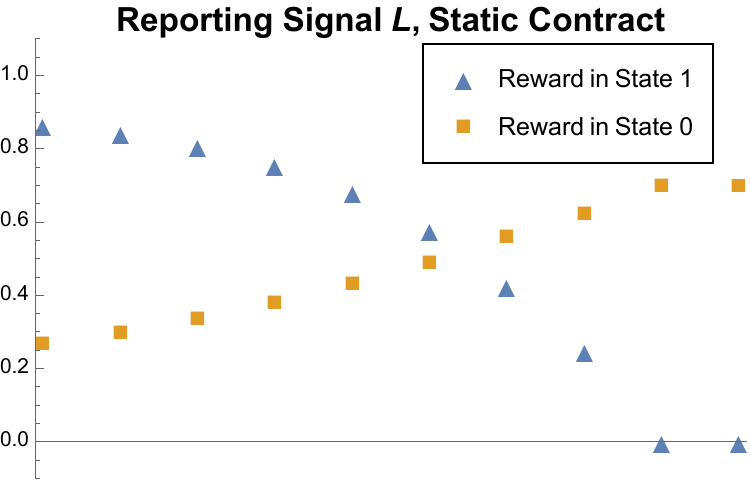}
   \hspace{-40mm}  \small{Time}
}
    \hspace{100pt}% \hfill
    \subfloat{\includegraphics[width=0.45\textwidth]{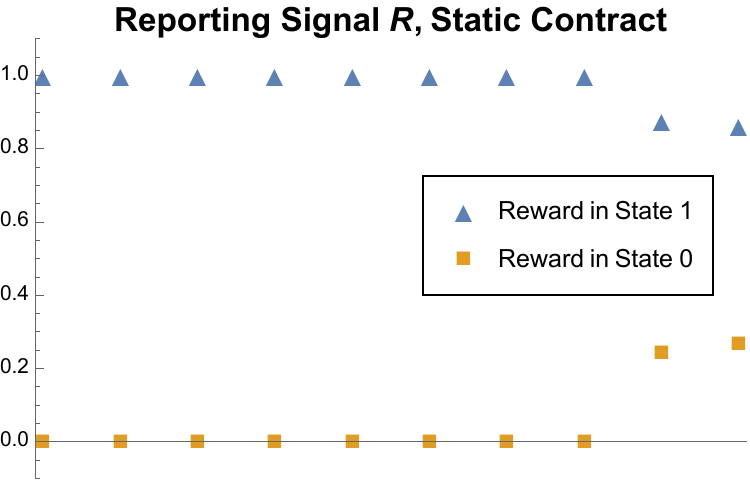}
  \hspace{-40mm}  \small{Time}
}
\caption{\footnotesize Solution to the linear program for the effort-maximizing contract with 10 periods, assuming that the contract involves a single elicitation at time 10---i.e., summarized advice. Note that while time is illustrated, any reward can be selected at any time, and the rewards presented are those selected consistent with incentive compatibility. The parameters are chosen as $\lambda_{1}^{\rsignal}=1/3, \lambda_{1}^{\lsignal}=\lambda_{0}^{\rsignal}=1/10, \lambda_{0}^{\lsignal}=1/5, \prior=2/3$.}
    \label{fig:StaticSolutionNumerical}
\end{figure}

We illustrate the solution to this linear program in \Cref{fig:StaticSolutionNumerical}. Here, the different rewards correspond to what the agent would optimally select from the \emph{same} menu if the signal were to arrive at that time, even though all reports are made after time $10$. 
A notable feature is that the rewards in state $0$ following signal $\lsignal$ are increasing over time, rather than decreasing as in the dynamic case. Intuitively, the dynamic contract in \Cref{fig:DynamicSolutionNumerical} uses a relatively high early reward in state $0$ after $\lsignal$ to encourage the agent to begin working, and then lowers it over time; however, under static contracts, this front-loading would violate incentive compatibility because an agent who worked longer would prefer to ``backdate'' his report and pick that generous early $\lsignal$ reward. To prevent such profitable backdating, the optimal static menu instead makes various offers at the same time, and the agent receiving a $\lsignal$ signal at a later time would be more certain of the state being $0$, selecting a reward option with a higher reward for state $0$. 
Dynamic elicitation avoids this tension because the contract available at time $t$ is not available at later times, so the principal can start with high incentives and taper them without creating incentives to mimic earlier types.

These computations show that a contract implementing maximum effort among \emph{static} contracts can be qualitatively quite different from those that can resort to dynamic elicitation. The linear programs further identify a range of costs such that the principal can implement effort up to time 10 under a dynamic contract, but not using static elicitation (specifically, whenever $0.041 \leq c \leq 0.049$). For this range of cost parameters, dynamic elicitation enables the principal to incentivize \emph{strictly more} effort. Our main results show that no gap arises in the single-signal, stationary, or perfectly revealing environments. Crucially, the parameters for these solutions do not fit under any of these environments. 

These illustrations show that some condition is necessary for static contracting to implement maximum effort. One natural way to try to prove the sufficiency of summarized advice would be to show that the constraints added to \eqref{constraint:static_IC} do not bind under these assumptions. The limitation is that it is not immediately clear how to determine which constraints will bind at which time given the properties of the linear program alone. Hence, this exercise motivates different techniques to determine the \emph{qualitative properties} of such contracts. The message we deliver is that the existence of a gap in this numerical example is attributable to the existence of noisy $\lsignal$ realizations.

\section{Implementing Maximum Effort via Scoring Rules}
\label{sec:static_is_optimal}

We now show that scoring rules implement maximum effort in the three environments highlighted in \cref{sec:infoacq}: stationary, perfect-learning, and single-signal.

The central challenge is that different beliefs require different reward options to maximize incentives for exerting effort. As the agent works without receiving a signal, his no-information belief $\mu^{\no}_t$ drifts toward $0$. When $\mu^{\no}_t$ is close to $1$, strong rewards for reporting $\theta=1$ weaken effort incentives---so the agent would exert effort only if $c$ is small. When $\mu^{\no}_t$ approaches $0$, the situation reverses: strong rewards for reporting $\theta=1$ are necessary to encourage the agent to keep working. Effort-maximizing contracts balance the incentives across different periods as the posterior belief drifts from close to 1 to close to 0. 

We demonstrate the sufficiency of static elicitation constructively: starting from any dynamic contract, we identify a static contract such that the agent's incentives for exerting effort are preserved in all continuation games. Anticipating \cref{sec:dynamic_optimal}, the ``problematic'' beliefs in our failure result arise when the agent obtains an imperfectly informative Poisson signal that moves beliefs toward state $0$.  \cref{sec:perfect} illustrates why the principal should never withhold rewards if the Poisson signal is \emph{conclusive} that the state is $0$; in particular, maximizing rewards at every time implies that these rewards should be made constant. 
\cref{sec:single_source} shows that the principal can convexify the agent's no-information utility to establish a static implementation when only one Poisson signal can arrive. 
While details differ across environments, the unifying principle is that, under these learning technologies, the structure of beliefs is sufficiently constrained to allow rewards to be adjusted in a way that renders the contract effectively static, thus managing the tension between providing incentives in early and late periods.

\subsection{Stationary Environment}
\label{sub:stationary}

In the stationary environment, the agent's belief never changes in the absence of Poisson signals: $\posterior^{\no}_t=\prior$ for all $t\le T$. This removes any intertemporal tension.

\begin{theorem}[Stationary Environment]\label{thm:stationary}\,\\
In the stationary environment, a V-shaped scoring rule with kink at $\prior$ is effort-maximizing.
\end{theorem}

\cref{prop:static_moral_hazard} implies that for any continuation game, the effort-maximizing reward depends only on the current belief. Since this belief is constant across time, the same V-shaped rule simultaneously maximizes incentives everywhere. There is no conflict between the incentives at date 0 and those at date $T$ (or any date in between).

The more subtle cases arise when beliefs drift as the agent exerts effort. In those cases, the optimal scoring rule may not take the form of a V-shape with a kink at the prior, reflecting a compromise between early-period incentives (where beliefs are near $\prior$) and late-period incentives (where beliefs approach $\mu^{\no}_{\stoptime}$).

\subsection{Perfect-learning Environment}
\label{sec:perfect}

When signals fully reveal the state, the sufficiency of scoring rules is less immediate.\footnote{This setting generalizes perfect-learning technologies from past work by allowing signals to potentially reveal \emph{either} state.}
Note that when signals fully reveal the state, it is without loss to focus on binary signals $\{\lsignal,\rsignal\}$.

\begin{theorem}[Perfect-learning Environment]\label{thm:perfect}\,\\
In the perfect-learning environment, there exists  $r_{1} \in (0,1]$ such that a V-shaped scoring rule with parameters $r_{0}=1$ and $r_{1}$ (and hence, a kink at $\frac{1}{1+r_{1}} \in [1/2,1)$) implements maximum effort. 
\end{theorem}

A key feature of perfect learning is that once a signal arrives, the problem ends: the agent becomes certain of the state. Thus, the only interior beliefs are those at which the agent has not observed any informative signal. This greatly simplifies the structure, as offering rewards
in both states can only weaken incentives: it is without loss to give rewards only in the
state the agent believes is more likely.

The proof follows two steps, each replacing part of an arbitrary dynamic contract with another contract that (weakly) strengthens effort incentives. Our discussion uses the notation for binary signals introduced in Section \ref{sec:numerical}.

\medskip

\noindent \textbf{Step One: Maximizing rewards in state $\lstate$.}
The first step is to show that the principal should always maximize the reward for the state realization of $0$ when a left-biased signal~$\lsignal$ is observed or no Poisson signal ever arrives.

To see why, fix an optimal contract implementing stopping time $\stoptime_{\rewardMenu}$, and consider the agent's decision in the period $\stoptime_{\rewardMenu}$. Since the agent will stop working in the next period whether or not a Poisson signal is observed, the agent's final payoff depends only on the signal observed in that period. His payoff from exerting effort at time $\stoptime_{\rewardMenu}$ is therefore: 

\begin{equation*} 
-c + (1-\mu_{\stoptime_{\rewardMenu}-1}^{\no})(\lambda_{0}^{\lsignal}\reward_{\stoptime_{\rewardMenu},0}^{\lsignal}+(1-\lambda_{0}^{\lsignal})\reward_{\stoptime_{\rewardMenu},0}^{\no})
+\mu_{\stoptime_{\rewardMenu}-1}^{\no} \lambda_{1}^{\rsignal} \reward_{\stoptime_{\rewardMenu},1}^{\rsignal}\end{equation*}

\noindent Since rewards are decreasing over time by \cref{thm:optimal_dynamic_menu}, and since the agent can always select the largest available reward, this expression is equal to: 

\begin{equation} 
-c + (1-\mu_{\stoptime_{\rewardMenu}-1}^{\no})\reward_{\stoptime_{\rewardMenu},0}^{\lsignal}+\mu_{\stoptime_{\rewardMenu}-1}^{\no} \lambda_{1}^{\rsignal} \reward_{\stoptime_{\rewardMenu},1}^{\rsignal} \label{eq:perfwork} \end{equation}

\noindent But stopping immediately and reporting 0, without exerting any effort at all, yields a payoff: 
\begin{equation} 
(1-\mu_{\stoptime_{\rewardMenu}-1}^{\no}) r_{\stoptime_{\rewardMenu}-1,0}^{\lsignal} \label{eq:perfshirk}
\end{equation}

\noindent Taking the difference between (\ref{eq:perfwork}) and (\ref{eq:perfshirk}) reveals that exerting effort is more favorable when $r_{\stoptime_{\rewardMenu},0}^{\lsignal}-r_{\stoptime_{\rewardMenu}-1,0}^{\lsignal}$ is larger. Since $r_{\stoptime_{\rewardMenu},0}^{\lsignal} \leq r_{\stoptime_{\rewardMenu}-1,0}^{\lsignal}$, raising both rewards---up to 1---can only increase this difference and therefore the gain from exerting effort. 

Now consider the agent's incentives at time $t$ more generally if the principal offers the full reward for a report of left-biased signal $\lsignal$. Adding this reward will always improve the continuation payoff from exerting effort by \emph{at least} $(1-\mu_{t-1}^{\no})(1-r_{t,0}^{\lsignal})$---following the previous logic, this change would exactly be the increase in payoffs if rewards following signal $\lsignal$ were constant after time~$t$. 
If, instead, they are strictly decreasing, then the change yields a strictly higher gain from effort. However, the gain from stopping effort is always \emph{at most} $(1-\mu_{t-1}^{\no})(1-r_{t-1,0}^{\lsignal})$---again, it would be exactly this amount if the agent were to report signal $\lsignal$ when shirking; otherwise, there would be no change at all. Since the gain from exerting effort is always at least as large as the gain from shirking, this modification always (at least weakly) strengthens the incentives to exert effort.  

\begin{figure}
\centering
\begin{tikzpicture}[xscale=0.8,yscale=0.8]
\hspace{-10pt}
\draw [<->] (0,5.2) -- (0,0) -- (5.2,0);
\draw [dotted] (0, 5) -- (5, 5) -- (5, 0);

\draw (-0.3, 5) node {$1$};

\draw (0, -0.3) node {$0$};
\draw (5, -0.3) node {$1$};

\draw [red] (0,5) -- (5,0);
\draw [dashed] (0,0) -- (5,4);

\draw [dotted] (0,4) -- (5,4);
\draw (-0.5, 4) node {$r^{\rsignal}_{\hat{t},1}$};

\draw [dotted] (2.78,0) -- (2.78,2.22);
\draw (2.78, -0.4) node {$\mu^{\no}_{\hat{t}}$};

\draw [thick] plot [smooth, tension=0.5] coordinates { (1.5,0.8) (2.3, 1.4 ) (2.78,2.22)};
\draw [thick] plot [smooth, tension=0.5] coordinates {(2.78,2.22) (3.5, 3.5) (4.3, 4.1)(4.5, 4.4)};
\end{tikzpicture}
\caption{\footnotesize Illustration of the second step of the proof of \cref{thm:perfect}}
\label{fig:secondstepthm2}
\end{figure}

\medskip

\noindent \textbf{Step Two: Collapsing the rewards after signal $\rsignal$ into a single payment}
The second step shows that it suffices to offer just one menu option in case signal $\rsignal$ is observed. The idea is to replace all rewards following signal $\rsignal$ with the \emph{single} menu option $\reward^{\rsignal}_{\hat{t}}$, where $\hat{t}$ is selected as the minimum time (or equivalently the highest belief $\mu_{\hat{t}}^{\no}$) at which the agent would prefer rewards $(1,0)$ to $\reward^{\rsignal}_{\hat{t}}$ in the absence of a Poisson signal arrival. 
\Cref{fig:secondstepthm2} illustrates this replacement: the red line represents the agent's expected payoff if choosing the reward $(1,0)$, while the thick black line represents the agent's payoff if stopping effort and pretending to have obtained the $\rsignal$ signal at a given on-path belief. The replacement we identify is illustrated by the dashed black line. 

Why does this modification weakly strengthen the incentives to exert effort? When $\mu_{t-1}^{\no} \leq \mu_{\hat{t}}^{\no}$ or equivalently $t \geq \hat{t}+1$, the agent's payoff from not exerting effort in the continuation game $\cgame_t$ remains unchanged under the new contract: the replacement is chosen precisely so that the reward option $(1,0)$ is selected at any such no information belief. But since rewards decrease over time, this modification ensures a larger reward following a right-biased signal $\rsignal$, strengthening the incentives to exert effort. 
The case of $\mu_{t-1}^{\no} > \mu_{\hat{t}}^{\no}$ or equivalently $t \leq \hat{t}$ is more subtle: Here, the payoff from both exerting effort \emph{and} stopping in the continuation game $\cgame_t$ \emph{decreases}. 
However, it turns out that the payoff from stopping decreases by more. Indeed, the decrease in the expected utility from stopping at no information belief $\mu_{t-1}^{\no}$ is $\mu_{t-1}^{\no}(\reward_{t-1,1}^{\rsignal}-\reward_{\hat{t},1}^{\rsignal})$. 
The decrease in expected utility from exerting effort is at most $\mu_{t-1}^{\no}(\reward_{t,1}^{\rsignal}-\reward_{\hat{t},1}^{\rsignal})$ since there is a chance that an $\rsignal$ signal is never observed, which implies that the probability of this change in rewards being relevant when exerting effort is less than $\mu_{t-1}^{\no}$. 
Thus, the payoff from exerting effort decreases by less than the payoff from stopping due to the decreasing rewards structure. Therefore, the agent is still willing to exert effort under this replacement. See \cref{apx:perfect} for complete formal details.

\subsection{Single-signal Environment}
\label{sec:single_source}
We now consider the environment where absent Poisson signals, $\posterior_t^{\no}$ drifts toward~0, while a single Poisson signal causes beliefs to jump upward.

\begin{theorem}[Single-signal Environment]\label{thm:single_source}\,\\ 
In the single-signal environment, there exists a scoring rule implementing maximum effort.
\end{theorem}

\noindent \cref{thm:single_source} is more subtle than \cref{thm:perfect} because imperfect signals might require positive rewards in both states to implement optimal dynamic efforts. This creates additional complexity: we cannot simply ``flatten'' rewards as in \cref{thm:perfect} without carefully tracking how these changes affect incentives across all beliefs.

\medskip

\noindent \textbf{Step One: Convexifying the no-information utility} Our first step is to replace the arbitrary dynamic contract with one where the no-information utility is ``convexified''---i.e., replacing the contract with one where $\util^{\no}_t$ is convex in $\posterior^{\no}_t$. While this property holds for any (static) scoring rule, this property need not hold for dynamic contracts.

We describe the argument. Define $\bar{t}$ to be the earliest time such that the no-information utility exhibits a nonconvexity. A nonconvexity in the no-information utility implies that $\bar{t} \leq \stoptime_{\rewardMenu}$. One of our observations is that the dynamic incentives imply that the effort constraint cannot be binding at $\bar{t}+1$ under non-convexity.\footnote{See \Cref{fig:overlapping_belief_discrete} in \cref{apx:single_source} for an illustration.} 
Intuitively, the no-information utility at time $t$ cannot be larger than the convex combination between (1) the reward of receiving a signal~$\rsignal$ immediately at time $t+1$; and (2) the continuation payoff at time $t+1$ for not receiving any Poisson signal at time $t$. 
If the effort constraint is binding at time $\bar{t}+1$, the latter equals the no information utility at time $\bar{t}+1$, which leads to a contradiction due to the non-convexity of the no information utility.

On the other hand, if the agent's incentive to exert effort is slack at time $\bar{t}+1$, we may raise some rewards in the menu option $\reward_{\bar t}^{\no}$ to smooth the non-convex piece without violating earlier IC constraints. Repeating this argument pushes convexity forward through time until the entire $\util_t^{\no}$ becomes convex.

\medskip

\noindent \textbf{Step Two: Identifying a static implementation of the scoring rule} Since convex functions are equal to the upper envelope of the linear functions below them, a natural conjecture in light of Step One is that the principal could offer a static contract providing the set of reward functions that are tangent to $\util_{t}^{\no}$ at some belief. This contract provides the same payoff if effort is not exerted; moreover, no higher rewards could be provided if Poisson signals arrive without increasing the no-information utility above $\util_{t}^{\no}$.

However, this argument does not work since the constructed rewards may lie outside of $[0,1]$ for an arbitrary no-information utility function; 
for instance, consider $\util^{\no}_t = (\mu^{\no}_t)^2$. 
A dynamic contract that implements this no-information utility is a constant reward $(\mu^{\no}_t)^2$ at time $t$ regardless of the realization of the state. 
However, to implement this utility function using a scoring rule, 
by \cref{prop:static_truthful},
the menu option for belief $\posterior\in[0,1]$ must be $(-\posterior^2,2\posterior-\posterior^2)$, 
which violates the ex post individual rationality constraint. 

\begin{figure}
\centering
\begin{tikzpicture}[xscale=0.8,yscale=0.8]
\hspace{-10pt}
\draw [<->] (0,3.1) -- (0,0) -- (7.3,0);

\draw [line width=0.8mm, dashed, gray] plot [smooth, tension=1] coordinates { (1, 2) (3, 1.5) (4, 1.5)};
\draw [line width=0.8mm, dashed, gray] (4, 1.49) -- (5, 1.85);

\draw [thick] plot [smooth, tension=1] coordinates { (1, 2) (3, 1.5) (4, 1.5)};
\draw [thick] plot [smooth, tension=0.8] coordinates { (4, 1.5) (4.5, 1.9) (5, 3)};

% \draw [dashed] (0.5, 1) -- (4, 1.5);
% \draw [dotted] (0.5, 0) -- (0.5, 1);
% \draw (0.2, -0.4) node {$\mu^{\no}_{\hat{t}+\Delta}$};

% \draw [dotted] (4, 0) -- (4, 1.5);
% \draw (4, -0.4) node {$\mu^{\no}_{t^*}$};

% \draw [dotted] (2.5, 0) -- (2.5, 1.6);
% \draw (2.5, -0.4) node {$\mu^{\rsignal}_{\hat{t}}$};
% \draw [red] (0.5, 1) -- (2.5, 1.6);

% \draw [dotted] (1, 0) -- (1, 2);
% \draw (1.1, -0.4) node {$\mu^{\no}_{\hat{t}}$};

\draw (7, -0.3) node {$1$};
\draw (0, -0.3) node {$0$};

\draw [dashed] (0, 2.32) -- (7, 0);
\draw [red] (0, 0) -- (7, 2.6);
\draw [dotted] (7, 0) -- (7, 2.6);

\draw [dotted] (4, 0) -- (4, 1.5);
\draw (4, -0.4) node {$\mu^{\no}_{\hat{t}}$};

\end{tikzpicture}
\caption{\footnotesize The black curve is the convex no-information utility of the agent, and $\hat{t}$ is the minimum time with a bounded tangent line (red line). 
The thick dashed line is the no-information utility of the agent
given a feasible scoring rule that offers a menu option that corresponds to the red line instead of the black curve for belief $\posterior \geq \posterior^{\no}_t$.
}
\label{fig:overlapping_bounded}
\end{figure}
Intuitively, this kind of a violation of the constraint that rewards lie in $[0,1]$ arises because the no-information utility is too convex.
In this case, we can flatten the no-information utility by decreasing the reward at earlier times---analogous to Step Two of \cref{thm:perfect}. 
Under this flattening, 
the no-information utility decreases weakly more than the continuation payoff in all continuation games, 
yielding stronger incentives to exert effort. \Cref{fig:overlapping_bounded} illustrates this idea, with details provided in \cref{apx:single_source}.

\section{Insufficiency under Noisy Learning and Slow Drift}

\label{sec:dynamic_optimal}

We now turn to environments with noisy signals and slow drift. We specialize to the binary-signal environment from Section \ref{sec:numerical}. Noisy signals satisfy: 

\[\lambda_1^{\rsignal} > \lambda_0^{\rsignal} > 0 \text{ and } 0 < \lambda_1^{\lsignal} < \lambda_0^{\lsignal},\] i.e., learning is imperfect, and either signal may arise under either state. Slow drift means: 
\[ \lambda_0^{\rsignal}+\lambda_0^{\lsignal} \lessapprox \lambda_1^{\rsignal}+\lambda_1^{\lsignal}  \] 
so that the belief moves gradually towards $0$ when no Poisson signal arrives. Importantly, the single-signal assumption implies that, for a fixed arrival rate of 
$\rsignal$, beliefs drift at the maximal rate conditional on non-arrival. In this sense, the single-signal environment represents an extremal benchmark opposite to settings with slow drift. 

Intuitively, increasing a late reward now affects rewards for exerting effort \emph{less} than shirking incentives under noisy signals: e.g., if the reward for the state being $0$ is increased, the agent may not enjoy this increase when exerting effort if he receives a right-biased signal $\rsignal$ while the state is $0$.
Thus, the logic of \cref{thm:perfect} fails. Furthermore, since beliefs drift toward state 0, rewards for receiving a left-biased signal $\lsignal$ can be lowered over time without incentivizing the agent to shirk and lie. As described in \cref{sect:discussions}, these modifications are in the direction of those that maximize the incentives for the agent to exert effort.  

To state our results, we first provide a formal bound on how long the horizon should be so that the time horizon is not a binding constraint. We define this time as $T_{\lambda,\prior,c}$: 

\begin{lemma}\label{lem:max_stop_time}
Fix any signal arrival rates satisfying $\lambda_{1}^{\rsignal} > \lambda_{0}^{\rsignal}$. Let $\posterior_{\lambda,c} \triangleq \min\{\frac{1}{2},\frac{c}{\lambda^{\rsignal}_1 - \lambda^{\rsignal}_0}\}$ 
and let $T_{\lambda,\prior,c}$ be the maximum time such that $\posterior^{\no}_{T_{\lambda,\prior,c}-1} \geq  \posterior_{\lambda,c}$. The stopping time $\stoptime_{\rewardMenu}$ satisfies 
$\stoptime_{\rewardMenu} \leq T_{\lambda,\prior,c}$, given any prior $\prior \in (0,1)$, cost of effort $c$, and contract $\rewardMenu$ with rewards belonging to $[0,1]$.
\end{lemma}
\noindent Intuitively, $T_{\lambda,\prior,c}$ is the maximum calendar time the agent can be incentivized to exert effort 
in any contract when the prior is $\prior$. The following result provides sufficient conditions under which more complex dynamic structures are necessary in an effort-maximizing contract:

\begin{theorem}[Strictly Less Effort Under Scoring Rules]\label{thm:dynamic optimal}\,\\
Fix any prior $\prior \in (0,\frac{1}{2})$, any cost of effort $c$,
and any constant $\kappa_0>0, \frac{1}{4}\geq\bar{\kappa}_1 > \underline{\kappa}_1 > 0$. There exists $\epsilon > 0$ such that for any $\lambda$ satisfying: 
\begin{itemize}
    \item $\lambda^{\rsignal}_1 - \lambda^{\rsignal}_0 \geq \frac{1}{\prior}(c + \kappa_0)$;
    \hfill (sufficient-incentive)
    \item $\lambda_1^{\lsignal},\lambda_0^{\lsignal}, \lambda_0^{\rsignal}, \lambda_1^{\rsignal} \in [\underline{\kappa}_1, \bar{\kappa}_1]$; 
    \hfill (noisy-signal) 
    \item $\lambda_1^{\rsignal}+\lambda_1^{\lsignal} \in (\lambda_0^{\rsignal}+\lambda_0^{\lsignal}, \lambda_0^{\rsignal}+\lambda_0^{\lsignal} +\epsilon) $, \hfill (slow-drift) 
    \item and $T\geq T_{\lambda,\prior,c}$, \hfill (sufficient-horizon)
\end{itemize}
any static scoring rule implements effort strictly less than the maximum. 
\end{theorem}

\noindent The first and last assumptions prevent trivial horizon or zero-incentive constraints. The substantive assumptions are noisy signals and slow drift. The former implies the technology is far enough from the case considered in \cref{thm:perfect}, while the latter reflects the technology is far enough from the case considered in \cref{thm:single_source} as described above.

Our proof considers a particular contract that can outperform any static contract:

\begin{definition}[Myopic-incentive Contract]\label{def:myopic-incentive}\,\\
When prior $\prior\in (0, \frac{1}{2})$, a contract $\rewardMenu$ with menu options $\{\reward^{\signal}_t\}_{t\leq \stoptime_{\rewardMenu},\signal\in\signals}\cup\{\reward^{\no}_{\stoptime_{\rewardMenu}}\}$ 
is a \emph{myopic-incentive contract} if $\reward^{\rsignal}_t = (0,1)$ and $\reward^{\lsignal}_t = (\frac{\posterior^{\no}_{t-1}}{1-\posterior^{\no}_{t-1}},0)$ for any $t\geq 1$, 
and $\reward^{\no}_{\stoptime_{\rewardMenu}}=\reward^{\lsignal}_{\stoptime_{\rewardMenu}}$. 
\end{definition}
% \begin{figure}
% \centering
% \input{fig/dynamic_score}
% \caption{Illustration of myopic-incentive contract.
% The solid line is the expected reward of the agent, as a function of his belief, for not exerting effort and reporting his belief truthfully at time $t$. }
% \label{fig:dynamic_score}
% \end{figure}

\noindent Intuitively, this contract offers the full reward when guessing state 1 and  $\frac{\posterior^{\no}_{t-1}}{1-\posterior^{\no}_{t-1}}$ when guessing state 0. 
The rewards are carefully chosen such that an agent with belief $\posterior^{\no}_{t-1}$ is indifferent between these two reward options. 
Note that since $\posterior_{t-1}^{\no}$ decreases over time, so does the reward when the agent guesses state 0. This property is necessary for the constructed contract to be incentive-compatible. 

The terminology `myopic-incentive' reflects that this contract reoptimizes rewards period-by-period. 
The rewards $\reward^{\rsignal}_t = (0,1)$ and $\reward^{\lsignal}_t = (\frac{\posterior^{\no}_{t-1}}{1-\posterior^{\no}_{t-1}},0)$ form a V-shaped scoring rule with a kink at belief $\posterior^{\no}_{t-1}$, which is the optimal scoring rule for the continuation game with the prior being $\posterior^{\no}_{t-1}$.
We mention that the myopic-incentive contract typically does \emph{not} provide maximal incentives to exert effort in all continuation games---after all, rewards decrease over time, so an agent who exerts effort and receives a Poisson signal in later periods may receive a lower reward. However, \Cref{lem:contract_approx_optimal} in the Online Appendix shows that these contracts nevertheless provide incentives \emph{close} to the effort-maximizing contract under slow-drift. Intuitively, slow drift ensures that the aforementioned decreases in the rewards over time have a minor impact on the agent's incentives to exert effort at any belief.

\cref{thm:dynamic optimal} follows from noting that the same cannot be said for the effort-maximizing \emph{static} contract. Letting $\rewardMenu^{*}$ denote the effort-maximizing dynamic contract, for a contract to perform only slightly worse than  $\rewardMenu^{*}$, it is still necessary to provide incentives to exert effort at times close to $\stoptime_{\rewardMenu^{*}}$. \cref{sub:comparative statics of static} illustrates that when signals are noisy, such static contracts are too weak to incentivize the agent to exert effort at time $0$, leading to a contradiction.  Simply put, dynamic contracts can adapt incentives to both the initial beliefs as well as the later beliefs; static contracts cannot. 

We mention that the conditions in \cref{thm:dynamic optimal} facilitate a comparison in the amount of effort implemented by the myopic-incentive contract and an arbitrary (static) scoring rule. In particular, these conditions avoid the need to provide more structure to make this comparison sharp---for instance, to avoid integer issues which would suggest the increased strength from dynamics are insufficient to induce another period of effort. Nevertheless, the theorem delivers our main qualitative message: a sufficiently strong joint violation of the conditions supporting static elicitation can make dynamic elicitation necessary.

Along these lines, an assumption in \cref{thm:dynamic optimal}, but one that appears relatively harmless, is that $\prior \in (0,\frac{1}{2})$. This restriction ensures the menus in the myopic-incentive contract remain within $[0,1]$. Combined with downward drift, this ensures feasibility and compatibility with truthful reporting in myopic incentive contracts.

\section{Additional Observations and Final Thoughts}
\label{sec:extension}

We conclude with some additional discussion, presenting our final thoughts and proposing several open questions in \cref{sub:conclusion}.

\subsection{Properties of Optimal Scoring Rules Under Dynamics}
\label{subsec:scoring_under_dynamic}

The results in \cref{sec:static_is_optimal} provide conditions such that the optimal contract has a static implementation as scoring rules. 
However, the structure of the optimal scoring rules remains elusive. 
In this section, we show that for perfect-learning and single-signal environments, 
the dynamics in information acquisition still play a crucial role in determining the optimal scoring rules. 
In particular, we show that the optimal scoring rules in those two dynamic environments differ from those in static environments. 

\paragraph{Perfect-learning} 
In \cref{thm:perfect}, we show that a V-shaped scoring rule implements the maximum effort. 
However, the effort-maximizing choice of the parameter $r_{1}$ is not provided. 
In this section, we show that the optimal choice of $r_{1}$ would not lead to a V-shaped scoring rule with a kink at the prior, resulting in a difference compared to the static environments. 

To illustrate the idea, we consider an environment where $T$ is sufficiently large so that the time horizon would not be a binding constraint for exerting costly effort. 
Given a V-shaped scoring rule $\score$ with parameters $r_0 = 1$ and $r_1 \in [0,1]$,
recall that $\util^{\no}_t$ is the agent's utility when not exerting effort after time~$t$. 
Let $U^+_t$ be the value function of the agent at the prior belief $\posterior^{\no}_{t-1}$: That is, the agent's payoff when exerting effort optimally in at least one period starting from (and including) time $t$.
It is straightforward that $U^+_t$ is convex in $\posterior^{\no}_{t-1}$, with its derivative between $-1$ and $r_1$. 
We let $\underline{\mu}(r_1)\leq \bar{\posterior}(r_1)$ be the beliefs such that $U^+_t$ intersects $u^N_t$.\footnote{More formally: Assuming that $U^+_t \geq \util^{\no}_{t-1}$ for some $t\geq 0$, we define $\underline{\mu}(r_1) = \min\{\posterior^{\no}_t, t\geq 0: U^+_t \geq \util^{\no}_{t-1}\}$
and $\bar{\posterior}(r_1) = \max\{\posterior^{\no}_t, t\geq 0: U^+_t \geq \util^{\no}_{t-1}\}$. If $U^+_t < \util^{\no}_{t-1}$ holds for all $t$, then the agent does not work and we take $\underline{\mu}(r_1) =\overline{\mu}(r_1)$ to be undefined.} 
The agent has incentives to exert effort at time $t$ given scoring rule $\score$ if and only if $\posterior^N_t \in [\underline{\mu}(r_1), \bar{\mu}(r_1)]$. 
\Cref{fig:v-shaped_binding_incentive_at_boundary} illustrates how $\underline{\mu}(r_1)$ and $\overline{\mu}(r_1)$ are determined, namely as the intersection between the agent's value function $U_{t}^{+}$ and the payoff attainable without exerting any further effort.

\begin{figure}[t]
\centering
\begin{tikzpicture}[xscale=1,yscale=0.8]
\hspace{-10pt}
\draw [<->] (0,5.2) -- (0,0) -- (5.2,0);
\draw [dotted] (0, 5) -- (5, 5) -- (5, 0);

\draw [thick, blue] plot [smooth, tension=0.5] coordinates { (1, 4) (3, 2.5 ) (4.5, 2.7)};

\draw (-0.3, 5) node {$1$};

\draw (0, -0.4) node {$0$};
\draw (5.1, -0.4) node {$1$};

\draw [red] (0,5) -- (5,0);
\draw [red] (0,0) -- (5,3);

\draw [dotted] (0,3) -- (5,3);
\draw (-0.3, 3) node {$r_1$};

\draw [dotted] (1,0) -- (1,4);
\draw (1, -0.4) node {$\underline{\mu}(r_1)$};
\draw [dotted] (4.5,0) -- (4.5,2.7);
\draw (4.2, -0.4) node {$\bar{\mu}(r_1)$};

% \draw (5.55, 0) node {$\posterior^{\no}_t$};

\draw [dashed] (0,3.5) -- (5,1.8);
\draw [dotted] (3,0) -- (3,2.5);
\draw (3, -0.4) node {$\mu^{\rsignal}_{t}$};

\draw [dotted] (0,1.8) -- (5,1.8);
\draw (-0.3, 3.5) node {$z_0$};
\draw (-0.3, 1.8) node {$z_1$};

\draw [thick] (1,4) -- (2.28,2.72) -- (3.73, 2.23) -- (4.5,2.7);

\end{tikzpicture}
\caption{\footnotesize The illustration for both perfect-learning and single-signal environments. The red lines are the agent's utility $\util^{\no}_{t-1}$ for not exerting effort and the blue curve is the agent's utility $U^+_t$ from exerting effort in at least one period, both as a function of the no information belief $\posterior^{\no}_{t-1}$.
The black curve is the agent's utility for not exerting effort in the alternative scoring rule with additional menu option $(z_0,z_1)$.}
\label{fig:additional_menu}
\label{fig:v-shaped_binding_incentive_at_boundary}
\end{figure}

\begin{lemma}\label{lem:monotone_in_r1}
Both $\underline{\mu}(r_1)$
and $\bar{\posterior}(r_1)$ are weakly decreasing in $r_1$. 
\end{lemma}

In particular, the agent has incentives to exert effort initially 
if and only if $\prior \in [\underline{\mu}(r_1), \bar{\mu}(r_1)]$. \cref{lem:monotone_in_r1} implies that by increasing the reward $r_1$ for predicting state 1 correctly in the scoring rule, 
the agent has incentives to exert effort for longer ($\underline{\mu}(r_1)$ is smaller), 
but the agent has a weaker incentive to exert effort at time $0$ ($\bar{\posterior}(r_1)$ is smaller). 
Therefore, the optimal choice of $r_1$ is 
\begin{align*}
r^*_1 = \max\{r_1 : \bar{\posterior}(r_1) \geq \prior\}.
\end{align*}
Naturally, the kink of the V-shaped scoring rule induced by $r^*_1$ would not be located at the prior $\prior$, as illustrated in \Cref{fig:v-shaped_binding_incentive_at_boundary}. 
Our analysis also highlights the economic intuition behind this difference. 
In dynamic environments, the principal needs to balance the incentives between earlier and later periods. By skewing the kink of the V-shaped scoring rule away from the prior, 
the agent faces a stronger incentive to exert effort in future periods. 

Finally, our observation here also illustrates an interesting dynamic effect in our model: for long time horizons, there exist $\mu_1>\mu_2>\mu_3$ 
such that the agent can be incentivized to work when the no information belief would drift from (a) $\mu_1$ to $\mu_2$
or (b) $\mu_2$ to $\mu_3$, but not when this belief would drift from $\mu_1$ to $\mu_3$.

\paragraph{Single-signal}
A notable feature in the single-signal environment is that, although the effort-maximizing contract can be implemented as a static scoring rule, 
this scoring rule may not be V-shaped. 
Put differently, effort-maximizing scoring rules need not simply involve the agent guessing the state and being rewarded for a correct guess. It may be necessary to reward the agent \emph{even} when the guess is wrong. Consider an interpretation of the minimum reward across the two states as the ``base reward'' and the difference between the rewards as the ``bonus reward.'' From this perspective, 
our results indicate that it may be necessary to consider scoring rules where the base reward is strictly positive. 
This observation may seem counterintuitive, as providing a strictly positive base reward decreases the bonus reward for the agent since total rewards are constrained to the unit interval. A smaller bonus may subsequently lower the agent's incentive to exert effort.

However, the correct intuition is as follows: while providing a strictly positive base reward at time~$t$ decreases the agent's incentive to exert effort at time $t$, it increases the agent's incentive to exert effort at earlier times $t'<t$; indeed, the addition of the positive base reward leads the agent to anticipate higher rewards from exerting effort in cases where the terminal belief is in an intermediate range. 
This modification induces more effort if the agent's incentive constraint for exerting effort initially binds at time~0 but becomes slack at intermediate time $t\in(0, \stoptime_{\rewardMenu})$. 
As illustrated in \Cref{fig:additional_menu}, 
by implementing the effort-maximizing V-shaped scoring rule, the agent's incentive for exerting effort is binding only at the 
extreme time $0$ with belief $\prior = \bar{\posterior}(r_1)$ 
and time $\stoptime_{\rewardMenu}$
with belief $\posterior^{\no}_{\stoptime_{\rewardMenu}} = \underline{\posterior}(r_1)$. 
In this case, since the signals are not perfectly revealing, 
there may exist a time $t$ such that the agent's belief for receiving a right biased signal is $\posterior^{\rsignal}_t \in (\underline{\posterior}(r_1),\bar{\posterior}(r_1))$. 
By providing an additional menu option with a strictly positive base reward in the scoring rule to increase the agent's utility at beliefs $\posterior^{\rsignal}_t$ (e.g., the additional menu option $(z_0,z_1)$ illustrated in \Cref{fig:additional_menu}), 
the agent's incentive constraint for exerting effort at time $0$ is relaxed, and the contract thus provides the agent incentives to exert effort following more extreme prior beliefs without influencing the stopping belief $\underline{\posterior}(r_1)$.

\subsection{Non-invariant Environments}
\label{subsec:non-invariant}

Our model assumes that both the cost of acquiring information and the signal arrival probabilities when exerting effort are fixed over time. On the other hand, we can also allow for more general time-dependent cost functions. 
In particular, there exist settings where the cost of acquiring information is lower closer to the decision deadline, regardless of the previous efforts exerted by the agent. 
In these applications, given a dynamic contract, the best response of the agent may not be a stopping strategy. 
Scoring rules still implement maximal effort in this extension under one of three conditions in \cref{sec:static_is_optimal}---in the sense that given any dynamic contract $\rewardMenu$ and any best response of the agent, 
there exists another contract $\widehat{\rewardMenu}$ that can be implemented as a scoring rule, and the agent's best response in contract $\widehat{\rewardMenu}$ first order stochastic dominates his best response in contract~$\rewardMenu$.\footnote{Let $z_t, \hat{z}_t\in\{0,1\}$ be the effort decision of the agent given contracts $\rewardMenu,\widehat{\rewardMenu}$ respectively conditional on not receiving any Poisson signal before $t$.
We can show that given any time $t$, $\sum_{i\leq t} z_i \leq \sum_{i\leq t} \hat{z}_i$.}
Therefore, the information acquired under contract $\widehat{\rewardMenu}$ is always weakly Blackwell more informative compared to the information acquired under contract~$\rewardMenu$.

An important implication of extending our results to non-invariant costs is that they also apply to settings where the agent discounts future payoffs. Essentially, when the agent discounts future payoffs, it corresponds exactly to exponentially decreasing cost functions, as payments are only made after the state realization after $T$. Consequently, all our main results, such as the optimality of scoring rules in perfect-learning or single-signal environments, extend naturally to this setting.

We can similarly allow for time dependence in the informational environment. Specifically, we can allow for the arrival rates of signals at any time to depend on the amount of effort the agent has exerted until that point. That is, suppose that if the agent has exerted effort for $\tilde{t}$ units of time, then exerting effort produces a right-biased signal that arrives in state~$\theta$ with probability $\lambda_{\theta,\tilde{t}}^{\rsignal}$, and a left-biased signal with probability $\lambda_{\theta,\tilde{t}}^{\lsignal}$ (and no signal with complementary probability). While seemingly minor, this modification induces more richness in the set of possible terminal beliefs as a function of the effort history---for instance, if the terminal beliefs are always in the set $\{\underline{p}, \overline{p}\}$, despite drifting over time. 

Our proof techniques do not make use of the particular belief paths induced by constant arrival rates and hence extend to this case, with the minor exception that Theorem \ref{thm:single_source} requires $\mu_{t}^{\rsignal}$ to be weakly monotone as a function of time (a property that holds when the arrival rate is constant). Otherwise, as long as the parameters stay within each environment articulated in Section \ref{sec:static_is_optimal}, the proofs of these results extend unchanged.

\subsection{Value of Dynamic Contracts}

Our analysis compares contracts in terms of Blackwell informativeness. This is useful because it separates the incentive problem from any particular downstream application. At the same time, the payoff value of the additional informativeness generated by unrestricted advice depends on the decision environment. A small informational improvement may have little value in one application and a large value in another.

This point is particularly stark under the conditions of \cref{thm:dynamic optimal}. Let $\stoptime^{\mathrm S}$ denote the largest stopping time implementable by a scoring rule. By \cref{thm:dynamic optimal}, unrestricted advice can induce strictly more effort, so in particular the agent can be induced to work through time $\stoptime^{\mathrm S}+1$. Since beliefs drift toward state $0$ absent signal arrival and $\lsignal$ is left-biased, the posterior $\posterior_t^{\lsignal}$ is strictly decreasing in $t$. Hence
$\posterior_{\stoptime^{\mathrm S}+1}^{\lsignal}
<
\posterior_{\stoptime^{\mathrm S}}^{\lsignal}$.

Now choose any cutoff
\[
q\in\bigl(\posterior_{\stoptime^{\mathrm S}+1}^{\lsignal},
\posterior_{\stoptime^{\mathrm S}}^{\lsignal}\bigr),
\]
and consider a binary decision problem with actions $a_0$ and $a_1$ such that $a_0$ is optimal if and only if the posterior belief that $\theta=1$ is at most $q$. For example, one can take payoffs
\[
v(a_1,0)=v(a_1,1)=0,\qquad
v(a_0,0)=1,\qquad
v(a_0,1)=-\frac{1-q}{q}.
\]
Since $q<\posterior_{\stoptime^{\mathrm S}}^{\lsignal}<\prior$, the decision maker chooses $a_1$ without additional information. Under any scoring rule, the posterior never falls below $\posterior_{\stoptime^{\mathrm S}}^{\lsignal}$, so summarized advice never changes the decision and has value $0$. Under unrestricted advice, however, a left-biased signal arriving at time $\stoptime^{\mathrm S}+1$ occurs with strictly positive probability, and on that event, the posterior falls below $q$. The decision maker then switches to $a_0$, so unrestricted advice has strictly positive value.

Thus, the additional value of dynamic contracts depends on the decision environment. Under slow drift, the gap between $\posterior_{\stoptime^{\mathrm S}}^{\lsignal}$ and $\posterior_{\stoptime^{\mathrm S}+1}^{\lsignal}$ can be made arbitrarily small. However, a slight increase in informativeness can change the value of information from $0$ to strictly positive. In this sense, the multiplicative gain from unrestricted advice can be unbounded.

\subsection{Final Remarks}
\label{sub:conclusion}
We have articulated how dynamic rewards can expand the set of implementable strategies in a simple yet fundamentally dynamic information acquisition problem. The economic importance of contracting for information acquisition is self-evident, and most natural stories for why information acquisition is costly involve some dynamic element. Our goal has been to take such dynamics seriously, for learning technologies with a natural interpretation in terms of forecasters seeking a particular sought-after piece of falsifiable evidence, and under a class of contracts in line with past work on forecaster incentive provision.

We have shown that whether the decision maker benefits from a contracting environment that facilitates dynamic reporting depends on the nature of the dynamic learning process. Along the way, we discussed how the relevant properties of the information acquisition technology have natural interpretations in various settings of practical interest. As our focus is on contracting under a general class of mechanisms, a fundamental difficulty underlying our exercise is the lack of any natural structure (e.g., stationarity) under an arbitrary dynamic contract. Such assumptions are often critical in similar settings. Despite this fundamental challenge, we provided simple, economically meaningful conditions under which maximum effort is implementable by a scoring rule and identified conditions under which dynamic reporting is strictly necessary.

There are many natural avenues for future work. Empirically, our results show how learning technologies influence the benefits of time variation to rewards. While we focus on a simple forecasting problem, such questions may be of interest when eliciting other kinds of information beyond a prediction of a future event---for example, if the learning process itself is privately known by the expert \citep{chambers2021dynamic}.   More broadly, we view questions regarding whether or not simple contracts are limited in power relative to dynamic mechanisms as a worthwhile agenda overall. Any further insights on these questions would prove valuable toward understanding how dynamics influence mechanism design for information acquisition, for both theory and practice. 
\bibliographystyle{apalike}
\bibliography{ref}

% \newpage

\appendix
\section{Effort-Maximizing Dynamic Contracts}
\label{apx:optimal_dynamic}
\begin{proof}[Proof of \cref{thm:optimal_dynamic_menu}]
By \cref{lem:menu}, it is without loss to work with a shrinking menu representation. 
Because menus are shrinking, any reward used on path at date $t$ is also available at every earlier date, so it must satisfy
\begin{equation}
\util(\posterior_{t'}^\no,\reward)\leq \util_{t'}^\no,
\qquad \forall t'\leq t. \label{eq:proof_dynamic_ic}
\end{equation}

For each $t$, let $\hat{\reward}_t^\no$ maximize $\reward_0$ subject to
\[
\reward = (r_0,r_1) \in [0,1]^2
\qquad\text{and}\qquad
\util(\posterior_t^\no,\reward)=\util_t^\no.
\]
Then $\hat{\reward}_{t,0}^\no=1$ whenever $\hat{\reward}_{t,1}^\no>0$, so each canonical reward is either of the form $(1,b_t)$ or $(a_t,0)$.
Moreover, $\hat{\reward}_t^\no$ is obtained from $\reward_t^\no$ by shifting reward toward state $0$ while keeping utility at $\posterior_t^\no$ fixed, hence
\[
\util(\posterior,\hat{\reward}_t^\no)\leq \util(\posterior,\reward_t^\no)
\qquad \forall \posterior\geq \posterior_t^\no.
\]
Therefore, adding $\hat{\reward}_t^\no$ to all earlier menus does not increase the stopping utility of any earlier no-information type. 
Doing this for every $t$ preserves effort maximization, so without loss we may take the no-information rewards themselves to be canonical.

Now fix $t\leq t'$. Since $\reward_{t'}^\no\in \rewardMenu_t$, 
\[
\util(\posterior_t^\no,\reward_t^\no)\geq \util(\posterior_t^\no,\reward_{t'}^\no). \tag{A.1}
\]
Because both rewards are canonical, (A.1) implies $\reward_{t'}^\no\leq \reward_t^\no$ coordinatewise:
if $\reward_t^\no=(1,b_t)$, then any later canonical reward is either $(1,b_{t'})$ with $b_{t'}\leq b_t$ or $(a_{t'},0)$;
if $\reward_t^\no=(a_t,0)$, then (A.1) rules out any later reward of the form $(1,b)$, so necessarily $\reward_{t'}^\no=(a_{t'},0)$ with $a_{t'}\leq a_t$.
This proves Part 1.

Next fix $t\leq \stoptime_{\rewardMenu}$ and $\signal\in\signals$.
By \eqref{eq:proof_dynamic_ic}, any reward used after signal $\signal$ at date $t$ must satisfy
\[
\util(\posterior_{t'}^\no,\reward)\leq \util_{t'}^\no,
\qquad \forall t'\in\{0,\dots,t\}.
\]
Conversely, adding at date $t$ any reward satisfying these inequalities does not increase the utility from stopping at any earlier no-information belief.
Hence, without loss, $\reward_t^\signal$ solves
\[
\max_{\reward:\states\to[0,1]} \util(\posterior_t^\signal,\reward)
\qquad
\text{s.t.}\qquad
\util(\posterior_{t'}^\no,\reward)\leq \util_{t'}^\no
\quad \forall t'\in\{0,\dots,t\}. \tag{A.2}
\]

If $\signal$ is left-biased, then $\posterior_t^\signal\leq \posterior_t^\no$.
In the relaxed version of (A.2) that keeps only the constraint at $t$, the objective is maximized by shifting as much reward as possible toward state $0$.
By construction, this is exactly the canonical reward $\reward_t^\no$.
Since $\reward_t^\no$ also satisfies all earlier constraints, it solves (A.2), proving Part 2.

If $\signal$ is right-biased, then by definition of $G_t$,
\[
\util(\posterior_t^\signal,\reward_t^\signal)=G_t(\posterior_t^\signal).
\]
Each feasible reward induces an affine function $\posterior\mapsto \util(\posterior,\reward)$ lying below the no-information points $\{(\posterior_{t'}^\no,\util_{t'}^\no)\}_{t'=0}^t$.
Hence $G_t$ is the pointwise supremum of such affine functions, i.e., the greatest convex minorant generated by feasible rewards bounded in $[0,1]^2$.
Since the affine function induced by $\reward_t^\signal$ attains this supremum at $\posterior_t^\signal$, it is a supporting line to $G_t$ at that belief.
This proves Part 3.
\end{proof}

\section{Effort-Maximizing Contracts as Scoring Rules}
\label{apx:proofs_static_is_opt}

\subsection{Stationary Environment}
\begin{proof}[Proof of \cref{thm:stationary}]
For any contract $\rewardMenu$ with stopping time $\stoptime_{\rewardMenu}\in[0,T]$,
to show that there exists a static scoring rule $\score$
such that the agent has incentive to exert effort at least until time $\stoptime_{\rewardMenu}$ given static scoring rule $\score$, 
it is sufficient to show that there exists a static scoring rule $\score$
such that the agent has an incentive to exert effort at any continuation game $\cgame_t$ for any $t\in[0,\stoptime_{\rewardMenu}]$.

First note that to maximize the expected score difference for the continuation game at any time~$t$, 
it is sufficient to consider a static scoring rule. 
This is because at any time $t'\in[t,\stoptime_{\rewardMenu}]$, 
we can allow the agent to pick any menu option from time $t$ to $\stoptime_{\rewardMenu}$. 
This leads to a static scoring rule 
where the agent's expected utility at time $t$ for stopping effort immediately is not affected but the continuation utility weakly increases. 
Finally, by \cref{prop:static_moral_hazard}, the effort-maximizing static scoring rule that maximizes the expected score difference is the V-shaped scoring rule $\score$ with a kink at prior~$\prior$. 
Since the optimal scoring rule remains the same for all continuation games in stationary environments, this scoring rule implements the optimal dynamic efforts. 
\end{proof}

\subsection{Perfect-learning Environment}
\label{apx:perfect}

\begin{proof}[Proof of \cref{thm:perfect}]
As alluded to in the main text, the proof follows two steps: 

\medskip

\noindent \textbf{Step One}: We first show that, for any prior $\prior$ and any signal arrival probabilities $\lambda$, 
there exists an effort-maximizing contract $\rewardMenu$ with a sequence of menu options $\{\reward^{\signal}_t\}_{t\leq \stoptime_{\rewardMenu},\signal\in\signals}\cup\{\reward^{\no}_{\stoptime_{\rewardMenu}}\}$
such that $\reward^{\lsignal}_t = \reward^{\no}_{\stoptime_{\rewardMenu}} = (1,0)$
for any $t\leq \stoptime_{\rewardMenu}$.

For any contract $\rewardMenu$, by applying the menu representation in \cref{lem:menu},
let $\{\reward^{\signal}_t\}_{t\leq \stoptime_{\rewardMenu},\signal\in\signals}$
be the set of menu options for receiving Poisson signals
and let $\reward^{\no}_{\stoptime_{\rewardMenu}}=(z_0,z_1)$ be the menu option for not receiving any Poisson signal before the stopping time $\stoptime_{\rewardMenu}$.
Note that in the perfect-learning environment, it is without loss to assume that $\reward^{\lsignal}_{t,1} = 0$ for any $t\leq \stoptime_{\rewardMenu}$ since the posterior probability of state 1 is 0 after receiving a Poisson signal $\lsignal$. 
Now consider another contract $\widehat{\rewardMenu}$ with menu options
$\{\hat{\reward}^{\signal}_t\}_{t\leq \stoptime_{\rewardMenu},\signal\in\signals}$
and $(\hat{z}_0,\hat{z}_1)$, 
where 
\begin{align}
(\hat{z}_0,\hat{z}_1) = \argmax_{z,z'\in[0,1]} \, z 
\quad\text{s.t.}\quad 
\posterior^{\no}_{\stoptime_{\rewardMenu}} z' + (1-\posterior^{\no}_{\stoptime_{\rewardMenu}}) z = \util^{\no}_{\stoptime_{\rewardMenu}},
\label{eq:program for z}
\end{align}
and for any time $t \leq \stoptime_{\rewardMenu}$ and any signal $\signal\in\signals$, 
\begin{align*}
\hat{\reward}^{\signal}_t = \begin{cases}
\reward^{\signal}_t & \util(\posterior^{\signal}_t, \reward^{\signal}_t) \geq \util(\posterior^{\signal}_t, (\hat{z}_0,\hat{z}_1)) \\
(\hat{z}_0,\hat{z}_1) & \text{otherwise}.
\end{cases}
\end{align*}
Essentially, contract $\widehat{\rewardMenu}$ adjusts the reward function for no information belief $\posterior^{\no}_{\stoptime_{\rewardMenu}}$
such that the reward for state being $0$ weakly increases,
the reward for state being $1$ weakly decreases, 
and the expected reward remains unchanged. 
Moreover, at any time $t\leq \stoptime_{\rewardMenu}$, 
contract $\widehat{\rewardMenu}$ allows the agent to optionally choose the additional option of $(\hat{z}_0,\hat{z}_1)$ to maximize his expected payoff for receiving an informative signal at time $t$.

It is easy to verify that for any signal $\signal\in\signals$ and any time $t\leq \stoptime_{\rewardMenu}$, 
the expected utility of the agent for receiving an informative signal $\signal$ is weakly higher, 
and hence, at any time $t\leq \stoptime_{\rewardMenu}$, the continuation payoff of the agent for exerting effort until time $\stoptime_{\rewardMenu}$ weakly increases in contract $\widehat{\rewardMenu}$. 
Moreover, at any time $t$, the expected reward of the agent for not exerting effort at time $t$ with belief $\posterior^{\no}_{t-1}$
satisfies $\hat{\util}^{\no}_{t-1} \leq \util^{\no}_{t-1}$. 
This is because, by our construction, at any time $t\leq \stoptime_{\rewardMenu}$, 
fewer options are available to the agent in contract~$\widehat{\rewardMenu}$, except for the additional option of $(\hat{z}_0,\hat{z}_1)$. 
However, $\util(\posterior^{\no}_{t-1}, (z_0,z_1)) \geq \util(\posterior^{\no}_{t-1}, (\hat{z}_0,\hat{z}_1))$
since $\posterior^{\no}_{t-1} \geq \posterior^{\no}_{\stoptime_{\rewardMenu}}$
and both options $(z_0,z_1)$ and $(\hat{z}_0,\hat{z}_1)$ give the same expected reward for the posterior belief of $\posterior^{\no}_{\stoptime_{\rewardMenu}}$. 
Combining both observations, we have $\stoptime_{\widehat{\rewardMenu}} \geq \stoptime_{\rewardMenu}$ and $\widehat{\rewardMenu}$ is also an effort-maximizing contract.

Note that in optimization program \eqref{eq:program for z}, 
it is easy to verify that $\hat{z}_1 = 0$ if $\hat{z}_0 < 1$. 
If $\hat{z}_0 = 1$, in this case, at any time $t$,
by the incentive constraint of the agent for any belief $\posterior^{\lsignal}_t$, 
we must have $\hat{\reward}^{\lsignal}_{t, 0} =1$ as well. 
Therefore, the agent receives the maximum reward of~1 whenever he receives a left-biased signal. 
In this case, incentive compatibility implies that $\hat{z}_1\leq \reward^{\rsignal}_{t,1}$ for all $t\leq \stoptime_{\rewardMenu}$.
Therefore, we can also decrease $\hat{z}_1$ and $\reward^{\rsignal}_{t,1}$ for all $t\leq \stoptime_{\rewardMenu}$ 
by $\hat{z}_1$, which does not affect the agent's incentive for effort
and hence the effort-maximizing contract satisfies that $\reward^{\lsignal}_t = \reward^{\no}_{\stoptime_{\rewardMenu}} = (1,0)$
for any $t\leq \stoptime_{\rewardMenu}$. 

Next, we will focus on the case when $\hat{z}_0 < 1$ and hence $\hat{z}_1 = 0$. 
Now consider another contract $\bar{\rewardMenu}$ with menu options
$\{\bar{\reward}^{\signal}_t\}_{t\leq \stoptime_{\widehat{\rewardMenu}},\signal\in\signals}$
and $(\bar{z}_0,\bar{z}_1)$, 
where $(\bar{z}_0,\bar{z}_1) = (1, 0)$
and for any time $t \leq \stoptime_{\rewardMenu}$, $\bar{\reward}^{\rsignal}_t = \hat{\reward}^{\rsignal}_t$
and $\bar{\reward}^{\lsignal}_t = (1, 0)$. We show that this weakly improves the agent's incentive to exert effort until time $\stoptime_{\widehat{\rewardMenu}}$ for any $t\leq \stoptime_{\widehat{\rewardMenu}}$. 
Specifically, for any $t\leq \stoptime_{\widehat{\rewardMenu}}$, the increase in the no information payoff for the continuation game $\cgame_t$ is 
\begin{align*}
\bar{\util}^{\no}_{t-1} - \hat{\util}^{\no}_{t-1} 
\leq (1-\posterior^{\no}_{t-1})(1-\hat{\reward}^{\lsignal}_{t-1, 0}).
\end{align*}
This is because in contract $\bar{\rewardMenu}$, 
either the agent prefers the menu option $\reward^{\rsignal}_{t'}$ for some $t'\geq t-1$, in which case the reward difference is 0, 
or the agent prefers the menu option $(1,0)$, 
in which case the reward difference is at most $(1-\posterior^{\no}_{t-1})(1-\hat{\reward}^{\lsignal}_{t-1, 0})$
since one feasible option for the agent in contract $\widehat{\rewardMenu}$ is $\hat{\reward}^{\lsignal}_{t-1}$ with an expected reward at least $(1-\posterior^{\no}_{t-1})\hat{\reward}^{\lsignal}_{t-1, 0}$.
Moreover, for any time $t\leq \stoptime_{\widehat{\rewardMenu}}$, the increase in continuation payoff for exerting effort from $t$ until $\stoptime_{\widehat{\rewardMenu}}$ is at least 
$(1-\posterior^{\no}_{t-1})(1-\hat{\reward}^{\lsignal}_{t, 0})$.
This is because incentive constraints imply that $\hat{\reward}^{\lsignal}_{t, 0}$ must decrease as $t$ increases, 
which is due to the fact that in the perfect-learning environment, the posterior belief $\posterior^{\lsignal}_{t}$ assigns a probability of 1 to the state being 0 at any time $t$.
Therefore, when the state is $0$, the reward of the agent is deterministically $1$ in contract $\bar{\rewardMenu}$ 
and the reward of the agent is at most $\hat{\reward}^{\lsignal}_{t, 0}$ in contract $\widehat{\rewardMenu}$, 
implying that the difference in expected reward is at least $(1-\posterior^{\no}_{t-1})(1-\hat{\reward}^{\lsignal}_{t, 0})$.
Combining the above observations and observing that $(1-\posterior^{\no}_{t-1})(1-\hat{\reward}^{\lsignal}_{t, 0}) \geq (1-\posterior^{\no}_{t-1})(1-\hat{\reward}^{\lsignal}_{t-1, 0})$,
we have $\stoptime_{\bar{\rewardMenu}}\geq\stoptime_{\widehat{\rewardMenu}}$, 
and hence $\bar{\rewardMenu}$ is also effort-maximizing.

\medskip 

\noindent \textbf{Step Two}: We now show that we can replace all other menu options except $(1,0)$ with a \emph{single} menu option that the agent can select at any time. 
By the previous step, it suffices to focus on contract $\rewardMenu$ with a sequence of menu options $\{\reward^{\signal}_t\}_{t\leq \stoptime_{\rewardMenu},\signal\in\signals}\cup\{\reward^{\no}_{\stoptime_{\rewardMenu}}\}$
such that $\reward^{\lsignal}_t = \reward^{\no}_{\stoptime_{\rewardMenu}} = (1,0)$
for any $t\leq \stoptime_{\rewardMenu}$. 
In addition, since signals are perfectly revealing, 
it is without loss to assume that $\reward^{\rsignal}_{t,0} = 0$
and the incentive constraints imply that $\reward^{\rsignal}_{t,1}$
is weakly decreasing in $t$. 

Let $\hat{t}\in[0,\stoptime_{\rewardMenu}]$ be the minimum time such that an agent with no information belief $\posterior^{\no}_{\hat{t}}$ weakly prefers menu option $\reward^{\no}_{\stoptime_{\rewardMenu}}=(1,0)$ compared to $\reward^{\rsignal}_{\hat{t}}$. 
Since both $\posterior^{\no}_t$ and $\reward^{\rsignal}_{t,1}$ are weakly decreasing in~$t$, 
an agent with posterior belief $\posterior^{\no}_t$ weakly prefers 
$\reward^{\rsignal}_{t'}$ compared to $\reward^{\no}_{\stoptime_{\rewardMenu}}$ for any $t,t'< \hat{t}$,
and weakly prefers~$\reward^{\no}_{\stoptime_{\rewardMenu}}$ 
compared to $\reward^{\rsignal}_{t'}$ for any $t,t'\geq \hat{t}$. 
Now consider another contract $\widehat{\rewardMenu}$ that offers only two menu options, 
$\reward^{\no}_{\stoptime_{\rewardMenu}} = (1,0)$ and $\reward^{\rsignal}_{\hat{t}}$, at every time~$t\leq \stoptime_{\rewardMenu}$. 
Contract $\widehat{\rewardMenu}$ can be implemented as a V-shaped scoring rule with parameters $r_0=1$ and $r_1 = \reward^{\rsignal}_{\hat{t},1} \in[0,1]$. Moreover, at any time $t\leq \stoptime_{\rewardMenu}$,
\begin{itemize}
\item if $t \geq \hat{t}+1$, in the continuation game $\cgame_{t,\stoptime_{\rewardMenu}}$, the agent's utility for not exerting effort is the same in both contract $\rewardMenu$ and $\widehat{\rewardMenu}$ because the agent with no information belief $\posterior^{\no}_{t-1}$ will choose the same menu option $\reward^{\no}_{\stoptime_{\rewardMenu}}$. 
However, the agent's utility for exerting effort is weakly higher in contract $\widehat{\rewardMenu}$ since the reward $\reward^{\rsignal}_{t,1}$ from receiving a right-biased signal at time $t$ weakly decreases in $t$. 

\item if $t\leq \hat{t}$, in the continuation game $\cgame_{t,\stoptime_{\rewardMenu}}$, 
by changing the contract from $\rewardMenu$ to $\widehat{\rewardMenu}$, 
the decrease in agent's utility for not exerting effort is exactly $\posterior^{\no}_{t-1}(\reward^{\rsignal}_{t-1, 1} - \reward^{\rsignal}_{\hat{t}, 1})$
by changing the menu option for no information belief $\posterior^{\no}_{t-1}$ from $\reward^{\rsignal}_{t-1}$ to~$\reward^{\rsignal}_{\hat{t}}$. 
However, the decrease in the agent's utility for exerting effort in $\cgame_{t,\stoptime_{\rewardMenu}}$ is at most $\posterior^{\no}_{t-1}(\reward^{\rsignal}_{t, 1} - \reward^{\rsignal}_{\hat{t}, 1})$ since the decrease in reward for receiving a right-biased signal $\rsignal$ is at most $\reward^{\rsignal}_{t, 1} - \reward^{\rsignal}_{\hat{t}, 1}$ and it only occurs when the state is $1$. 
\end{itemize}
Therefore, given contract $\widehat{\rewardMenu}$, the agent has stronger incentives to exert effort in all continuation games $\cgame_{t,\stoptime_{\rewardMenu}}$ with $t\leq \stoptime_{\rewardMenu}$, 
which implies that $\stoptime_{\widehat{\rewardMenu}} \geq \stoptime_{\rewardMenu}$
and hence contract $\widehat{\rewardMenu}$ is also effort-maximizing.
\end{proof}

\begin{proof}[Proof of \cref{lem:monotone_in_r1}]
Note that it is easy to verify that if there exists a belief such that the agent is incentivized to exert effort,
the intersection belief $\underline{\posterior}(r_1)$ is such that the agent with belief $\underline{\posterior}(r_1)$ would prefer menu option $(1,0)$ to $(0, r_1)$
and $\bar{\posterior}(r_1)$ is such that the agent with belief $\bar{\posterior}(r_1)$ would prefer menu option $(0, r_1)$ to $(1,0)$.

Consider the case of decreasing the reward parameter from $r_1 = z$ to $r_1 = z'$ for $0\leq z'< z\leq 1$. 
The agent's utility for not exerting effort given menu option $(1,0)$ remains unchanged, 
but the agent's utility for exerting effort in at least one period decreases. 
Therefore, $\underline{\posterior}(r_1)$ weakly increases. 
Moreover, given posterior belief $\posterior^{\no}_{t-1}$, 
the agent's utility for not exerting effort given menu option $(0, r_1)$ decreases by $\posterior^{\no}_{t-1}(z-z')$, 
while the agent's utility for exerting effort in at least one period decreases by at most $\posterior^{\no}_{t-1}(z-z')$ since the reward decrease can only occur when the state is $1$. 
Therefore, $\bar{\posterior}(r_1)$ also weakly increases.
\end{proof}

\subsection{Single-signal Environment}
\label{apx:single_source}

\begin{proof}[Proof of \cref{thm:single_source}]
As mentioned in the main text, the proof proceeds in two steps: 

\medskip

\noindent \textbf{Step One:} We first show that an effort-maximizing contract exists with the no-information utility~$\util^{\no}_t$ convex in $\posterior^{\no}_t$. For any contract $\rewardMenu$, let $\underline{\util}_{t}(\posterior)$ be the convex hull of the no information payoff $\util^{\no}_{t'}$ for $t'\leq t$
by viewing $\util^{\no}_{t'}$ as a function of $\posterior^{\no}_{t'}$. 
Consider an effort-maximizing contract $\rewardMenu$ with the following selection:
\begin{enumerate}
\item maximizes the time $\bar{t}$ such that $\underline{\util}_{\bar{t}-1}(\posterior^{\no}_{t}) = \util^{\no}_{t}$ for any time $t\leq \bar{t}-1$;
\item conditional on maximizing $\bar{t}$, selecting the one that maximizes the weighted average no information payoff after time $\bar{t}$, i.e., 
$\sum_{i\geq 0} e^{-i} \cdot \util^{\no}_{\bar{t}+i}$. \label{criterion:util}
\end{enumerate}
Note that the exponential weight $e^{-i}$ is purely for the tie breaking selection and ensures a finite sum of the no information payoff. 
The existence of an effort-maximizing contract given such a selection rule can be shown using standard arguments since, recalling that we have a discrete-time model, the set of effort-maximizing contracts that satisfy the first criterion is compact 
and the objective in the second selection criterion is continuous. 
Let~$\stoptime_{\rewardMenu}$ be the stopping time of the agent for contract $\rewardMenu$. 
We will show that $\bar{t} = \stoptime_{\rewardMenu}+1$. 

First, we observe that it cannot be the case that $\bar{t}=\stoptime_{\rewardMenu}$. 
This is because the agent does not have incentives to exert effort after time $\stoptime_{\rewardMenu}$. 
By increasing the agent's utility $\util^{\no}_{\bar{t}}$ in this case, we can restore the convexity of $\util^{\no}_t$ at $t=\bar{t}$ without violating the effort incentives. 

Now, suppose by contradiction we have $\bar{t} \leq \stoptime_{\rewardMenu} - 1$. 
At any time $t\leq \stoptime_{\rewardMenu}$, 
recall that $\cgame_{t}$ is the continuation game at time $t$ with prior belief $\posterior^{\no}_{t-1}$
such that the agent's utility for not exerting effort is $\util^{\no}_{t-1}$
and the agent's utility for exerting effort in $\cgame_{t}$
is 
\begin{align*}
U_t \triangleq \rbr{1-\sum_{\signal\in\signals}F^{\signal}_t(\stoptime_{\rewardMenu})}
\cdot u(\posterior^{\no}_{\stoptime_{\rewardMenu}}, \reward^{\no}_{\stoptime_{\rewardMenu}})
+ \sum_{t'=t}^{\stoptime_{\rewardMenu}} \sum_{\signal\in\signals} f^{\signal}_t(t') \cdot u(\posterior^{\signal}_{t'}, \reward^{\signal}_{t'}).
\end{align*}
Note that $U_{t}\geq\util_{t-1}^{\no}$ for any $t\leq \stoptime_{\rewardMenu}$. 
We first show that the equality must hold at time $\bar{t}+1$,
i.e., $U_{\bar{t}+1}=\util_{\bar{t}}^{\no}$. 
Let $\bar{u}_t(\posterior)$ be the upper bound on the expected reward at any belief $\posterior$ at time $t$ given that no Poisson signal has arrived before $t$.
Specifically, 
\begin{align*}
\bar{u}_t(\posterior) = \max_{z_0,z_1\in[0,1]} \posterior z_1 + (1-\posterior)z_0
\quad\text{s.t.}\quad
\posterior^{\no}_{t'} z_1 + (1-\posterior^{\no}_{t'})z_0\leq u^{\no}_{t'}, 
\forall t'\leq t.
\end{align*}
It is easy to verify that function $\bar{u}_t(\posterior)$ is convex in $\posterior$ for all $t$
and $\bar{u}_{t}(\posterior^{\no}_{t'})$ is an upper bound on the no information utility $u^{\no}_{t'}$ for all $t' > t$.
Moreover, for any $\posterior \leq \posterior^{\no}_{t}$, 
$\bar{u}_t(\posterior)$ is a linear function in $\posterior$. 
The reward function $\bar{u}_t(\posterior)$ for $t=\bar{t}-1$
and $\posterior \leq \posterior^{\no}_{\bar{t}}$
is illustrated in \Cref{fig:overlapping_belief_discrete} as the blue straight line. 

Since the no information utility is not convex at time $t=\bar{t}$, 
we have $\bar{u}_{\bar{t}-1}(\posterior^{\no}_{\bar{t}}) > u^{\no}_{\bar{t}}$.
In this case, if $U_{\bar{t}+1}>\util_{\bar{t}}^{\no}$, 
by increasing $\util_{\bar{t}}^{\no}$
to $\min\{U_{\bar{t}+1}, \bar{u}_{\bar{t}-1}(\posterior^{\no}_{\bar{t}})\}$, 
the agent's incentive for exerting effort is not violated. 
Moreover, selection rule \eqref{criterion:util} of maximizing the no information utility after time $\bar{t}$ is violated, a contradiction. 
Therefore, we can focus on the situation where the agent's incentive for exerting effort at time $\bar{t}+1$ is binding.

\begin{figure}
\centering
\begin{tikzpicture}[xscale=1,yscale=1]
\hspace{-10pt}
\draw [<->] (0,3.1) -- (0,0) -- (7.3,0);

\draw [thick] plot [smooth, tension=0.7] coordinates { (1, 2) (4, 1.5) (6, 3)};

\draw [dashed] (0.5, 1) -- (4, 1.5);
\draw [dotted] (0.5, 0) -- (0.5, 1);
\draw (0.4, -0.4) node {$\mu^{\no}_{\bar{t}}$};

\draw [dotted] (4, 0) -- (4, 1.5);
\draw (4, -0.4) node {$\mu^{\no}_{t^*}$};

\draw [dotted] (2.5, 0) -- (2.5, 1.6);
\draw (2.5, -0.4) node {$\mu^{\rsignal}_{\bar{t}}$};
\draw [red] (0.5, 1) -- (2.5, 1.6);

\draw [dotted] (1, 0) -- (1, 2);
\draw (1.1, -0.4) node {$\mu^{\no}_{\bar{t}-1}$};

\draw (7, -0.3) node {$\mu$};

\draw [blue] (0, 2.32) -- (1, 2);
\draw [dotted] (1, 2) -- (7, 0);

\end{tikzpicture}
\caption{\footnotesize This figure illustrates the case when $\posterior^{\rsignal}_{\bar{t}} \leq \posterior^{\no}_{t^*}$.
The black curve is the function $\underline{\util}_{\bar{t}}(\posterior)$,
blue line is the function $\bar{\util}(\posterior)$
and the red line is the function $y(\posterior)$. }
\label{fig:overlapping_belief_discrete}
\end{figure}

By the construction of $\rewardMenu$, 
there exists $t\leq \bar{t}$ such that $\underline{\util}_{\bar{t}}(\posterior^{\no}_t) < \util^{\no}_t$. 
Let $t^*$ be the maximum time such that $\underline{\util}_{\bar{t}}(\posterior^{\no}_{t^*})=\util^{\no}_{t^*}$.
That is, $\posterior^{\no}_{t^*}$ is the tangent point such that $u^{\no}_t$ coincides with its convex hull. 
See \Cref{fig:overlapping_belief_discrete} for an illustration.
We consider two cases separately. 
\begin{itemize}
\item $\posterior^{\no}_{t^*} \geq \posterior^{\rsignal}_{\bar{t}}$.
In this case, let $y(\posterior)$ be a linear function of posterior $\posterior$
such that $y(\posterior^{\no}_{\bar{t}}) = \util^{\no}_{\bar{t}}$
and $y(\posterior^{\rsignal}_{\bar{t}}) = \underline{\util}(\posterior^{\rsignal}_{\bar{t}})$. 
Function $y$ is illustrated in \Cref{fig:overlapping_belief_discrete} as the red line. 
Note that in this case, we have $y(\posterior^{\no}_{\bar{t}-1}) < \underline{\util}_{\bar{t}}(\posterior^{\no}_{\bar{t}-1})= \util^{\no}_{\bar{t}-1}$. 
Moreover, $y(\posterior^{\no}_{\bar{t}-1})$ is the maximum continuation payoff of the agent for exerting effort at time $\bar{t}$ given belief $\posterior^{\no}_{\bar{t}-1}$. 
This is because, by exerting effort, either the agent receives a Poisson signal $\rsignal$ at time $\bar{t}$, which leads to posterior belief $\posterior^{\rsignal}_{\bar{t}}$ with expected payoff 
$\underline{\util}(\posterior^{\rsignal}_{\bar{t}}) = y(\posterior^{\rsignal}_{\bar{t}})$, 
or the agent does not receive a Poisson signal, 
which leads to belief drift to $\posterior^{\no}_{\bar{t}}$, 
with the optimal continuation payoff being $U_{\bar{t}+1} = \util^{\no}_{\bar{t}} = y(\posterior^{\no}_{\bar{t}})$. 
However, $y(\posterior^{\no}_{\bar{t}-1}) < \util^{\no}_{\bar{t}-1}$
implies that the agent has a strict incentive not to exert effort at time $\bar{t}$, a contradiction.

\item $\posterior^{\no}_{t^*} < \posterior^{\rsignal}_{\bar{t}}$. 
In this case, consider another contract $\bar{\rewardMenu}$ 
such that the no information utility in contract $\bar{\rewardMenu}$
is $\util^{\no}_{t; \bar{\rewardMenu}} = \underline{\util}(\posterior^{\no}_t)$ for any $t\leq \bar{t}$. 
Note that in contract $\bar{\rewardMenu}$, the expected reward of the agent at any time $t$ for receiving a Poisson signal is the same as in contract $\rewardMenu$, 
while the expected reward for stopping when not receiving Poisson signals weakly decreases. 
Therefore, contract $\bar{\rewardMenu}$ is also an effort-maximizing contract. 
However, the time such that the no information payoff is a convex function is strictly larger in $\bar{\rewardMenu}$, contradicting our selection rule for $\rewardMenu$. 
\end{itemize}

Therefore, we have $\bar{t} = \stoptime_{\rewardMenu}+1$ and the no information utility of the agent is a convex function. 

\medskip
\noindent \textbf{Step Two:} We now show that we can ``flatten'' rewards to strengthen the agent's incentives to exert effort in case the reward constraint is violated. By Step One, there exists a contract $\rewardMenu$ with a sequence of menu options $\{\reward^s_t\}_{s\in S,t\leq \stoptime_{\rewardMenu}}\cup\{\reward^{\no}_{\stoptime_{\rewardMenu}}\}$ in which the no information payoff is convex in the no information belief. 
If $\posterior^{\no}_t < u^{\no}_t$ for all $t\leq \stoptime_{\rewardMenu}$, 
let $\hat{z}_1 = 1$ 
and let $\hat{z}_0\leq 1$ be the maximum reward such that 
$\posterior^{\no}_t + \hat{z}_0(1-\posterior^{\no}_t) \leq u^{\no}_t$ for all $t\leq \stoptime_{\rewardMenu}$. 
Otherwise, let $\hat{z}_0=0$ and let $\hat{z}_1\leq 1$ be the maximum reward such that $\hat{z}_1\cdot \posterior^{\no}_t \leq u^{\no}_t$ for all $t\leq \stoptime_{\rewardMenu}$.
Essentially, the straight line $(\hat{z}_0,\hat{z}_1)$ is tangent to the agent's utility curve for not receiving Poisson signals subject to the reward bound. 
Let $\hat{t}$ be the time corresponding to the rightmost tangent point.
See \Cref{fig:convex_util_belief} for an illustration. 

\begin{figure}
\centering
\begin{tikzpicture}[xscale=0.9,yscale=0.9]
\draw [<->] (0,4.2) -- (0,0) -- (8.2,0);

\draw [thick] (0.5,1.9375) -- (1, 1.875);
\draw [thick] plot [smooth, tension=0.5] coordinates { (1, 1.875) (1.9,1.85) (3, 2.08) (3.6, 2.6) (4, 3.5)};

\draw [thick] plot [smooth, tension=0.7] coordinates { (5, 2) (5.7,3) (6.2,3.1) (7, 3.8) (7.5, 3.85)};

\draw [dotted] (6.2, 0) -- (6.2,3.1);
\draw (6.2, -0.4) node {$\posterior^{\rsignal}_{\hat{t}}$};

\draw [dotted] (2.9, 0) -- (2.9,2);
\draw (2.9, -0.4) node {$\posterior^{\no}_{\hat{t}}$};
\draw [dotted] (8, 0) -- (8,4);
\draw (8, -0.4) node {$1$};
\draw (0, -0.4) node {$0$};
\draw [red] (0,2) -- (2.14634,1.7317);
\draw [red] (2.14634,1.7317) -- (8, 4);

\draw [red, dashed] (2.14634,1.7317) -- (8, 1);
\draw [red, dashed] (0,0.9) -- (2.14634,1.7317);

\draw [dotted] (0, 1) -- (8,1);
\draw (-0.5, 1.2) node {$r^{\no}_{\stoptime_{\contract},1}$};
\draw (-0.5, 2) node {$r^{\no}_{\stoptime_{\contract},0}$};
\draw [dotted] (0, 4) -- (8,4);
\draw (-0.3, 4) node {$\hat{z}_1$};
\draw (-0.3, 0.7) node {$\hat{z}_0$};
\end{tikzpicture}
\caption{\footnotesize The black solid curves are the agent's expected utilities for not exerting effort as a function of his belief at any time $t$. The left curve is the expected utility for beliefs without receiving Poisson signals, and the right curve is the one receiving the Poisson signal $\rsignal$. 
}
\label{fig:convex_util_belief}
\end{figure}

Let $\underline{u}(\posterior)$ be the function that coincides with $u^{\no}_t$ for $\posterior \leq \posterior^{\no}_{\hat{t}}$
and $\underline{u}(\posterior) = (\hat{z}_1-\hat{z}_0)\posterior + \hat{z}_0$. 
Note that $\underline{u}$ is convex. 
Consider another contract $\widehat{\rewardMenu}$ that is implemented by the scoring rule 
$\score(\posterior,\state) = \underline{u}(\posterior) + \xi(\posterior)(\state - \posterior)$ for all $\posterior\in [0,1]$ and $\state\in\{0,1\}$, 
where $\xi(\posterior)$ is a subgradient of~$\underline{u}$. 
It is easy to verify that the implemented scoring rule satisfies the bounded constraint on rewards. 
Next, we show that $\stoptime_{\widehat{\rewardMenu}} \geq \stoptime_\rewardMenu$ and hence contract $\widehat{\rewardMenu}$ must also be effort-maximizing, which concludes the proof of \cref{thm:single_source}. 

In any continuation game $\cgame_{t}$, recall that $u^{\no}_{t-1}$ is the utility of the agent for not exerting effort 
and $U_t$ is the utility of the agent for exerting effort
given contract $\rewardMenu$. 
For any time $t\leq \stoptime_{\rewardMenu}$, the agent has incentive to exert effort at time $t$ given contract $\rewardMenu$
implies that $u^{\no}_{t-1} \leq U_t$. 
Given contract $\widehat{\rewardMenu}$, we similarly define $\hat{u}^{\no}_{t-1}$ and $\hat{U}_t$
and show that for any time $t\leq \stoptime_{\rewardMenu}$, 
$U_t - \hat{U}_t \leq u^{\no}_{t-1}- \hat{u}^{\no}_{t-1}$.
This immediately implies that the agent also has incentive to exert effort at any time $t\leq \stoptime_{\rewardMenu}$ given contract $\widehat{\rewardMenu}$
and hence $\stoptime_{\widehat{\rewardMenu}} \geq \stoptime_\rewardMenu$. 

Our analysis for showing that $U_t - \hat{U}_t \leq u^{\no}_{t-1} - \hat{u}^{\no}_{t-1}$ is divided into two cases. 
\begin{enumerate}[{Case} 1:]
\item $t\geq \hat{t}$. 
In this case, since $u^{\no}_{t-1}\geq \hat{u}^{\no}_{t-1}$ for any $t\leq \stoptime_{\rewardMenu}$ by the construction of contract~$\widehat{\rewardMenu}$, 
it is sufficient to show that $\hat{U}_t \geq U_t$
for any $t\in [\hat{t},\stoptime_{\rewardMenu}]$. 
We first show that for any $t\in [\hat{t},\stoptime_{\rewardMenu}]$, 
if $\posterior^{\rsignal}_t\leq \posterior^{\no}_{\hat{t}}$, 
we must have $\underline{u}(\posterior^{\rsignal}_t) \geq \util^{\rsignal}_t$ in order to satisfy the dynamic incentive constraint in contract $\rewardMenu$. 
Next, we focus on the case where $\posterior^{\rsignal}_t> \posterior^{\no}_{\hat{t}}$
and show that 
$\hat{u}^{\rsignal}_t = \posterior^{\rsignal}_t \hat{z}_1 + (1-\posterior^{\rsignal}_t)\hat{z}_0 \geq u^{\rsignal}_t$. 
We prove this by contradiction. 
Suppose that $u^{\rsignal}_t > \posterior^{\rsignal}_t \hat{z}_1 + (1-\posterior^{\rsignal}_t)\hat{z}_0$. 
Recall that $(\reward^{\rsignal}_{t,0}, \reward^{\rsignal}_{t,1})$ are the options offered to the agent at time $t$ that attain expected utility $u^{\rsignal}_t$ under belief $\posterior^{\rsignal}_t$.
Moreover, in our construction, either $\hat{z}_0 = 0$, or $\hat{z}_1=1$, or both equalities hold. 
Therefore, the bounded constraints $\reward^{\rsignal}_{t,0},\reward^{\rsignal}_{t,1}\in[0,1]$
and the fact that agent with belief $\posterior^{\rsignal}_t$
prefers $(\reward^{\rsignal}_{t,0},\reward^{\rsignal}_{t,1})$ over $(\hat{z}_0,\hat{z}_1)$
imply that $\reward^{\rsignal}_{t,0} \geq \hat{z}_0.$
% \begin{align*}
% \reward^{\rsignal}_{t,0} \geq \hat{z}_0.
% \end{align*}
\begin{figure}
\centering
\begin{tikzpicture}[xscale=0.9,yscale=0.9]
\draw [<->] (0,4.2) -- (0,0) -- (8.2,0);

\draw [thick] (0.5,1.9375) -- (1, 1.875);
\draw [thick] plot [smooth, tension=0.5] coordinates { (1, 1.875) (1.9,1.85) (3, 2.08) (3.6, 2.6) (4, 3.5)};

\draw [thick] plot [smooth, tension=0.7] coordinates { (5, 2) (5.7,3) (6.2,3.1) (7, 3.8) (7.5, 3.85)};

\draw [dotted] (6.2, 0) -- (6.2,3.45);
\draw [blue] (0, 2.2) -- (8,3.8);
\filldraw [blue] (6.2,3.45) circle (1pt);
\draw (6.2, -0.4) node {$\posterior^{\rsignal}_{\hat{t}}$};

\draw [dotted] (2.9, 0) -- (2.9,2);
\draw (2.9, -0.4) node {$\posterior^{\no}_{\hat{t}}$};
\draw [dotted] (8, 0) -- (8,4);
\draw (8, -0.4) node {$1$};
\draw (0, -0.4) node {$0$};
\draw [red] (0,2) -- (2.14634,1.7317);
\draw [red] (2.14634,1.7317) -- (8, 4);

\draw [red, dashed] (2.14634,1.7317) -- (8, 1);
\draw [red, dashed] (0,0.9) -- (2.14634,1.7317);

\draw [dotted] (0, 4) -- (8,4);
\draw (-0.3, 4) node {$\hat{z}_1$};
\draw (-0.3, 0.9) node {$\hat{z}_0$};
\end{tikzpicture}
\caption{\footnotesize The blue line is the expected utility of the agent for choosing the menu option $(\reward^{\rsignal}_{t,0},\reward^{\rsignal}_{t,1})$ if the utility at belief $\posterior^{\rsignal}_{\hat{t}}$ is higher than the red line.}
\label{fig:low_continuation_util}
\end{figure}
See \Cref{fig:low_continuation_util} for an illustration.
Since $\posterior^{\no}_{\hat{t}} < \posterior^{\rsignal}_t$, 
this implies that the agent's utility at belief $\posterior^{\no}_{\hat{t}}$ given option $(\reward^{\rsignal}_{t,0}, \reward^{\rsignal}_{t,1})$ is strictly larger than his utility under $(\hat{z}_0,\hat{z}_1)$, 
i.e., 
\begin{align*}
\posterior^{\no}_{\hat{t}} \reward^{\rsignal}_{t,1} + (1-\posterior^{\no}_{\hat{t}})\reward^{\rsignal}_{t,0}
> \posterior^{\no}_{\hat{t}} \hat{z}_1 + (1-\posterior^{\no}_{\hat{t}})\hat{z}_0
= u^{\no}_{\hat{t}}.
\end{align*}
However, option $(\reward^{\rsignal}_{t,0}, \reward^{\rsignal}_{t,1})$ is a feasible choice for the agent 
at time $\hat{t}$ in dynamic contract~$\rewardMenu$ since $t\geq \hat{t}$, 
which implies that $\posterior^{\no}_{\hat{t}} \reward^{\rsignal}_{t,1} + (1-\posterior^{\no}_{\hat{t}})\reward^{\rsignal}_{t,0}
\leq u^{\no}_{\hat{t}}$.
This leads to a contradiction.

Finally, for $t\in [\hat{t},\stoptime_{\rewardMenu}]$, conditional on the event that the informative signal did not arrive at any time before $t$, 
since the agent expected utility given contract~$\widehat{\rewardMenu}$ is weakly higher compared to contract~$\rewardMenu$ 
given any arrival time of the Poisson signal, 
taking the expectation we have $\hat{U}_t \geq U_t$.

\item $t < \hat{t}$. 
In this case, the continuation value for both stopping effort immediately and exerting effort until time $\stoptime_{\rewardMenu}$ weakly decreases. 
However, we will show that the expected decrease for stopping effort is weakly higher.
For any $t<\hat{t}$, let $\hat{\reward}_t$ be the reward such that $u^{\no}_t = \posterior^{\no}_t \hat{\reward}_t + (1-\posterior^{\no}_t) \hat{z}_0$.
Note that $\hat{\reward}_t \geq \hat{z}_1$ and it is possible that $\hat{\reward}_t \geq 1$. 
The construction of $\hat{\reward}_t$ is only used in the intermediate analysis, not in the constructed scoring rules. 
Let $\tilde{u}_t(\posterior) \triangleq \posterior \hat{\reward}_t + (1-\posterior) \hat{z}_0$ be the expected utility of the agent for choosing option $(\hat{z}_0, \hat{\reward}_t)$ given belief $\posterior$.
This is illustrated in \cref{fig:larger_diff}.

By construction, the expected utility decrease for not exerting effort in $\cgame_t$ is 
\begin{align*}
u^{\no}_{t-1} - \hat{u}^{\no}_{t-1} 
= \tilde{u}_{t-1}(\posterior^{\no}_{t-1}) - \hat{u}^{\no}_{t-1} 
= \posterior^{\no}_{t-1} (\hat{\reward}_{t-1} - \hat{z}_1).
\end{align*}
Next, observe that for any time $t'\in [t,\stoptime_{\rewardMenu}]$, $u^{\rsignal}_{t'} \leq \tilde{u}_t(\posterior^{\rsignal}_{t'})$. 
This argument is identical to the proof in Case 1, and hence omitted here. 
Therefore, the expected utility decrease for exerting effort until $\stoptime_{\rewardMenu}$ is 
\begin{align*}
U_t - \hat{U}_t 
&= \sum_{t'=t}^{\stoptime_{\rewardMenu}}
\rbr{u^{\rsignal}_{t'} - \hat{u}^{\rsignal}_{t'}} \cdot f^{\rsignal}_t(t')
\leq \sum_{t'=t}^{\stoptime_{\rewardMenu}}\rbr{\tilde{u}_t(\posterior^{\rsignal}_{t'}) - \hat{u}^{\rsignal}_{t'}} \cdot f^{\rsignal}_t(t')\\
&= (\hat{\reward}_t - \hat{z}_1) \cdot \sum_{t'=t}^{\stoptime_{\rewardMenu}}\posterior^{\rsignal}_{t'} \cdot f^{\rsignal}_t(t')
\leq (\hat{\reward}_t - \hat{z}_1)\cdot \posterior^{\no}_{t-1}
\leq (\hat{\reward}_{t-1} - \hat{z}_1)\cdot \posterior^{\no}_{t-1}
\end{align*}
where the second inequality holds by Bayesian plausibility
and the last inequality holds since the no information belief drifts towards state $0$. 
Combining the inequalities, we have $U_t - \hat{U}_t \leq u^{\no}_{t-1} - \hat{u}^{\no}_{t-1}$. 

\begin{figure}[H]
\centering
\begin{tikzpicture}[xscale=0.9,yscale=0.9]
\draw [<->] (0,5.2) -- (0,0) -- (8.2,0);

\draw [thick] (0.5,1.9375) -- (1, 1.875);
\draw [thick] plot [smooth, tension=0.5] coordinates { (1, 1.875) (1.9,1.85) (3, 2.08) (3.6, 2.6) (4, 3.5)};

\draw [thick] plot [smooth, tension=0.7] coordinates { (5, 2) (5.7,3) (6.2,3.1) (7, 3.8) (7.5, 3.85)};

\draw [dotted] (6.2, 0) -- (6.2,3.1);
\draw (6.2, -0.4) node {$\posterior^{\rsignal}_{\hat{t}}$};

\draw [dotted] (2.9, 0) -- (2.9,2);
\draw (2.78, -0.4) node {$\posterior^{\no}_{\hat{t}}$};
\draw [dotted] (8, 0) -- (8,5);
\draw (8, -0.4) node {$1$};
\draw (0, -0.4) node {$0$};
\draw [red] (0,2) -- (2.14634,1.7317);
\draw [red] (2.14634,1.7317) -- (8, 4);

\draw [red, dashed] (2.14634,1.7317) -- (8, 1);
\draw [red, dashed] (0,0.9) -- (2.14634,1.7317);

\draw [dotted] (0, 4) -- (8,4);
\draw (-0.3, 4) node {$\hat{z}_1$};
\draw (-0.3, 0.9) node {$\hat{z}_0$};

\draw [blue] (0, 0.9) -- (8,5);
\draw [dotted] (3.22, 0) -- (3.22,2.55);
\draw (3.38, -0.4) node {$\posterior^{\no}_t$};
\draw [dotted] (7, 0) -- (7,3.8);
\draw (7, -0.4) node {$\posterior^{\rsignal}_t$};

\draw [dotted] (0, 5) -- (8, 5);
\draw (-0.3, 5) node {$\hat{r}_t$};

\end{tikzpicture}
\caption{\footnotesize The blue line is the utility function $\tilde{u}_t$, which serves as an upper bound on the utility $u^{\rsignal}_{t'}$ for any $t'\geq t$.}
\label{fig:larger_diff}
\end{figure}
\end{enumerate}

Combining the above two cases, we have $U_t - \hat{U}_t \leq u^{\no}_{t-1} - \hat{u}^{\no}_{t-1}$ for any $t\leq \stoptime_{\rewardMenu}$. 
Since the agent's optimal effort strategy is to stop at time $\stoptime_{\rewardMenu}$ given contract $\rewardMenu$, 
this implies that at any time $t\leq \stoptime_{\rewardMenu}$, if the agent has not received any informative signal by time $t$, 
the agent also has incentive to exert effort until time $\stoptime_{\rewardMenu}$ 
given contract $\widehat{\rewardMenu}$ that can be implemented as a scoring rule.
\end{proof}

\newpage{}
\begin{center}
{\Large{}{}{}Online Appendix for ``Incentivizing Forecasters to Learn: Summarized vs. Unrestricted Advice''}{\Large\par}
\par\end{center}

\begin{center}
{\large{}{}{}Yingkai Li and Jonathan Libgober}{\large\par}
\par\end{center}

\global\long\def\thesection{OA \arabic{section}}%
\global\long\def\thedefn{OA\arabic{defn}}%
\global\long\def\thelem{OA\arabic{lem}}%
\global\long\def\theprop{OA\arabic{prop}}%
\global\long\def\thecor{OA\arabic{cor}}%
\global\long\def\theassumption{OA\arabic{assumption}}%
\setcounter{section}{0}
\setcounter{page}{1}

\section{Optimality of Dynamic Contracts}
\label{apx:multi-source}
\begin{proof}[Proof of \cref{lem:max_stop_time}]
If $T \leq T_{\lambda,\prior,c}$, the lemma holds trivially. 
Next we focus on the case $T > T_{\lambda,\prior,c}$. 

Suppose there exists a contract $\rewardMenu$ such that $\stoptime_{\rewardMenu} > T_{\lambda,\prior,c}$. 
The prior belief in the continuation game $\cgame_{\stoptime_{\rewardMenu}}$ 
is $\posterior^{\no}_{\stoptime_{\rewardMenu}-1} < \posterior_{\lambda,c} \leq \frac{1}{2}$. 
By the definition of $\stoptime_{\rewardMenu}$, the agent's optimal strategy is to exert effort for one period given contract $\rewardMenu$. 
This implies that the agent has incentive to exert effort in continuation game $\cgame_{\stoptime_{\rewardMenu}}$ 
given the effort-maximizing scoring rule for $\cgame_{\stoptime_{\rewardMenu}}$. 
By \cref{prop:static_moral_hazard}, the effort-maximizing scoring rule for $\cgame_{\stoptime_{\rewardMenu}}$ is the V-shaped scoring rule $\score$ with kink at $\posterior^{\no}_{\stoptime_{\rewardMenu}-1}$. 
By simple algebraic calculation, the expected utility increase given scoring rule~$\score$ for exerting effort in $\cgame_{\stoptime_{\rewardMenu}}$
is $\posterior^{\no}_{\stoptime_{\rewardMenu}-1}(\lambda^{\rsignal}_1 - \lambda^{\rsignal}_0)$, 
which must be at least the cost of effort~$c$. 
However, this violates the assumption that $\posterior^{\no}_{\stoptime_{\rewardMenu}-1} < \posterior_{\lambda,c}$, a contradiction. 
\end{proof}

\begin{lemma}[Approximate Effort Maximization of Myopic-incentive Contracts]\label{lem:contract_approx_optimal}\,\\
Given any prior $\prior \in (0,\frac{1}{2})$, any cost of effort $c$, any constant $\kappa_0 > 0$,
and any $\eta > 0$, 
there exists $\epsilon > 0$ such that for any $T\geq T_{\lambda,\prior,c}$ and any $\lambda$ satisfying 
that $\lambda^s_{\state} \leq \frac{1}{4}$ for all $s\in \{\lsignal,\rsignal\}$ and $\state\in\{0,1\}$, 
and 
\begin{itemize}
    \item $\lambda^{\rsignal}_1 - \lambda^{\rsignal}_0 \geq \frac{1}{\prior}(c + \kappa_0)$;
    \hfill (sufficient-incentive)
    \item $\lambda_1^{\rsignal}+\lambda_1^{\lsignal} \leq  \lambda_0^{\rsignal}+\lambda_0^{\lsignal} +\epsilon $, \hfill (slow-drift) 
\end{itemize}
letting $\rewardMenu$ be the myopic-incentive contract (\cref{def:myopic-incentive}), 
we have $\posterior^{\no}_{\stoptime_{\rewardMenu}} - \posterior^{\no}_{T_{\lambda,\prior,c}}\leq \eta$.
\end{lemma}
\begin{proof}[Proof of \cref{lem:contract_approx_optimal}]
For any time $t\geq 1$, 
given any information arrival probabilities $\lambda$
such that $\lambda_1^{\rsignal}+\lambda_1^{\lsignal} \leq  \lambda_0^{\rsignal}+\lambda_0^{\lsignal} +\epsilon $, 
we have 
\begin{align}
\posterior^{\no}_{t-1} - \posterior^{\no}_t &= 
\posterior^{\no}_{t-1} - \frac{\posterior^{\no}_{t-1}(1-\lambda^{\rsignal}_1-\lambda^{\lsignal}_1) }
{\posterior^{\no}_{t-1}(1-\lambda^{\rsignal}_1-\lambda^{\lsignal}_1) + 
(1-\posterior^{\no}_{t-1})(1-\lambda^{\rsignal}_0-\lambda^{\lsignal}_0) } \nonumber\\
& \leq \posterior^{\no}_{t-1} \rbr{1-\frac{(1-\lambda^{\rsignal}_1-\lambda^{\lsignal}_1)}{(1-\lambda^{\rsignal}_1-\lambda^{\lsignal}_1) + (1-\posterior^{\no}_{t-1}) \epsilon}}\nonumber\\
&\leq 2\posterior^{\no}_{t-1}(1-\posterior^{\no}_{t-1})\epsilon 
\leq \frac{1}{2}\epsilon.\label{eq:drift_bound}
\end{align}
The second inequality holds since $\lambda^{\rsignal}_1+\lambda^{\lsignal}_1\leq \frac{1}{2}$
and the last inequality holds since $\posterior^{\no}_{t-1}(1-\posterior^{\no}_{t-1}) \leq \frac{1}{4}$. 

Let
$\Delta_\lambda \triangleq \lambda^{\rsignal}_1-\lambda^{\rsignal}_0$.
By the sufficient-incentive assumption,
$\Delta_\lambda \ge \frac{c+\kappa_0}{\prior}$ and hence
$\frac{c}{\Delta_\lambda}\le \frac{c\prior}{c+\kappa_0}<\prior<\frac12$.
Therefore,
\[
\posterior_{\lambda,c}
=
\min\left\{\frac12,\frac{c}{\Delta_\lambda}\right\}
=
\frac{c}{\Delta_\lambda}.
\]

Fix any $\eta>0$, and let $\epsilon=\eta$. By \eqref{eq:drift_bound}, for every $t\ge 1$,
\[
0\le \posterior^{\no}_{t-1}-\posterior^{\no}_{t}\le \frac{\epsilon}{2}.
\]
Since $T_{\lambda,\prior,c}$ is the maximum time such that
$\posterior^{\no}_{T_{\lambda,\prior,c}-1}\ge \posterior_{\lambda,c}$,
we have
$\posterior^{\no}_{T_{\lambda,\prior,c}}
\ge \posterior_{\lambda,c}-\frac{\epsilon}{2}$.
Now consider any time $t$ such that
$\posterior^{\no}_{t-1}\ge \posterior^{\no}_{T_{\lambda,\prior,c}}+\eta$.
Then
\[
\posterior^{\no}_{t-1}
\ge \posterior_{\lambda,c}-\frac{\epsilon}{2}+\eta
= \posterior_{\lambda,c}+\frac{\eta}{2}
\ge \posterior_{\lambda,c}.
\]

Under the myopic-incentive contract, the expected reward gain from exerting effort for one period at time $t$ and then stopping is at least
\begin{align*}
\posterior^{\no}_{t-1}\lambda^{\rsignal}_1
+(1-\posterior^{\no}_{t-1})(1-\lambda^{\rsignal}_0)
\cdot \frac{\posterior^{\no}_{t-1}}{1-\posterior^{\no}_{t-1}}
-\posterior^{\no}_{t-1}
=
\posterior^{\no}_{t-1}\Delta_\lambda
\ge
\posterior_{\lambda,c}\Delta_\lambda
=
c.
\end{align*}
Hence the agent has an incentive to exert effort at time $t$.

Therefore, it cannot be that
$\posterior^{\no}_{\stoptime_{\rewardMenu}}
>
\posterior^{\no}_{T_{\lambda,\prior,c}}+\eta$.
Indeed, since $\posterior^{\no}_t$ is weakly decreasing in $t$, this inequality implies
$\stoptime_{\rewardMenu}<T_{\lambda,\prior,c}$, so the time
$t=\stoptime_{\rewardMenu}+1$ is feasible and satisfies
\[
\posterior^{\no}_{t-1}
=
\posterior^{\no}_{\stoptime_{\rewardMenu}}
>
\posterior^{\no}_{T_{\lambda,\prior,c}}+\eta.
\]
By the argument above, the agent would then still prefer to exert effort for one more period at time $\stoptime_{\rewardMenu}+1$, contradicting the definition of $\stoptime_{\rewardMenu}$ (together with our tie-breaking rule toward the largest stopping time). Thus,
\[
\posterior^{\no}_{\stoptime_{\rewardMenu}}
-
\posterior^{\no}_{T_{\lambda,\prior,c}}
\le \eta,
\]
and the claim holds.
\end{proof}

To prove \cref{thm:dynamic optimal}, we also utilize the following lemma to bound the difference in expected scores when the posterior beliefs differ by a small constant of $\epsilon$ given any bounded scoring rule. 
\begin{lemma}\label{lem:bounded_score}
For any bounded static scoring rule $\score$ with expected reward function $U_{\score}(\posterior)$ given posterior belief $\posterior$, 
we have 
\begin{align*}
\abs{U_{\score}(\posterior+\epsilon) - U_{\score}(\posterior)} \leq \epsilon, \qquad\forall \epsilon>0, \posterior\in[0,1-\epsilon].
\end{align*}
\end{lemma}
\begin{proof}
For any static scoring rule $\score$, the subgradient of $U_{\score}$ evaluated at belief $\posterior$ equals its difference in rewards between realized states $0$ and $1$, 
which is bounded between $[-1,1]$ 
since the scoring rule is bounded within $[0,1]$. 
This further implies that 
$\abs{U_{\score}(\posterior+\epsilon) - U_{\score}(\posterior)} \leq \epsilon$ for any $ \epsilon>0$ and $\posterior\in[0,1-\epsilon]$.
\end{proof}

\begin{proof}[Proof of \cref{thm:dynamic optimal}]
By \cref{lem:contract_approx_optimal}, it is sufficient to show that there exists $\eta > 0$
and $\epsilon>0$ such that when the slow-drift condition is satisfied for constant $\epsilon$, 
for any contract $\rewardMenu$ that can be implemented as a scoring rule, we have 
$\posterior^{\no}_{\stoptime_{\rewardMenu}} - \posterior^{\no}_{T_{\lambda,\prior,c}} > \eta$. 
\begin{figure}
\centering
\begin{tikzpicture}[xscale=1,yscale=1]
\hspace{-10pt}
\draw [<->] (0,4) -- (0,0) -- (7.3,0);

\draw [thick] plot [smooth, tension=0.7] coordinates { (0.5, 2.5) (3, 1.5) (4.8,2) (6, 3)};

\draw [red] (0, 2.5) -- (7, 0.3);
\draw [red] (0, 0.2) -- (7, 3);

\draw [dotted] (3.22, 0) -- (3.22, 1.49);
\draw (3.22, -0.4) node {$\posterior^{\no}_{\stoptime_{\contract}}$};

\draw [dotted] (1.38, 0) -- (1.38, 2.05);
\draw (1.38, -0.4) node {$\posterior^{\lsignal}_{\stoptime_{\contract}}$};

\draw [dotted] (5.3, 0) -- (5.3, 2.3);
\draw (5.3, -0.4) node {$\posterior^{\rsignal}_{\stoptime_{\contract}}$};

\draw [dotted] (7, 0) -- (7, 3);
\draw (7, -0.3) node {$1$};

\end{tikzpicture}
\caption{\footnotesize The black curve is the expected score function $\expectScore$. 
The red lines are linear functions $\underline{U}$ and $\overline{U}$ respectively. }
\label{fig:static_bound}
\end{figure}

Suppose by contradiction there exists a contract $\rewardMenu$ that can be implemented as a scoring rule and
$\posterior^{\no}_{\stoptime_{\rewardMenu}} - \posterior^{\no}_{T_{\lambda,\prior,c}} \leq \eta$.
Let $\score$ be the scoring rule that implements contract $\rewardMenu$ and let $\expectScore(\posterior) = \expect[\state\sim\posterior]{\score(\posterior,\state)}$
be the expected score of the agent. 
Let $\underline{U}(\posterior)$ be a linear function such that 
$\underline{U}(\posterior^{\lsignal}_{\stoptime_{\rewardMenu}}) = \expectScore(\posterior^{\lsignal}_{\stoptime_{\rewardMenu}})$
and $\underline{U}(\posterior^{\no}_{\stoptime_{\rewardMenu}}) = \expectScore(\posterior^{\no}_{\stoptime_{\rewardMenu}})$. 
Let $\overline{U}(\posterior)$ be a linear function such that 
$\overline{U}(\posterior^{\rsignal}_{\stoptime_{\rewardMenu}}) = \expectScore(\posterior^{\rsignal}_{\stoptime_{\rewardMenu}})$
and $\overline{U}(\posterior^{\no}_{\stoptime_{\rewardMenu}}) = \expectScore(\posterior^{\no}_{\stoptime_{\rewardMenu}})$. 
See \Cref{fig:static_bound} for an illustration. 
Let $f^s_t \triangleq \posterior^{\no}_{t}\lambda^{s}_1 + (1-\posterior^{\no}_t)\lambda^{s}_0$ for any $s\in\{\lsignal,\rsignal\}$. 
At time $\stoptime_{\rewardMenu}$, the agent has incentive to exert effort, which implies that the cost of effort $c$ is at most the utility increase for exerting effort
\begin{align*}
&f^{\rsignal}_{\stoptime_{\rewardMenu}-1} \cdot\expectScore(\posterior^{\rsignal}_{\stoptime_{\rewardMenu}})
+ f^{\lsignal}_{\stoptime_{\rewardMenu}-1} \cdot\expectScore(\posterior^{\lsignal}_{\stoptime_{\rewardMenu}})
+ (1- f^{\rsignal}_{\stoptime_{\rewardMenu}-1} -f^{\lsignal}_{\stoptime_{\rewardMenu}-1} ) \cdot\expectScore(\posterior^{\no}_{\stoptime_{\rewardMenu}})
- \expectScore(\posterior^{\no}_{\stoptime_{\rewardMenu}-1})\\
& = f^{\rsignal}_{\stoptime_{\rewardMenu}-1} \cdot
(\expectScore(\posterior^{\rsignal}_{\stoptime_{\rewardMenu}})- 
\underline{U}(\posterior^{\rsignal}_{\stoptime_{\rewardMenu}}))
+ \underline{U}(\posterior^{\no}_{\stoptime_{\rewardMenu}-1})
- \expectScore(\posterior^{\no}_{\stoptime_{\rewardMenu}-1})
\leq f^{\rsignal}_{\stoptime_{\rewardMenu}-1} \cdot
(\expectScore(\posterior^{\rsignal}_{\stoptime_{\rewardMenu}})- 
\underline{U}(\posterior^{\rsignal}_{\stoptime_{\rewardMenu}}))
\end{align*}
where the equality holds by linearity of expectation 
and the inequality holds by the convexity of utility function $\expectScore$. 
Therefore, we have 
\begin{align*}
\expectScore(\posterior^{\rsignal}_{\stoptime_{\rewardMenu}})- 
\underline{U}(\posterior^{\rsignal}_{\stoptime_{\rewardMenu}})
&\geq \frac{c}{f^{\rsignal}_{\stoptime_{\rewardMenu}-1}}
= \frac{c}{\posterior^{\no}_{\stoptime_{\rewardMenu}-1}\lambda^{\rsignal}_1
+ (1-\posterior^{\no}_{\stoptime_{\rewardMenu}-1})\lambda^{\rsignal}_0}\\
&\geq\frac{(\lambda^{\rsignal}_1-\lambda^{\rsignal}_0) (1-\frac{\eta}{\posterior^{\no}_{\stoptime_{\rewardMenu}-1}})}{\lambda^{\rsignal}_1
+ \frac{1}{\posterior^{\no}_{\stoptime_{\rewardMenu}-1}}(1-\posterior^{\no}_{\stoptime_{\rewardMenu}-1})\lambda^{\rsignal}_0}
\geq \frac{(\lambda^{\rsignal}_1-\lambda^{\rsignal}_0)}
{\lambda^{\rsignal}_1
+ \frac{1}{\posterior^{\no}_{\stoptime_{\rewardMenu}-1}}(1-\posterior^{\no}_{\stoptime_{\rewardMenu}-1})\lambda^{\rsignal}_0}
- \frac{2\eta}{\underline{\kappa}_1}
\end{align*}
where the second inequality holds since  $\posterior^{\no}_{\stoptime_{\rewardMenu}} \leq \posterior^{\no}_{T_{\lambda,\prior,c}} + \eta
\leq \posterior_{\lambda,c} + \eta$, 
and the last inequality holds since $\lambda^{\rsignal}_0 \geq \underline{\kappa}_1$. 

Fix any $\gamma > 0$ such that $\gamma < \frac{D \kappa_{0}}{2(c+\kappa_{0})}$. For any sufficiently small $\eta$ and $\varepsilon$, let time $t \leq \stoptime_{\rewardMenu}$ be such that $\posterior^{\no}_{t-1} - \posterior^{\no}_{\stoptime_{\rewardMenu}-1} > \gamma$. 
First note that the convexity of $\expectScore$ 
and the constraint on rewards belonging to the unit interval at state $1$ implies that 
\begin{align*}
\overline{U}(\posterior^{\no}_{t-1})-\expectScore(\posterior^{\no}_{t-1})
\leq \frac{2\eta}{\underline{\kappa}_1} \cdot \frac{1-\posterior^{\no}_{\stoptime_{\rewardMenu}}}{1-\posterior^{\rsignal}_{\stoptime_{\rewardMenu}}}
\cdot \frac{\posterior^{\rsignal}_{\stoptime_{\rewardMenu}}-\posterior^{\no}_{t-1}}{\posterior^{\rsignal}_{\stoptime_{\rewardMenu}}-\posterior^{\no}_{\stoptime_{\rewardMenu}}}
\leq \frac{2\eta(\underline{\kappa}_1 + \bar{\kappa}_1)}{\underline{\kappa}^2_1}
\end{align*}
where the last inequality holds since 
$(1-\posterior^{\no}_{\stoptime_{\rewardMenu}})
\cdot \frac{\posterior^{\rsignal}_{\stoptime_{\rewardMenu}}-\posterior^{\no}_{t-1}}{\posterior^{\rsignal}_{\stoptime_{\rewardMenu}}-\posterior^{\no}_{\stoptime_{\rewardMenu}}} \leq 1$
and $ \frac{1}{1-\posterior^{\rsignal}_{\stoptime_{\rewardMenu}}}\leq \frac{\underline{\kappa}_1 + \bar{\kappa}_1}{\underline{\kappa}_1}$.
Moreover, 
\begin{align*}
\expectScore(\posterior^{\rsignal}_t)-\overline{U}(\posterior^{\rsignal}_t)
\leq \frac{2\eta}{\underline{\kappa}_1} \cdot \frac{1-\posterior^{\no}_{\stoptime_{\rewardMenu}}}{\posterior^{\rsignal}_{\stoptime_{\rewardMenu}}-\posterior^{\no}_{\stoptime_{\rewardMenu}}}
\leq \frac{2\eta}{\underline{\kappa}_1} \cdot \frac{\lambda^{\rsignal}_1}{\posterior^{\no}_{\stoptime_{\rewardMenu}}(\lambda^{\rsignal}_1-\lambda^{\rsignal}_0)}
\leq \frac{2\eta\bar{\kappa}_1}{\underline{\kappa}_1 c}.
\end{align*}
Therefore, the utility increase for exerting effort in one period at time $t$ is 
\begin{align*}
&f^{\rsignal}_{t-1} \cdot\expectScore(\posterior^{\rsignal}_t)
+ f^{\lsignal}_{t-1} \cdot\expectScore(\posterior^{\lsignal}_t)
+ (1- f^{\rsignal}_{t-1} -f^{\lsignal}_{t-1} ) \cdot\expectScore(\posterior^{\no}_t)
- \expectScore(\posterior^{\no}_{t-1})\\
&\leq f^{\rsignal}_{t-1} \cdot
\overline{U}(\posterior^{\rsignal}_t)
+ f^{\lsignal}_{t-1} \cdot\underline{U}(\posterior^{\lsignal}_t)
-(f^{\rsignal}_{t-1} + f^{\lsignal}_{t-1}) \cdot
\overline{U}(\posterior^{\no}_{t})
+ \epsilon
+ \frac{2\eta\bar{\kappa}_1}{\underline{\kappa}_1 c}
+ \frac{2\eta(\underline{\kappa}_1 + \bar{\kappa}_1)}{\underline{\kappa}^2_1}\\
&\leq \posterior^{\no}_{\stoptime_{\rewardMenu}}(\lambda^{\lsignal}_0-\lambda^{\lsignal}_1)-\gamma\underline{\kappa}_1
+ \epsilon
+ \frac{2\eta\bar{\kappa}_1}{\underline{\kappa}_1 c}
+ \frac{2\eta(\underline{\kappa}_1 + \bar{\kappa}_1)}{\underline{\kappa}^2_1}.
\end{align*}
Since $c = \posterior_{\lambda,c}(\lambda^{\rsignal}_1-\lambda^{\rsignal}_0)
\geq (\posterior^{\no}_{\stoptime_{\rewardMenu}}-\eta)(\lambda^{\lsignal}_0-\lambda^{\lsignal}_1)$, 
the agent suffers a loss at least
\begin{align*}
\gamma\underline{\kappa}_1
- \eta \bar{\kappa}_1
- \epsilon
- \frac{2\eta\bar{\kappa}_1}{\underline{\kappa}_1 c}
- \frac{2\eta(\underline{\kappa}_1 + \bar{\kappa}_1)}{\underline{\kappa}^2_1}
\end{align*}
for exerting effort in one period. 

Now consider the utility increase for exerting effort in continuation game $\cgame_t$, starting from belief $\posterior^{\no}_{t-1}$. 
Note that for any $\delta > 0$, by \eqref{eq:drift_bound} the no-information belief can decrease by at most $\epsilon/2$ in one period. Hence before the no-information belief drifts by a total distance~$\delta$, the agent must work for at least
\[
N \triangleq \left\lceil \frac{2\delta}{\epsilon}\right\rceil
\]
periods.
Moreover, in each such period the total probability of receiving a Poisson signal is at least
\[
f^{\rsignal}_{u-1}+f^{\lsignal}_{u-1}\geq 2\underline{\kappa}_1,
\]
because each arrival rate lies in $[\underline{\kappa}_1,\bar{\kappa}_1]$.
Therefore the probability of receiving at least one Poisson signal before the no-information belief drifts by $\delta$ is at least
\begin{align*}
1-(1-2\underline{\kappa}_1)^N
\geq
1-\exp\left(-2\underline{\kappa}_1 N\right).
\end{align*}

By continuity of the preceding one-period upper bound as a function of the current no-information belief, after possibly shrinking $\delta$ we may ensure that the loss is at least 
\begin{align*}
m \triangleq
\gamma\underline{\kappa}_1
- 2\delta
- \eta \bar{\kappa}_1
- \epsilon
- \frac{2\eta\bar{\kappa}_1}{\underline{\kappa}_1 c}
- \frac{2\eta(\underline{\kappa}_1 + \bar{\kappa}_1)}{\underline{\kappa}_1^2}
\end{align*}
in each period before the no-information belief drifts a $\delta$ distance. 
In contrast, the benefit from exerting effort after the no-information belief drifts a $\delta$ distance is at most $1$, and it can occur only if no Poisson signal arrives in the first $N$ periods. The probability of this event is at most
\[
(1-2\underline{\kappa}_1)^N \leq \exp\left(-2\underline{\kappa}_1 N\right).
\]
Moreover, the probability that the agent reaches period $j+1$ without a Poisson signal in the first $j$ periods is at most
$(1-2\underline{\kappa}_1)^j$,
for $j=0,\dots,N-1$.
Therefore, by decomposing the continuation gain into the incremental gains from each successive exertion decision, the maximal continuation gain from exerting effort in continuation game $\cgame_t$, relative to stopping immediately, is at most
\begin{align*}
-m\sum_{j=0}^{N-1}(1-2\underline{\kappa}_1)^j + (1-2\underline{\kappa}_1)^N
\leq -m\sum_{j=0}^{N-1}(1-2\underline{\kappa}_1)^j + \exp\left(-2\underline{\kappa}_1 N\right).
\end{align*}
Since $m>0$ when $\delta,\eta,\epsilon$ are chosen sufficiently small compared to $\gamma$, we can then keep $\delta>0$ fixed and choose $\epsilon$ sufficiently small so that $N=\lceil 2\delta/\epsilon\rceil$ is large enough for the right-hand side to be negative. Therefore, the agent's utility for exerting effort is smaller than not exerting effort in continuation game $\cgame_t$. This leads to a contradiction since the agent at time $t$ will not choose to exert effort given scoring rule $\score$.
\end{proof}

\section{General Dynamic Contracts}

\subsection{Menu Representation}
\begin{proof}[Proof of \cref{lem:menu}]
Recall that by \cref{lem:agent_effort_response}, it is without loss to assume that the agent uses a stopping strategy. 
Given any time $t$, let $r^s_t$ be the menu option chosen by the agent if the agent receives a Poisson signal $s\in S$ at time $t$.
Let $\reward^{\no}_{\stoptime_{\rewardMenu}}$ be the menu option the agent chooses if the agent does not receive any Poisson signal before stopping effort. 
By offering the menu options $\rewardMenu_t = \{\reward^{\signal}_{t'}\}_{t'\geq t,\signal\in\signals}\cup\{\reward^{\no}_{\stoptime_{\rewardMenu}}\}$
at time $t$, the agent's utilities for stopping and exerting effort remain the same, 
and hence this simplified menu representation implements the same stopping strategy. 
\end{proof}

\subsection{Communication-based Contracts}
\label{subapx:contract_definition}

Let $\messages_t$ be the message space of the agent at any time $t$. We denote the history at time~$t$ as $\history_t = \lbr{\mess_{t'}}_{t'\leq t}$. Let $\histories_t$ be the set of all possible histories at time $t$. Since we allow rewards to condition on the future realized outcome of interest (i.e., the state), we take the rewards to be within the unit interval. As mentioned in our model, we can interpret the reward as an endorsement that positively influences the forecaster's reputation. 
The agent thus receives a benefit according to the function: 
\begin{align*}
\contract:\histories_T \times \states \to [0,1]
\end{align*}
where $\contract(\history_T, \state)$ is the fraction of the total available reward provided to the agent
when his history of reports is $\history_T$, and the realized state is $\state$---or, alternatively, the probability that the agent receives the reward, so that $\contract(h_{T}, \theta)$ can be interpreted as an endorsement probability.\footnote{While we allow randomization over the event that the agent receives the full reward, we do not allow stochastic messages from the principal to the agent. A discussion of how the possibility of such randomization influences the results is deferred to \cref{sub:randomized}. Allowing stochastic messages would introduce a second wedge between scoring rules and arbitrary contracts, which we do not consider to maintain focus on the ability to elicit information over time.
} If the agent has exerted effort in~$\tilde{t}$ periods, his final payoff is $\contract(h_{T}, \theta)-c\tilde{t}$.

This representation of the contract is equivalent to our menu representation. In particular, given any contract $\contract$, at any time $t$, it is sufficient to provide the agent with menu options that will be implemented for the on-path beliefs $\posterior^s_t$ for receiving a Poisson signal $s$, or a belief $\posterior^N_t$ for not receiving any Poisson signal. 
Therefore, all our results extend naturally.

\subsection{Randomized Contracts}
\label{sub:randomized} 

Our goal has been to determine the maximum effort implementable within the class of contracts defined in \cref{sec:model:contracting} and \ref{subapx:contract_definition}. As histories only include messages sent by the agent, we implicitly rule out the use of stochastic messages sent from the designer to the agent  (as in, for instance, \cite{DPS2018}). Deterministic mechanisms have significant practical appeal, as it is not always obvious what a randomization might correspond to or how a mechanism designer could commit to implementing this. These issues have been discussed extensively in the contracting literature; we refer the reader to discussions in \cite{LaffontMartimort} as well as \cite{Bester2001} on this point to avoid detours. Now, randomization provides no additional benefit in implementing maximum effort with probability~1. But it is natural to ask whether some designer objectives may yield benefits to randomization. Our analysis speaks to this question as well. 

We discuss randomization formally. Let $\randomDevice=\lbr{\randomDevice_t}_{t\leq T}$ be a sequence of random variables with~$\randomDevice_t$ drawn from a uniform distribution in $[0,1]$. 
A \emph{randomized contract} is a mapping 
\begin{align*}
\contract(\cdot \given \randomDevice):\histories_T \times \states \to [0,1].
\end{align*}
Crucially, in randomized contracts, at any time $t$, 
the history of the randomization device $\lbr{\randomDevice_{t'}}_{t'\leq t}$ is publicly revealed to the agent before determining his choice of effort or the message sent to the designer. 
The randomization revealed before $t$ affects the agent's incentives after time $t$, 
and without it, such contracts reduce back to deterministic contracts. 
Implementing such randomized contracts would require either a public randomization device or the principal to commit to a random effort recommendation policy.

It may no longer be without loss of generality to restrict attention to stopping strategies under a randomized contract. 
In particular, the agent may decide whether to work or not, depending on the past realizations of the public randomization. 
The agent may also strategically delay exerting effort to wait for the realization of the public randomization.
As a result, simplifying the objective to maximizing the agent's stopping time is not appropriate for randomized contracts. 

At the same time, our analysis provides some insights into why randomization can expand the set of implementable strategies. We previously observed that it may be possible to get the agent to work from $\mu_{1}$ to $\mu_{2}$ and to get the agent to work from $\mu_{2}$ to $\mu_{3}$, but not from $\mu_{1}$ to $\mu_{3}$. This would occur if the reward necessary to get the agent to work to $\mu_{3}$ were so high that the agent would ``shirk-and-lie'' at $\mu_{1}$.  However, the agent may be willing to start working at $\mu_{1}$, not knowing whether the reward will be ``high'' or ``low''---but once the agent starts working, the designer can randomly inflate or decrease the rewards of the agent. Once time has passed, the outcome of the randomization can be revealed, and if the rewards inflate, the agent has incentives to exert effort for longer in the absence of a Poisson signal arrival---so that the realized stopping time increases for \emph{some} realization of the randomization.

We illustrate this intuition more formally assuming perfect ``good-news'' learning, i.e., $\lambda_{1}^{\rsignal} > 0 $ and $ \lambda_{0}^{\rsignal}=\lambda_{1}^{\lsignal} = \lambda_{0}^{\lsignal}=0$, where learning is perfect and only right-biased signal $\rsignal$ arrives with positive probability. We describe the resulting solution in \cref{sec:illustration}. 
In this setting, our results imply that, under deterministic contracts, a V-shaped scoring rule with parameters $r_0=1$ and $r_1\in[0,1]$ implements the effort-maximizing contract $\contract$; in particular, \cref{sec:illustration} characterizes when in fact effort-maximizing contracts require $r_1 < 1$. 
Recall that we denote $\stoptime_{\rewardMenu}$ as the stopping time in the effort-maximizing deterministic contract and let $\posterior^{\no}_{\stoptime_{\rewardMenu}}$
be the stopping belief when no Poisson signal is observed. Let $\delta\in(0, 1-r_1]$ be the maximum number such that 
(1) the agent has strict incentives to exert effort until time $\frac{2\stoptime_{\rewardMenu}}{3}$ 
absent signal arrival
given menu options $(1, 0)$ and $(0, r_1 - \delta)$;
and (2) the agent can be incentivized to exert effort 
given menu options $(1, 0)$ and $(0, r_1 + \delta)$
given belief $\posterior^{\no}_{\frac{\stoptime_{\rewardMenu}}{2}}$. 
Consider the randomized contract $\hat{\contract}$ that provides menu options $(1, 0)$ and $(0, r_1 - \delta)$ from time $0$ to $\frac{\stoptime_{\rewardMenu}}{2}$, 
and after time $\frac{\stoptime_{\rewardMenu}}{2}$, offers menu options $(1,0)$ and $(0,r_1 + \delta)$ with probability $\epsilon^2$, 
offers menu options $(1,0)$ and $(0,0)$ with probability~$\epsilon$, 
and offers the same menu options $(1, 0)$ and $(0, r_1 - \delta)$ otherwise. 
With sufficiently small $\epsilon>0$, 
the agent still has incentives to exert effort at any time $t\leq \frac{\stoptime_{\rewardMenu}}{2}$. 
Moreover, after time $\frac{\stoptime_{\rewardMenu}}{2}$, with probability $\epsilon^2$, 
the realized menu options are $(1,0)$ and $(0,r_1+\delta)$, and the agent can be incentivized to exert effort to a time strictly larger than $\stoptime_{\rewardMenu}$ in the absence of a Poisson signal.\footnote{This kind of modification may be of interest to a designer who only values extreme posterior beliefs. For instance, suppose the designer faced a decision problem where the possible decisions belonged to a set $A=\{0,1\}$. Consider a designer payoff function of $\val(0,\state)=0$ for all $\state\in\states$, 
and $\val(1,0)=1, \val(1,1)=-\frac{1-\posterior^{\no}_{\stoptime_{\rewardMenu}}}{\posterior^{\no}_{\stoptime_{\rewardMenu}}}$; plainly, the designer only seeks to change her action from 1 to 0 if the posterior belief is below $\posterior^{\no}_{\stoptime_{\rewardMenu}}$. 
In this case, the payoff under any deterministic contract is 0. However, under the randomized contract outlined, the designer would obtain a positive payoff when implementing the identified randomized contract.}

\section{Additional Review of Scoring Rules}
A scoring rule is \emph{proper} if it incentivizes the agent to truthfully report his belief to the mechanism, 
i.e., 
\begin{align*}
\expect[\state\sim\posterior]{\score(\posterior,\state)}
\geq \expect[\state\sim\posterior]{\score(\posterior',\state)}, 
\quad\forall \posterior,\posterior'\in\Delta(\states). 
\end{align*}
By the revelation principle, it is without loss of generality to focus on proper scoring rules when the designer adopts contracts that can be implemented as a scoring rule. 
\begin{lemma}[\citealp{Mcc-56}]
\label{prop:static_truthful}
For any finite state space $\states$, 
a scoring rule $\score$ is proper if there exists a convex function $\expectScore:\Delta(\states) \to \reals$ 
such that 
\begin{align*}
\score(\posterior,\state) = \expectScore(\posterior) + \subgradient(\posterior)\cdot(\state - \posterior)
\end{align*}
for any $\posterior\in \Delta(\states)$ and $\state\in\states$
where $\subgradient(\posterior)$ is a subgradient of $\expectScore$.\footnote{Here for finite state space $\states$, 
we represent $\state\in\states$ and $\posterior\in\Delta(\states)$ 
as $|\states|$-dimensional vectors 
where the $i$th coordinate of $\state$ is $1$ if the state is the $i$th element in $\states$  
and is $0$ otherwise,
and the $i$th coordinate of posterior $\posterior$ is the probability of the $i$th element in $\states$ given posterior $\posterior$. } 
\end{lemma}

\section{Comparative Statics of Scoring Rules}
\label{sub:comparative statics of static}
As discussed in \cref{sect:discussions}, the ideal situation cannot be implemented in dynamic contracts since we claim that the effort-maximizing static scoring rule at time $\stoptime$ is not effort-maximizing at time $t<\stoptime$. 
However, this argument alone is insufficient since, in earlier times, the agent is more uncertain about the states and hence it is easier to incentivize. 
In this appendix, we formalize this intuition using comparative statics on static scoring rules.

To simplify the exposition, we consider a specific static environment where
, if the agent exerts effort, the agent may receive an informative signal in $\lbr{\rsignal,\lsignal}$ that is partially informative about the state. 
Otherwise, the agent does not receive any signals and the prior belief is not updated. 
Let $f_{\theta,s}\in(0,1)$ be the probability of receiving signal $s$ conditional on state~$\theta$, with $\sum_{s\in\{\rsignal,\lsignal\}} f_{\theta,s} = 1$ for any $\theta$.
That is, signals are not perfectly revealing. 
We focus on the case when the prior $\prior < \frac{1}{2}$.

\cref{prop:static_moral_hazard} shows that the utility function of the effort-maximizing scoring rule is V-shaped with a kink at the prior,
that is, the effort-maximizing scoring rule offers the agent the following two options:
$(0,1)$ and $(\frac{\prior}{1-\prior},0)$.
The agent with prior belief $\prior$ is indifferent between these two options. 
Moreover, any belief $\posterior > \prior$ would strictly prefer $(0,1)$
and any belief $\posterior<\prior$ would strictly prefer $(\frac{\prior}{1-\prior},0)$.

Next, we conduct comparative statics.
The expected score increase for exerting effort under the effort-maximizing scoring rule is
\begin{align*}
\inc(\prior) \triangleq (1-\prior)\cdot f_{0,\lsignal}\cdot \frac{\prior}{1-\prior} + \prior \cdot f_{1,\rsignal} - \prior
= \prior(f_{0,\lsignal}+f_{1,\rsignal} - 1).
\end{align*}
Since $f_{0,\lsignal} > f_{1,\lsignal}$, 
we have $f_{0,\lsignal}+f_{1,\rsignal} = f_{0,\lsignal}+ 1 - f_{1,\lsignal} > 1$.
Therefore, the expected score increase is monotone increasing in prior $\prior$. 
That is, the closer the prior is to $\frac{1}{2}$, 
the easier it is to incentivize the agent to exert effort. 

Next, we conduct comparative statics on prior $\prior'$ by fixing the scoring rule $\score$ to be effort-maximizing for~$\prior$, 
i.e., $\score$ is the V-shaped scoring rule with a kink at $\prior$. 
The expected score increase for exerting effort given scoring rule $\score$ is 
\begin{align*}
\inc(\prior'; \prior) &\triangleq (1-\prior')\cdot f_{0,\lsignal}\cdot \frac{\prior}{1-\prior} + \prior' \cdot f_{1,\rsignal} - \prior'\nonumber\\
&= \prior'(f_{1,\rsignal}-1-f_{0,\lsignal}\cdot \frac{\prior}{1-\prior}) + f_{0,\lsignal}\cdot \frac{\prior}{1-\prior}.
\end{align*}
Since $f_{1,\rsignal}<1$, 
the strength of the incentives the designer can provide is strictly decreasing in prior $\prior'$. Therefore, even though the prior is closer to $\frac{1}{2}$, it is easier to incentivize the agent to exert effort; 
the effort-maximizing scoring rule for lower priors may not be sufficient to incentivize the agent (assuming that the cost of effort is the same in both settings).

\section{Complete Solution under Perfect Good News Learning}
\label{sec:illustration}

\subsection{Effort-Maximizing Contracts} 

Here we describe our findings in the model where the time interval is $\Delta$ and take the $\Delta \rightarrow 0$ limit. We focus on the setting with a single perfectly revealing signal, i.e., $$\lambda^{\lsignal}_0=\lambda^{\lsignal}_1=\lambda^{\rsignal}_0 = 0, ~~ \lambda^{\rsignal}_{1} > 0.$$ 
We take the horizon $T$ to be sufficiently large so that it will not be a binding constraint. Illustrating this solution highlights the tensions in maximizing the incentives to exert effort at different points in time. 

For this learning environment, \cref{thm:perfect} implies that a scoring rule with two reward functions implements maximum effort: $(1,0)$, corresponding to a guess of state~0, and $(0, r_{1})$, corresponding to a guess of state 1. The value of $r_{1}$, the reward when guessing state 1 (correctly), depends on the initial prior, $\prior$ (which coincides with $\mu_{0}^{\no}$).  

We describe how $r_{1}$ is determined. Providing a higher reward for guessing state 1 encourages the agent to continue exerting effort, even as this state appears increasingly unlikely. In fact, one can show the agent is indifferent between continuing effort and selecting a reward when $\mu_{\tau}^{\no} = \frac{c}{\lambda_{1}^{\rsignal}r_{1}}$ (see Appendix for details). If the event that $\theta=1$ is not too likely according to the initial prior, then setting $r_{1}=1$ gets the agent to work for as long as possible. However, if the probability that $\theta=1$ is initially high, setting $r_{1}=1$ may violate the agent's initial incentive constraints. 

Figure \ref{fig:solution} illustrates the agent's value function when $r_{1}=1$, assuming the agent works until $\mu_{\tau}^{\no} = \frac{c}{\lambda_{1}^{\rsignal}}$, along with the expected payoff when selecting each reward function. While the adverse selection constraint holds for this contract, moral hazard is violated if the event that $\theta=1$ is sufficiently likely initially. This can be seen by observing that the value function is below the expected payoff when choosing $(0,1)$, so the agent would prefer to guess state 1 rather than exert any effort at all. 

\begin{figure}
\centering
\includegraphics[scale=.4]{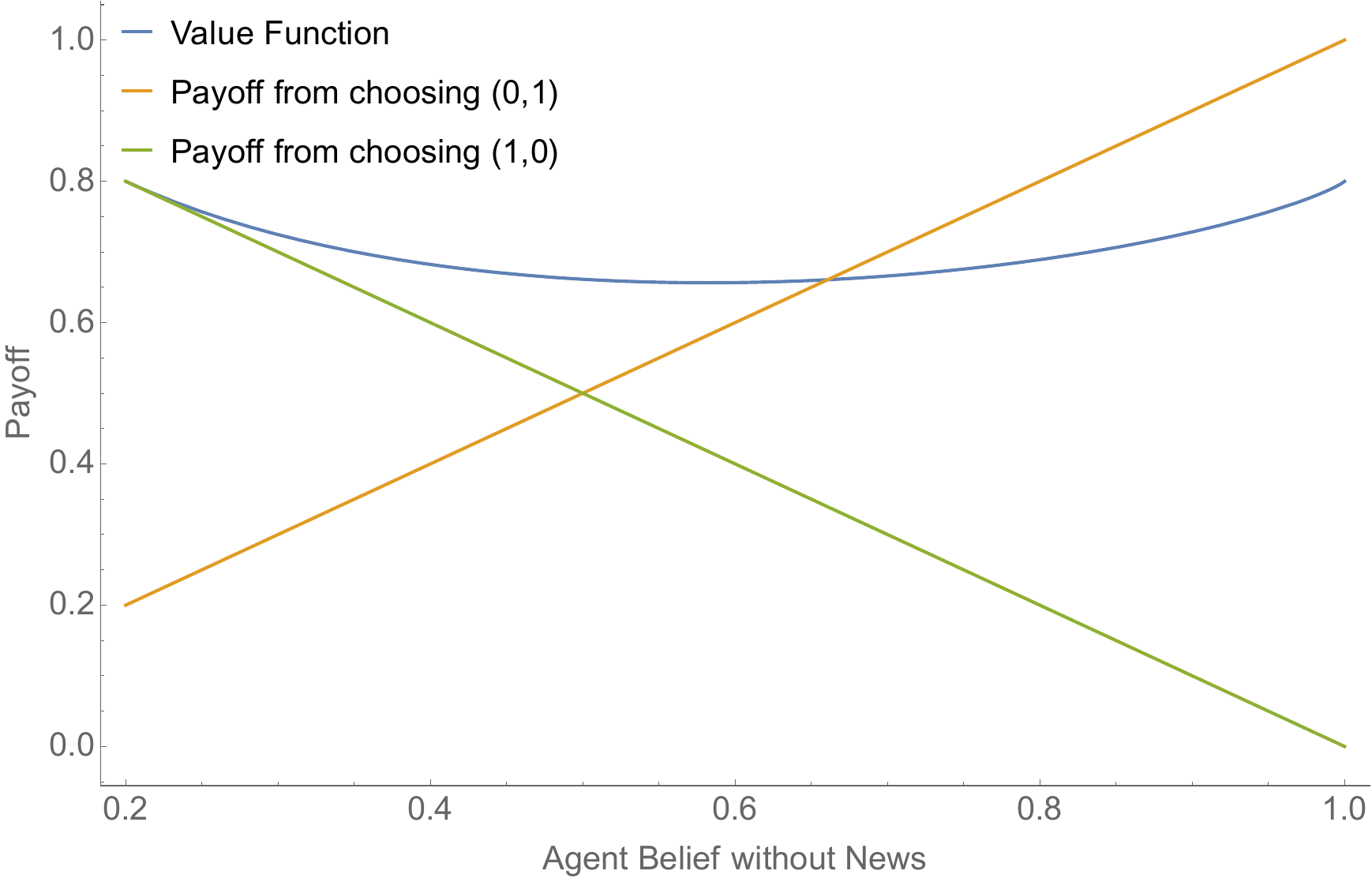}
\caption{\footnotesize Value function with perfect learning; $r_1=1, c=.2, \lambda_{1}^{\rsignal}=1$.}
\label{fig:solution}
\end{figure}

In this case, lowering $r_{1}$ is necessary to motivate the agent to begin working. The cost, of course, is that now the agent stops working once $\mu_{\tau}^{\no} = \frac{c}{\lambda_{1}^{\rsignal}r_{1}}$, so that the agent stops sooner when $r_{1}$ is lower. Now, when $\posterior_{0}^{\no} < c/ \lambda_{1}^{\rsignal}$, it is impossible to induce the agent to acquire any information at all. Outside of this range, a pair of thresholds, $\posterior^{*}, \posterior^{**}$ satisfying $c/ \lambda_{1}^{\rsignal} < \posterior^{*} < \posterior^{**} < 1$ determine the form of the effort-maximizing scoring rule:

\begin{itemize} 

\item For $c/\lambda_{1}^{\rsignal} \leq \mu_{0}^{\no} \leq \posterior^{*}$,  the effort-maximizing scoring rule sets $r_1=1$. 

\item For $\posterior^{*} < \mu_{0}^{\no}  \leq \posterior^{**}$,  the effort-maximizing scoring rule sets $r_1 < 1$,  with the exact value pinned down by the condition that at time $0$, the agent is indifferent between (i) working absent signal arrival until their belief is $c/ (\lambda_{1}^{\rsignal} r_1)$ and (ii) never working.

\item For $\mu_{0}^{\no}  > \posterior^{**}$,  it is impossible to induce the agent to acquire any information. 
\end{itemize}
\noindent  Of course, when $\mu_{0}^{\no} \in (\mu^{*}, \mu^{**})$, if the agent works, beliefs may eventually leave this region. If the agent had \emph{started} at such a belief, more effort could be induced by setting $r_{1}=1$. On the other hand, it is clear why this reoptimization cannot help. When $\mu_{0}^{\no} \in (\mu^{*}, \mu^{**})$, $r_{1} < 1$ will be set so that the initial moral hazard constraint binds---but if the agent expected $r_{1}$ to increase later, he would simply shirk and claim to have observed a Poisson signal once the reward is increased to $1$. Given that such adjustments are impossible, effort-maximizing mechanisms cannot utilize dynamics and are therefore static. 

The lesson is that the tensions between optimizing incentives to exert effort both earlier and later may be unavoidable. A designer may be unable to re-optimize rewards because the re-optimized rewards would violate an incentive constraint at some other time. One deceptive aspect of this example is that the agent's moral hazard constraint only ever binds at the stopping belief and (possibly) at their initial belief. A technical challenge we face in Section \ref{sec:single_source}, for instance, is that if signals are not fully revealing, it may be that the moral hazard constraint binds somewhere ``in between.'' This feature will imply extra reward functions should be provided to the agent in the effort-maximizing scoring rule. Still, this example illustrates some of the intuition on how effort-maximizing rewards vary with the agent's beliefs.

\subsection{Details Behind the Calculation}
\label{apx:illustration_detail}

Here we provide some additional details behind the calculations in Section \ref{sec:illustration}. We consider any scoring rule which involves the choice between $(r_{0},0)$ and $(0,r_{1})$, with the former being chosen when stopping in the absence of any signal arrival and the latter being chosen if one does occur. We note that stopping and accepting contract $(r_0, 0)$ delivers payoff $r_0(1- \mu_{t}^{\no})$; if, at time $\stoptime$, the agent continues for a length $\Delta$ and then stops, the payoff is: 

\begin{equation*} 
-c \Delta + (1- \lambda_{1}^{\rsignal} \mu_{\stoptime}^{\no} \Delta)r_0 (1-\mu_{\stoptime+\Delta}^{\no})+ \lambda_{1}^{\rsignal} \mu_{\stoptime}^{\no} \Delta r_1. 
\end{equation*}

Imposing indifference between stopping and continuing yields:

\begin{equation*} 
r_0\frac{\mu_{\stoptime}^{\no}-\mu_{\stoptime+ \Delta}^{\no}}{\Delta} + \lambda_{1}^{\rsignal} \mu_{\stoptime}^{\no} r_1- \lambda_{1}^{\rsignal} \mu_{\stoptime}^{\no} (1-\mu_{\stoptime+\Delta}^{\no})r_0 = c
\end{equation*}

As $\Delta \rightarrow 0$, $\frac{\mu_{\stoptime}^{\no}-\mu_{\stoptime+ \Delta}^{\no}}{\Delta} \rightarrow -\dot{\mu}_{\stoptime}^{\no}$;  substituting in for this expression and using continuity of beliefs yields the expression for the stopping belief and the stopping payoff: 

\begin{equation*} 
r_0\lambda_{1}^{\rsignal} \mu_{\stoptime}^{\no}(1-\mu_{\stoptime}^{\no}) + \lambda_{1}^{\rsignal} \mu_{\stoptime}^{\no} r_1- \lambda_{1}^{\rsignal} \mu_{\stoptime}^{\no} (1-\mu_{\stoptime}^{\no})r_0 = c.
\end{equation*}

\noindent Algebraic manipulations show this coincides with the expression for the stopping belief from the main text. In particular, note that this stopping belief is independent of $r_{0}$ (as in the main text). From this, it immediately follows that $r_{0}=1$ maximizes effort, as it does not influence the length of time the agent works but may make the agent more willing to initially start exerting effort. 

We now solve for the agent's value function, $V(\mu_{t}^{\no})$, for all agent beliefs $\mu_{t}^{\no} > \mu_{\stoptime}^{\no}$ (assuming the agent works until time $\stoptime$---recalling that beliefs ``drift down''). Writing out the HJB yields: 

\begin{equation*} 
V(\mu_{t}^{\no}) = -c \Delta + \lambda_{1}^{\rsignal} \mu_{t}^{\no} \Delta r_1 + (1- \lambda_{1}^{\rsignal} \mu_{t}^{\no} \Delta) V(\mu_{t+\Delta}).
\end{equation*}

\noindent From this, we obtain the following differential equation: 

\begin{equation*} 
V'(\mu_{t}^{\no}) \lambda_{1}^{\rsignal} \mu_{t}^{\no}(1-\mu^{\no}_{t})= -c+\lambda_{1}^{\rsignal} \mu_{t}^{\no}(r_1 - V(\mu_{t}^{\no})).
\end{equation*}

\noindent Solving this first-order differential equation gives us the following expression for the value function, up to a constant $k$ (which is pinned down by the condition $V(\mu_{\stoptime}^{\no})=(1-\mu_{\stoptime}^{\no})$): 

\begin{equation*} 
V(\mu_{t}^{\no})= k(1-\mu_{t}^{\no}) + \frac{r_1 \lambda_{1}^{\rsignal}-c+c(1-\mu_{t}^{\no})  \log \left( \frac{1-\mu_{t}^{\no}}{\mu_{t}^{\no}} \right)}{\lambda_{1}^{\rsignal}}. 
\end{equation*} 

Note that $V''(\mu_{t}^{\no}) > 0$ for this solution, as well as that $V'(c/(r_1 \lambda_{1}^{\rsignal}))=-1$, so that the value function is everywhere above $(1-\mu_{t}^{\no})$. Thus, the agent would \emph{never} shirk and choose option $(1,0)$ prior to time $\stoptime$; so, as long as the value function is also above $r_1\mu_{t}^{\no}$, the moral hazard constraint does not bind before $\stoptime$. This shows that $r_{1}$ should be set so that the initial moral hazard constraint holds, as discussed in the main text.

\section{Two Periods}
\label{apx:two_period}
In this section, we consider a simple two-period model and show that in the two-period environment, without any assumptions on the information acquisition process, the optimal contract always has a static implementation as a scoring rule. 

\begin{proposition}
In the two-period environments, there exists a scoring rule that implements maximum effort.
\end{proposition}

\begin{proof}
First, note that if in the optimal contract, the agent only exerts effort for one period or does not exert effort at all, the contract for the second period is irrelevant; hence, it is immediate that there exists a scoring rule that implements maximum effort.

Now we focus on the case in which the agent can be incentivized to exert effort in two periods. Without loss of generality, we can assume that the principal offers a menu $\rewardMenu_2 = \{r^s_2\}_{s\in S} \cup \reward^{\no}_2$ at time $2$, 
and a menu $\rewardMenu_1 = \{r^s_1\}_{s\in S} \cup \rewardMenu_2$ at time $1$. 
In this case, consider another menu $\widehat{\rewardMenu}$ where $\widehat{\rewardMenu}_1 = \widehat{\rewardMenu}_2 = \rewardMenu_1$. 
The agent's incentive to exert costly effort weakly increases in the first period as his continuation payoff weakly increases. Thus, the agent exerts effort for at least one period. 
Moreover, observe that if the agent knows that he will choose not to exert effort in the second period, he could pretend to receive a Poisson signal in the first period and obtain a payoff from menu $\rewardMenu_1$ before entering the second period. 
Therefore, by including $\rewardMenu_1$ in $\widehat{\rewardMenu}_2$, the agent's payoff from not exerting effort remains the same. 
On the other hand, the agent has more options to choose from in the second period if he exerts effort. His incentive for exerting effort also weakly increases in the second period. 

Combining all cases, there exists a scoring rule that implements maximum effort.
\end{proof}

Intuitively, the tension in dynamic incentives is that if we provide a higher reward to the agent at certain periods, the agent has a stronger incentive to exert effort in previous periods but a weaker incentive to exert effort in later periods. 
However, in the two-period model, the latter effect disappears since when we adjust the rewards in the second period, there won't be any future periods. 
Therefore, a static implementation, i.e., a scoring rule, implements the maximum effort. 

\end{document}